\documentclass[twoside,11pt,titlepage]{amsart}
\usepackage{amsmath} 
\makeatletter
\newcommand{\addresseshere}{%
  \enddoc@text\let\enddoc@text\relax
}
\makeatother

\newtheorem{innercustomthm}{Theorem}
\newenvironment{customthm}[1]
  {\renewcommand\theinnercustomthm{#1}\innercustomthm}
  {\endinnercustomthm}
  
\newtheorem{innercustomcnj}{Conjecture}
\newenvironment{customcnj}[1]
  {\renewcommand\theinnercustomcnj{#1}\innercustomcnj}
  {\endinnercustomcnj}

\usepackage{amsthm} 
\usepackage{amssymb}	
\usepackage{graphicx} 
\usepackage{multicol} 
\usepackage[marginratio=1:1,height=584pt,width=360pt,tmargin=117pt]{geometry}
\usepackage{hyperref}
\usepackage{pst-node}
\usepackage{tikz-cd}
\usepackage{tikz}
\usepackage{slashed}
\usetikzlibrary{calc}
\usepackage[mathscr]{euscript}
\usepackage[numbers]{natbib}
\usepackage[toc,page]{appendix}

\author{Modjtaba Shokrian Zini & Zhenghan Wang}
\title{Conformal Field Theories as Scaling Limit of Anyonic Chains}

\makeatletter
\newcommand\Author{Modjtaba Shokrian Zini \& Zhenghan Wang}
\let\Title\@title
\def\ps@mystyle{%
      \let\@oddfoot\@empty\let\@evenfoot\@empty
      \def\@evenhead{\makebox[0pt][l]{\thepage}\hfill\Author\hfill}%
      \def\@oddhead{\hfill\Title\hfill\makebox[0pt][l]{\thepage}}%
      \let\@mkboth\markboth}
\makeatother
\makeatletter 
\g@addto@macro{\endabstract}{\@setabstract}
\newcommand{\authorfootnotes}{\renewcommand\thefootnote{\@fnsymbol\c@footnote}}%
\makeatother
\makeatletter
\renewcommand{\maketitle} 
{ \begingroup \vskip 10pt \begin{center} \large {\bf \@title}
	\vskip 10pt \large \@author \hskip 20pt \@date \end{center}
  \vskip 10pt \endgroup \setcounter{footnote}{0} }
\makeatother 

\newcommand{\abs}[1]{\left| #1 \right|} 
 
\newcommand{\ket}[1]{\left| #1 \right>} 
\let\baraccent=\= 
\renewcommand{\=}[1]{\stackrel{#1}{=}} 

\newtheorem{thm}{Theorem}[section]
\newtheorem{lem}[thm]{Lemma}

\newtheorem{cor}[thm]{Corollary}

\newtheorem{cnj}[thm]{Conjecture}
\theoremstyle{definition}
\newtheorem{dfn}{Definition}
\theoremstyle{remark}
\newtheorem{rmk}{Remark}
\newtheorem{fct}{Fact}

\usepackage{url,times,wasysym,geometry,indentfirst,rotating,tikz}
\DeclareRobustCommand{\rchi}{{\mathpalette\irchi\relax}}
\newcommand{\irchi}[2]{\raisebox{\depth}{$#1\chi$}}
\title{\LARGE{\bf
\textsc{Conformal Field Theories as Scaling Limit of Anyonic Chains
}}}

\begin{document}
\begin{center}
  \LARGE 
   \maketitle \par \bigskip

  \normalsize
  \authorfootnotes
  Modjtaba Shokrian Zini \footnote{shokrian@math.ucsb.edu}\textsuperscript{,\hyperref[1]{1}}, Zhenghan Wang \footnote{zhenghwa@microsoft.com;
  zhenghwa@math.ucsb.edu}\textsuperscript{,\hyperref[2]{2}} \par 
  \bigskip
\end{center}

\begin{abstract}

We provide a mathematical definition of a low energy scaling limit of a sequence of general non-relativistic quantum theories in any dimension, and apply our formalism to anyonic chains.  We formulate \hyperref[cnj4.3]{\textbf{Conjecture 4.3}} on conditions when a chiral unitary rational (1+1)-conformal field theory would arise as such a limit and verify the conjecture for the Ising minimal model $M(4,3)$ using Ising anyonic chains.  Part of the conjecture is a precise relation between Temperley-Lieb generators $\{e_i\}$ and some finite stage operators of the Virasoro generators $\{L_m+L_{-m}\}$ and $\{i(L_m-L_{-m})\}$ for unitary minimal models $M(k+2,k+1)$ in \hyperref[cnj5.5]{\textbf{Conjecture 5.5}}.  A similar earlier relation is known as the Koo-Saleur formula in the physics literature \cite{koo1994representations}. Assuming \hyperref[cnj4.3]{\textbf{Conjecture 4.3}}, most of our main results for the Ising minimal model $M(4,3)$ hold for unitary minimal models $M(k+2,k+1), k\geq 3$ as well. Our approach is inspired by an eventual application to an efficient simulation of conformal field theories by quantum computers, and supported by extensive numerical simulation and physical proofs in the physics literature.

\end{abstract}

\tableofcontents

\section{Introduction}\label{0}

Quantum field theory (QFT) is arguably the best experimentally tested model of Nature, yet its mathematically rigorous formulation is far from clear as exemplified by the Yang-Mills existence and mass gap millennium problem (recently Seiberg articulated that QFT is not even mature physically \cite{seibergQFT}).  Besides the intrinsic beauty of a mathematical formulation, a mathematical definition will provide the missing foundation for proving the conjecture that all theoretically physical QFTs can be efficiently simulated by quantum computers.
The circuit model of quantum computing is based on quantum mechanics, but it was stated explicitly as a conjecture in \cite{tqc03} that QFTs would not provide extra computational power beyond quantum mechanics as suggested by the efficient simulation of (2+1)-topological quantum field theories (TQFTs) \cite{simulation02}.  Quantum estimate of scattering probabilities in massive scalar quantum field theories also supports such an extended quantum polynomial Church thesis \cite{jordan12} if the convergence of lattice models to the continuum is addressed mathematically.  Our approach to a potential quantum simulation of rational (1+1)-conformal field theories (CFTs) would follow closely the efficient quantum simulation of TQFTs in \cite{simulation02} using the functorial definitions of TQFTs and CFTs.  An important difference between TQFTs and CFTs is that, while TQFTs being realized as gapped quantum systems, CFTs represent universality classes of gapless critical phases.  Our program seems to be a first attempt towards a quantum simulation of gapless QFTs mathematically.

While the efficient simulation of CFTs is not addressed in detail in this paper, the approach to CFTs here is inspired by this eventual application as outlined in the last section.  How to represent CFTs on quantum computer is already a challenging problem.  The main issues that we are addressing in this paper are the algorithmic convergence of finite lattice theories to the continuum limit with an emphasis on the convergence of observables as an algebra, and hidden localities of CFTs with respect to both space and energy as required by a quantum simulation.  Another closely related motivation is to study CFTs based on the significant new insights of TQFTs from their applications to topological phases of matter. From this angle, our paper is a second possible answer to the question how to recover a chiral CFT or vertex operator algebra (VOA) from a TQFT or modular tensor category (MTC) (another one in \cite{tener16}).  Our anyonic chain approach to chiral CFTs should encompass all other formulations such as VOAs and local conformal nets (LCNs) as anyonic chains (ACs) are finite versions of microscopic models of CFTs, and we succeed in proving that the algebra of local observables in VOAs and LCNs for unitary minimal models (assuming \hyperref[cnj4.3]{\textbf{Conjecture 4.3}} beyond Ising) can be recovered from our formalism.

CFTs and TQFTs are exemplars of QFTs that both have rigorous mathematical formulations and important mathematical and physical applications such as in the monster Moonshine conjecture, critical statistical mechanics models and topological phases of matter, respectively.  CFTs and TQFTs are closely related to each other first by the bulk-edge correspondence as in the fractional quantum Hall effect, and secondly their algebraic data, VOAs and MTCs respectively, are conjectured to be Tannaka-Krein dual to each other (see \cite{tener16} and the references therein).  The bulk-edge correspondence and Tannaka-Krein duality suggest the possibility that VOAs can be reconstructed from MTCs as a generalization of the reconstruction of compact groups from symmetric fusion categories, though the correspondence is far from one-to-one. 

Locality is a salient feature for any physical QFT.  Since TQFTs are low energy effective theories, their locality is not intrinsic and usually hidden.  For example, in the Witten-Chern-Simons (WCS) modeling of the fractional quantum Hall liquids, the WCS theory is an effective description for the emergent anyons, it follows that locality of WCS TQFTs should be derived from that of the underlying electron systems.  The simulation of TQFTs in \cite{simulation02} uses a hidden locality given by pairs of pants decomposition of the space surfaces.  Similarly, there are no intrinsic infinite degrees of freedom for a CFT to define locality.  

Anyons are modelled by simple objects in unitary MTCs.  
ACs are the anyonic analogues of quantum Heisenberg spin chains investigated purely as an academic curiosity \cite{feiguin2007interacting}.  ACs' conceptual origin can be traced back at least to Jones' Baxterization of braid group representations and his idea of generalized spin chains regarding "spins" as something each with a large algebra of observables at sites and being tensored together with generalized tensor products such as Connes fusion (see section 4 of \cite{jones03}.)   In the scaling limit, ACs are exactly solvable but not known to be rigorously solvable mathematically \cite{feiguin2007interacting,saleur90}.  We reverse the logic in this paper to regard ACs as localization of CFTs, thus provide a space locality for VOAs.  Our philosophy, as inspired by algorithmic discrete mathematics, is that instead of using ACs to approximate VOAs, VOAs serve as good approximations of sufficiently large finite ACs in their low energy spectrum.

Our limit of a sequence of quantum theories $\{(\mathcal{W}_n,H_n)\}$ will be dictated by both space and energy localities.  The Hilbert space $\mathcal{W}$ of a quantum theory has two important bases: the basis encoding the spacial locality, and the basis of energy eigenstates of $H$. The two bases will be referred to as space basis and energy basis, respectively.  Operators can be local with respect to one of the two bases, but there is a tension of locality with respect to both bases.  To define a limit of the sequence of quantum theories $\{(\mathcal{W}_n,H_n)\}$, embedding the Hilbert space $\mathcal{W}_n$ into $\mathcal{W}_{n+1}$ is the first step.  Which locality of space and energy is preserved by the embedding leads to different notions of limit.  We will construct the scaling limit of a sequence of quantum theories $\{(\mathcal{W}_n,H_n)\}$ from their low energy behaviors when the lattice sizes go to zero, therefore we preserve energy locality.  Preservation of space locality will lead to the thermodynamic limit. 

Besides the Hilbert space and Hamiltonian, another essential feature of any quantum theory is the algebra of observables. Since our quantum theories are non-relativistic, time needs to be addressed separately.  The algebra structure of observables encodes consecutive measurements as multiplication, hence somewhat reflects time in the limit.  As noted above, our formulation of scaling limit will have \textbf{everything} that can be computed using some limit of physical objects. Compared to other well-established formulations of chiral CFTs such as VOAs following Wightman's axioms,  and LCNs, our scaling limit results in a much bigger set of observables. In fact, we will show, in the case of Ising anyonic chain, the resulting observables contain a subset corresponding to smeared fields (or Wightman's) observables $\phi(f)$, a subset corresponding to bounded observables of LCN and a subset corresponding to observables in the VOA $M(4,3)$. We conjecture the same holds for all unitary minimal models $M(k+2,k+1)$ for $k\geq 3$.

An important desideratum of our scaling limit is finitely complete and accessible in the sense that any sequence that should have a limit indeed has one in the scaling limit, and anything in the scaling limit is a limit of some sequence. So the theory in the limit should be completely describable by the sequence of finite theories and there should be no \textit{extra} object that is not some limit of finite objects.  Our scaling limits VOA $\mathcal{V}=\oplus_{n=0}^{\infty} \mathcal{V}_n$ should be regarded as computable using the AC approximations.  Philosophically, such VOAs $\mathcal{V}=\oplus_{n=0}^{\infty} \mathcal{V}_n$ from ACs categorify computable integral sequences such that each vector space $\mathcal{V}_i$ serves as a categorification of the integer $\textrm{dim} \mathcal{V}_i$. 

\section{Preliminaries and outline of main results}\label{1.}

First we recall the notion of a VOA, which we regard as the mathematical definition of a \textit{chiral} CFT ($\chi$CFT), along with Wightman's observables, local conformal nets (LCNs), and finally, anyonic chains (ACs). Then we outline our results.  While the results on Ising ACs are interesting, our most important contribution of the paper is a framework for addressing the reconstruction of CFTs from MTCs, and the potential simulation of CFTs by quantum computers. Note that full CFTs can be constructed from a nice $\chi$CFT with a choice of an indecomposible module category over its representation category \cite{frs06}.  In the following, CFTs will mean $\chi$CFTs, but there are a few cases where a CFT can be interpreted either as a chiral or full one.

A VOA is mathematical axiomatization of the chiral algebra of a CFT.  The vertex operator $Y(a,z)$ implements the state-operator correspondence of CFTs.  They are the field operators which insert the state $a$ at a space-time point $z=0$ with a small neighborhood locally parameterized by $z$.  VOA is our preferred framework for our discussion on CFTs. Other frameworks for CFTs such as LCNs will also be discussed at times. 

The notations and definitions for VOA follow closely those of \cite{konig2017matrix}.

\subsection{Vertex operator algebra}\label{1.1}

\hfill \break
Let $\mathbb{N}_0$ be the set of non-negative integers.  Consider an $\mathbb{N}_0$-graded $\mathbb{C}$-vector space $\mathcal{V}=\oplus_{n=0}^{\infty} \mathcal{V}_n$, where the \textit{weight} spaces $\mathcal{V}_n$ satisfy $\dim \mathcal{V}_n<\infty$, equipped with a linear map called the \textit{vertex operator}, 
$$Y(\cdot,z):\mathcal{V} \to \text{End}(\mathcal{V})[[z,z^{-1}]], \ \ Y(v,z)=\sum_{n\in \mathbb{Z}} v_nz^{-n-1},$$
where  $v_n \in \text{End}(\mathcal{V)}$ are called the \textit{mode} operators of $v$. The mode operators satisfy
$$v_nu=0, \text{ for all } v,u \in \mathcal{V} \text{ and } n \text{ sufficiently large.}$$
As a different notation, which will be motivated later, for a \textit{homogeneous} vector $v$ in some weight space with weight $\text{wt }v$, we can shift the index to obtain  
$$Y(v,z)=\sum_{n\in \mathbb{Z}} \textbf{y}(v)_nz^{-n-\text{wt }v},$$ 
where $\textbf{y}(v)_n=v_{n+ \text{wt }v-1}$.

Further, there are two distinguished vectors, the \textit{vacuum} $1 \in \mathcal{V}_0$ and the \textit{conformal} or \textit{Virasoro} vector $\omega \in \mathcal{V}_2$. 

The vacuum vector satisfies $Y(1,z)=\text{id}_\mathcal{V}$ and the \textit{creation} property holds
$$Y(v,z)1=v+\ldots \in \mathcal{V}[[z]],$$
giving the \textit{operator-state} or \textit{field-state correspondence} $\lim_{z \to 0} Y(v,z)1=v$ when we replace the indeterminate $z$ with a complex number. That is why we may sometimes use the expression ``conformal field'' which is the field associated to the conformal vector $\omega$ and we may also sometimes use the word field while actually meaning the vector (this will be clear from the context). 

On the other hand, the Virasoro vector $\omega$ gives us the modes and field
$$\omega_{n+1}=\textbf{y}(\omega)_n=L_n, \ \ Y(\omega,z)=\sum_{n \in \mathbb{Z}} L_nz^{-n-2},$$
where the $L_n$s generate the \textit{Virasoro (lie) algebra} with relations
$$[L_n,L_m]=(n-m)L_{n+m}+\frac{c}{12}n(n^2-1)\delta_{n+m,0}\cdot \text{id}_\mathcal{V}, \ \forall m,n \in \mathbb{Z},$$
where the constant $c$ is called the central charge (also called rank $\mathcal{V}$). The grading of $\mathcal{V}$ is the spectral decomposition of $L_0$, so $L_0v=n v$ for any homogeneous $v \in \mathcal{V}_n$. A homogeneous vector $v$ is \textit{quasi-primary} if $L_1v=0$ and it is \textit{primary} if $L_nv=0, \forall n>0$. We also have the \textit{translation} property
$$\frac{\text{d}}{\text{d}z}Y(v,z)=Y(L_{-1}v,z),$$
where the left side is the formal derivative of a Laurent series. 

Finally, for all $a,b \in \mathcal{V}$, there exists $k \in \mathbb{N}_0$ such that
$$(z_1-z_2)^k[Y(a,z_1),Y(b,z_2)]=0 \ \ \text{(\textit{locality} condition)}.$$
Evidently, we are defining products of vertex operators using formal series. This finishes the description of vertex operator algebra.
\begin{dfn} \label{dfn1}
The tuple $(\mathcal{V},Y,1,\omega)$ with the above properties is called a \textit{vertex operator algebra} (VOA). 
\end{dfn}

There are some immediate implications of the above axioms. For any two homogeneous vectors $u,v \in \mathcal{V}$,
$$\text{wt }(\textbf{y}(v)_nu)=\text{wt u}-n.$$

The locality axiom implies (the \textit{Jacobi} or) the \textit{Borcherds identity}, 
which can be formulated as 
\begin{gather*}
\scalebox{0.9}{$\text{Res}_{z_1-z_2}(Y(Y(a,z_1-z_2)b,z_2)(z_1-z_2)^p \iota_{z_2,z_1-z_2}(z_2+(z_1-z_2))^q)=$}
\end{gather*}
\begin{gather*}
\scalebox{0.9}{$\text{Res}_{z_1}(Y(a,z_1)Y(b,z_2)\iota_{z_1,z_2}(z_1-z_2)^pz_1^q)-\text{Res}_{z_1}(Y(b,z_2)Y(a,z_1)\iota_{z_2,z_1}(z_1-z_2)^pz_1^q),$}
\end{gather*}
for all $p,q\in \mathbb{Z}$. In the above expression, $\text{Res}_zf(z)$ is the coefficient of $z^{-1}$ in $f(z)$. $\iota_{z_1,z_2}f(z_1,z_2)$ is the series expansion of $f(z_1,z_2)$ in the domain $|z_1|>|z_2|$. As an equivalent formulation:
$$\sum_{j=0}^\infty \binom{p}{j} (a_{q+j}b)_{p+k-j}c=\sum_{j=0}^{\infty}(-1)^j\binom{q}{j}a_{p+q-j}b_{k+j}c$$
$$-\sum_{j=0}^{\infty}(-1)^{j+q} \binom{q}{j}b_{q+k-j}a_{p+j}c, \ \ \ a,b,c \in \mathcal{V}, p,q,k \in \mathbb{Z}.$$
 
The next objects to discuss are the \textit{admissible modules} of a VOA. A module has a structure similar to that of a VOA and some compatibility properties with the VOA. We will focus on \textit{irreducible} modules.

An \textit{irreducible admissible} module $(A,Y_A)$ for a VOA $(\mathcal{V},Y,1,\omega)$, is an $\mathbb{N}_0$-graded vector space $A$ with a linear map
$$Y_A(\cdot,z): \mathcal{V} \to \text{End}(A)[[z,z^{-1}]], \ \ Y_A(v,z)=\sum_{n \in \mathbb{Z}}v_n^Az^{-n-1},$$
where $v_n^A$ are the mode operators of $v$ and $\omega_{1}^A=L_{A,0}$ gives the \textit{weights} in the module, which are eigenvalues of $L_{A,0}$, with the following difference 
$$\forall a \in A_n, \ \ L_{A,0}a=(\alpha+n)a, \  \text{for some unique }\alpha \text{ depending on } A.$$
The unique \textit{conformal} or \textit{highest} weight $\alpha$ gives the grading 
$$A=\bigoplus_{n \in \mathbb{N}_0} A_{n}.$$
$A_0$ is called the \textit{top-level} and $A_n$ the $n-$th level of module $A$.  Lastly, there is an analogous notation of $\textbf{y}(v)_n$ for a homogeneous vector $v \in \mathcal{V}$,
$$Y_A(v,z)=\sum_{n\in \mathbb{Z}} \textbf{y}_A(v)_nz^{-n-\text{wt }v},$$ 
where $\textbf{y}_A(v)_n=v_{n+ \text{wt }v-1}$ and for any two homogeneous vectors $u,v \in \mathcal{V}$,  
$$\text{wt }(\textbf{y}_A(v)_nu)=\text{wt u}-n.$$

The vertex operator and the modes $L_{A,n}$ of $Y_A(\omega,z)$ satisfy all the axioms of a VOA ($A$ should be seen as a \textit{representation} of a VOA), except the creativity property. Locality holds and more importantly for us, Borcherds identity also holds in this case with the obvious necessary change $c \in A$. The sub(/super)script $A$ will be dropped from the operators involved as it will be clear from the context.

In this paper, the weight spaces need to be finite dimensional. Moreover, the condition $C_2-$\textit{co-finiteness} is imposed on the VOAs. This means the space $C_2=\text{span} \{u_{-2}v| \ u,v \in \mathcal{V}\}$ has finite co-dimension $C_\mathcal{V}=\dim \mathcal{V}/C_{2}$. Then, a result \cite[(27)]{konig2017matrix} on the growth of the dimension of the weight spaces of an irreducible module $A$ follows
$$\dim A_n \le (\dim A_0) \cdot e^{2\pi\sqrt{\frac{C_\mathcal{V}n}{6}}},$$
where $C_\mathcal{V}=\dim \mathcal{V}/C_2$. This at most exponential growth is necessary if an approach to simulation requires a truncation of energy up to some $N$, where one can not afford more than polynomially many qubits to be used to simulate the vector space.

Finally, the character for a module $A$ is defined as
$$\text{char}(A)=\text{Tr}_A(q^{L_0-\frac{c}{24}})=\sum_{n \in \mathbb{N}_0} \dim(A_{n})\  q^{n+h-\frac{c}{24}}.$$
An important class of VOAs consists of the unitary minimal models (UMMs) introduced in the next section. A UMM $\mathcal{V}$ satisfy many properties such as being \textit{CFT-type}, i.e. $V_0=\mathbb{C}1$, or in other words only the vacuum has energy zero. Also, $\mathcal{V}$ is \textit{rational}, i.e. every admissible $\mathcal{V}$-module is a direct sum of irreducible $\mathcal{V}$-module. Last but not least, $\mathcal{V}$ is unitary (as explained below). As a convention, the expression \textit{CFT} or \textit{chiral CFT} or \textit{full CFT}, will be referring to a VOA with the described properties.

For us, a unitary VOA has some a positive definite hermitian form $(\cdot,\cdot)_\mathcal{V}:\mathcal{V} \times \mathcal{V} \to \mathbb{C}$ with respect to which one can define adjoint of mode operators. Specifically, $L_n^\dagger=L_{-n}$. One can similarly define \textit{unitary modules}: a positive definite $(\cdot,\cdot)_A: A\times A \to \mathbb{C}$ with respect to which $L_n^\dagger=L_{-n}$. Using the hermitian form, one can define a norm in the obvious way and get the completion of a graded unitary module (which includes $\mathcal{V}$ itself). They will be represented by $\overline{A}$. For a complete definition of a unitary VOA, see \cite{konig2017matrix}.

Next, to describe a \textit{full} CFT, we will focus only on full diagonal CFTs. The description here will not be completely elaborate as full CFTs are not discussed much in our work. Only the essential concepts will be mentioned in simple terms. Consider a chiral CFT with the restrictions imposed earlier. The idea is to take any irreducible module coupled with the contragredient module (which, assuming e.g. unitarity, is isomorphic to the module itself) and consider the Hilbert space it gives after completion
$$\mathcal{H}=\bigoplus_{\text{irreducible modules}} \mathcal{A}_i\otimes\mathcal{A}_i',$$
where $\mathcal{A}_i'$ is the contragredient module of $\mathcal{A}_i$. The contragredient module $\mathcal{V}'$ is defined as the linear functionals that vanish except on finitely many of the weight spaces, in other words
$$\mathcal{V}'=\bigoplus_{n \in \mathbb{N}_0} \mathcal{V}_n',$$
and it can be given a $\mathcal{V}$-module structure. In the case of minimal modules, this means another isomorphic (as $\mathcal{V}$-module) copy of the module itself. 

The vacuum vector $1=1_L\otimes1_R$ and the conformal vector $\omega=\omega_L\otimes1_R+1_L\otimes\omega_R$ are defined in the obvious way using the left (and right) vacuum $1_L$($1_R$) and left (and right) conformal vector $\omega_L$($\omega_R$). The Virasoro mode operators are defined accordingly as
$$\mathbb{L}_n=L_n+\overline{L}_n,$$
where the first term is the $n$-th Virasoro mode for the chiral copy and the second, for the antichiral copy. Primary fields are accordingly defined as those $a \in \mathcal{H}$ that satisfy 
$$\mathbb{L}_n a=0, \  \forall n>0.$$
Introducing the analog of vertex operator $Y(\cdot,\cdot)$ for the full CFT requires us to describe what intertwiners are, but we will only need to introduce the conformal field, which is
$$\mathbb{Y}(\omega,(z,\overline{z}))=\sum_{n \in \mathbb{Z}} L_nz^{-n-2}+\overline{L}_n\overline{z}^{-n-2}.$$

\subsection{Unitary minimal models and Ising CFT}\label{1.2}

\hfill \break 
A special class of VOAs are the highest weight representations of the Virasoro algebra with central charge $c<1$ that are unitary. These highest weight representations can be completely characterized by their central charge,  which form a \textit{discrete series} $c=1-\frac{6}{(k+1)(k+2)}$ for $k\geq 2$.   

These VOAs $M(k+2,k+1)$, called the \textit{unitary minimal models} (UMMs), can be constructed as cosets $\frac{SU(2)_k\times SU(2)_1}{SU(2)_{k+1}}$ and have central charge $c=1-\frac{6}{(k+1)(k+2)}$ for $k\ge 2$. We will refer to $M(k+2,k+1)$ as the UMM at level $k$ (of $SU(2)_k$). They have finitely many irreducible modules determined by their conformal weights
$$h_{r,s}=\frac{((k+1)r-(k+2)s)^2-1}{4(k+1)(k+2)}, \ \ 1\le r\le k+1, 1\le s \le k.$$
Hence, due to the symmetry $h_{k+2-r,k+1-s}=h_{r,s}$, there are $\frac{k(k+1)}{2}$ many irreducible modules. In particular, the \textit{Ising CFT} has central charge $\frac{1}{2}$ corresponding to level $k=2$ (see \cite{francesco2012conformal} for more on minimal models, and \cite{francesco2012conformal} and particularly the notes \cite{ginsparg1988applied} for the Ising CFT).

The chiral Ising CFT has $3$ irreducible modules with conformal weights $h_{1,1}=0$ (the VOA $\rchi_0$ itself with the vacuum field $1$), $h_{2,1}=\frac{1}{2}$ (the module $\rchi_\frac{1}{2}$ corresponding to the free Fermionic field $\psi$), $h_{3,1}=\frac{1}{16}$ (the module $\rchi_\frac{1}{16}$ corresponding to the spin field $\sigma$). 

In the Ising CFT, the \textit{fusion} rules are as follows:
$$\rchi_\frac{1}{2} \otimes \rchi_\frac{1}{2} = \rchi_0$$
$$\rchi_\frac{1}{16} \otimes \rchi_\frac{1}{16} = \rchi_\frac{1}{2} \oplus \rchi_0$$
$$\rchi_\frac{1}{2} \otimes \rchi_\frac{1}{16} = \rchi_\frac{1}{16} $$

and of course anything fused with $\rchi_0$ becomes itself. 

The Fermionic algebra is used to generate the Hilbert spaces $\rchi_i$. The Hilbert spaces $\rchi_0,\rchi_\frac{1}{2}$ are generated by the Fermionic modes $\{\Psi_{n-\frac{1}{2}}\}_{n \in \mathbb{Z}}$ satisfying the anticommutative canonical relations (ACR)
$$\{\Psi_k,\Psi_{k'}\}=\delta_{k+k',0},$$
and their actions satisfy the conjugacy relation
$$\Psi_k=\Psi_{-k}^\dagger.$$
The third Hilbert space $\rchi_\frac{1}{16}$ is generated by $\{\Psi_{n}\}_{n\in \mathbb{Z}}$ which are also another version of the Fermionic algebra where, this time, the modes are indexed by integers and they satisfy the same properties:
$$\{\Psi_k,\Psi_{k'}\}=\delta_{k+k',0} \ \ , \ \ \Psi_k=\Psi_{-k}^\dagger.$$
The first algebra generates $\rchi_0$ and $\rchi_\frac{1}{2}$ by acting on the vacuum $1$. Indeed, the vectors 
$$\{\Psi_{-k_r}\ldots\Psi_{-k_1}1|\ k_1<\ldots<k_r, \ \ k_i \in \mathbb{N}-\frac{1}{2}\},$$
with corresponding weight $\sum k_i$, give an orthonormal basis for $\rchi_0 \oplus \rchi_\frac{1}{2}$, hence giving the character
$$q^{-\frac{c}{24}}\prod_{n=1}^\infty (1+q^{n-\frac{1}{2}}).$$
As a matter of convenience, the factor $q^{-\frac{c}{24}}=q^{-\frac{1}{48}}$ will sometimes get dropped. Obviously, the part of which has powers of $q$ in $\mathbb{N}-\frac{1}{2}$ corresponds to $\rchi_\frac{1}{2}$ and the rest with powers of $q$ in $\mathbb{N}_0$ corresponds to $\rchi_0$. 
The second algebra $\{\Psi_{n}\}_{n\in \mathbb{Z}}$ generates $\rchi_\frac{1}{16}$ in a similar way: the orthonormal basis
$$\{\Psi_{-k_r}\ldots\Psi_{-k_1}\ket{\frac{1}{16}}|\ 0<k_1<\ldots<k_r, \ \ k_i \in \mathbb{N}\},$$
where $\ket{\frac{1}{16}}$ is the highest weight vector, or the vector at the top level satisfying $L_0\ket{\frac{1}{16}}=\frac{1}{16}\ket{\frac{1}{16}}$. The corresponding weight is naturally $\sum k_i$. Notice that  $\ket{\frac{1}{16}}$ is sent to a scalar multiple of itself by $\Psi_0$. The character is
$$\text{char}(\rchi_\frac{1}{16})=q^{\frac{1}{16}-\frac{1}{48}}\prod_{n=0}^\infty(1+q^n).$$
Although not mentioned, but from the above description, it is clear what the hermitian form should be. The formulae for $L_n$s are well-known \cite{ginsparg1988applied} and will be derived in the \hyperref[A.1]{appendix}.

As a final note, the Ising full CFT is
$$\mathcal{H}=\rchi_0\overline{\rchi}_0 +\rchi_\frac{1}{2}\overline{\rchi}_\frac{1}{2} +\rchi_\frac{1}{16}\overline{\rchi}_\frac{1}{16},$$
with the corresponding operators $\mathbb{L}_n$ which will be derived using the formulae for $L_n$s.

\subsection{Wightman's observables and local conformal nets}\label{1.3}

\hfill \break
In addition to VOA, we will also work with observables coming from LCNs and Wightman's axioms. One of the objectives of this work is to obtain the fields in the scaling limit and prove that products of fields are also in the scaling limit, hence obtaining a ``scaling limit of algebras''. Only the \textbf{observables} (or fields) in each framework will be defined. As we shall see, observables are related to the fields $Y(a,z)$ we have been using so far. For this section, the definitions and facts follow those of \cite{carpi2015vertex}.

So far, the observables or fields that are \textit{point-like} have been described; the insertion of the field is exactly at a point. Other types of observables that can be derived formally from these are called \textit{smeared} field operators or \textit{Wightman's} observables. Taking some function $f \in C^\infty(S^1)$, define formally
$$Y(a,f):=\oint Y(a,z)f(z)z^{\text{wt }a} \ \frac{\text{d}z}{2\pi i z}=\sum_{n \in \mathbb{Z}} \hat{f}_n\textbf{y}(a)_n,$$
where $\hat{f}_n$s are the Fourier coefficients of $f$. As $f$ is smooth, it is known that its Fourier coefficients will be rapidly decreasing:
$$\forall k , \exists N_k \text{  such that   } \forall |n| \ge N_k \implies |\hat{f}_n| \le \frac{1}{n^k}.$$

In order to have truly a linear operator defined on $\mathcal{V}$ (before taking its completion), an \textit{energy bound} on the mode operators is needed
$$||\textbf{y}(a)_nb|| \le C_a(|n|+1)^{r_a}||(L_0+\textbf{1})^{s_a}b||, \  \forall b \in \mathcal{V}$$
where the constants $C_a,r_a,s_a>0$ dependent on $a$, and the norm is given by the unitary structure. If the above inequality holds, then one easily observes that 
$$||Y(a,f)b||\le C_a||f||_{r_a}||(L_0+\textbf{1})^{s_a}b||,$$
where the $r_a$\textit{-norm} of $f$ is defined as
$$||f||_{r_a}=\sum_n |\hat{f}_n|(|n|+1)^{r_a}.$$
As it will be observed in \hyperref[4]{section 4}, the \textit{correlation function} $F^0$ using the smeared formalism can be defined as
\begin{gather*}
\scalebox{0.95}{$
F^0((a_1,f^{(1)}),\ldots,(a_k,f^{(k)}),u,v):=(u,Y(a_1,f^{(1)})Y(a_k,f^{(k)})v), \ \ f^{(i)} \in C^\infty(S^1)$}
\end{gather*}
where the fields $a_i$ satisfy an energy bound and $u,v \in \mathcal{V}$. The correlation function for smeared fields on a full CFT is defined similarly, but we will defer that to \hyperref[4]{section 4} since this will be used in a very restricted case.

For the UMMs, all $\textbf{y}(a)_n$ are energy bounded. The most important field for us is the conformal smeared field $Y(\omega,f)$ which is denoted by 
$$L(f)=\sum_{n \in \mathbb{Z}} \hat{f}_nL_n.$$

For all UMMs, it can be shown that $L_n$ satisfies the energy bound
$$||L_nb|| \le \sqrt{\frac{c}{2}}(|n|+1)^\frac{3}{2}||(L_0+\textbf{1})b||,$$
giving as a result
$$||L(f)b||\le \sqrt{\frac{c}{2}} ||f||_{\frac{3}{2}}||(L_0+\textbf{1})b||.$$

The next observables are the ones coming from the LCN picture of CFT. Only the most relevant features of this framework shall be discussed. Denote by $\mathcal{I}$ the family of proper intervals
of $S^1$. A net $\mathcal{A}$ of von Neumann algebras on $S^1$ is a map that associates a von Neumann algebra $\mathcal{A}(I) \subset \mathcal{B}(\mathcal{H})$ for some fixed Hilbert space $\mathcal{H}$. These nets should be \textit{local} in the sense that for $I_1,I_2 \in \mathcal{I}$ with $I_1 \cap I_2=\emptyset$,
$$[\mathcal{A}(I_1),\mathcal{A}(I_2)]=\{0\},$$
and they satisfy \textit{isotony}, 
$$I_1 \subset I_2 \implies \mathcal{A}(I_1) \subset \mathcal{A}(I_2).$$
In the case of UMMs, and more generally \textit{unitary Virasoro VOAs}, taking the Hilbert space to be any irreducible module of conformal weight $h$,
$$\mathcal{A}(I)=\{e^{iL(f)}| \ f \in C^\infty(S^1), \text{supp}(f) \subset I\}'' \ \ (\text{see \cite{kawahigashi2004classification}}).$$
So $\mathcal{A}(I)$ is the double-commutant of the algebra generated by the unitary operators $e^{iL(f)}$ associated to functions with support inside $I$. The \textit{double-commutant} theorem implies that the strong (or weak limit) of the algebra generated by $e^{iL(f)}$s is also $\mathcal{A}(I)$, a fact that will be used later in \hyperref[4.2]{section 4.2}.

\subsection{Anyonic chains}\label{1.4}

\hfill \break
Though ACs are closely related to and inspired by spin chains, there are some fundamental differences between them.  The most salient difference touches on the trade-off between explicit locality and unitarity in QFTs.  Spin chains implement locality explicitly by attaching local state spaces to each site, while the Hilbert spaces of ACs do not have such explicit tensor product decomposition.  In general, it is harder to obtain unitary interacting exactly solvable spin chains with CFT scaling limits, while such examples of ACs are ubiquitous \cite{feiguin2007interacting}.  This phenomenon is related to the localization of braid group representations, where finite order unitary $R$-matrices are very rare \cite{RW12}.

This section follows the exposition of anyonic chains (ACs) in \cite{gils2013anyonic,feiguin2007interacting}. An AC is a periodic or open (with boundary condition) chain, along which pairwise interactions occur between quasi-particles (the anyons). e.g. the generalized spin $j$ anyons of $SU(2)_k$. The chain is usually presented along a straight path if it is non periodic and as a loop if it is periodic. We will also put the nonperiodic chain along the upper half-circle (\hyperref[fig1]{Figure 1}) as this picture will be used in \hyperref[4.2]{section 4.2} to relate the AC to LCN.
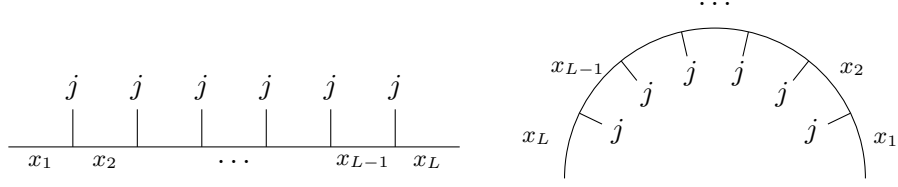
\begin{figure}\label{fig1}
\centering
\begin{tikzpicture}[scale=1]
\draw (0,0) -- (6,0);
\foreach \i in {1,...,6}
{
\draw ({\i*6/7},{0}) -- ({\i*6/7},{0.5}) node[anchor=south] {$j$};
}
\node at ({3/7},{-0.2}) {\Small$x_1$};
\node at ({9/7},{-0.2}) {\Small$x_2$};
\node at ({3},{-0.2}) {$\ldots$};
\node at ({6-9/7},{-0.2}) {\Small$x_{L-1}$};
\node at ({6-3/7},{-0.2}) {\Small$x_{L}$};
\end{tikzpicture}
\hspace{0.5cm}
\begin{tikzpicture}[scale=1]
\draw (2,0) arc (0:180:2);
\foreach \i in {1,...,6}
{
\draw ({2*cos(\i*180/7)},{2*sin(\i*180/7)}) -- ({2*cos(\i*180/7)*5/6},{2*sin(\i*180/7)*5/6});
}
\foreach \i in {1,...,6}
{
\node at ({2*cos(\i*180/7)*4.3/6},{2*sin(\i*180/7)*4.3/6}) {$j$};
}
\node at ({2*cos(180/14)*7/6},{2*sin(180/14)*7/6}) {\Small$x_1$};
\node at ({2*cos(3*180/14)*7/6},{2*sin(3*180/14)*7/6}) {\Small$x_2$};
\node at ({2*cos(90)*7/6},{2*sin(90)*7/6}) {$\ldots$};
\node at ({-2*cos(3*180/14)*7/6},{2*sin(3*180/14)*7/6}) {\Small$x_{L-1}$};
\node at ({-2*cos(180/14)*7.3/6},{2*sin(180/14)*7.3/6}) {\Small $x_{L}$};
\end{tikzpicture}
\caption{Anyonic chain on a straight path and on a half-circle}
\end{figure}
The channel between each two anyons provides the means for fusion. A \textit{boundary condition} $(a,b)$ means $x_1=a$ and $x_L=b$. Each \textit{admissible} fusion path has to satisfy the fusion rules of $SU(2)_k$:
$$j_1 \otimes j_2=|j_1-j_2| \oplus (|j_1-j_2|+1) \oplus \ldots \oplus \min\{j_1+j_2,k-j_1-j_2\}.$$
All admissible fusion paths form an orthogonal basis of the Hilbert space $\text{Hom}(x_1\otimes x_L,j^{\otimes (L-1)})$ where the inner-product comes from the diagram calculus of the unitary MTC $SU(2)_k$. An important observation is that generally for all $j$ and $k$, the resulting Hilbert spaces do not have a tensorial structure, though the case of Ising AC does have one.
\begin{figure}\label{fig2}
\centering
\begin{tikzpicture}[scale=0.75]
\draw (0,0) -- (0,{-0.5*sqrt(2)}) node[anchor= north] {\small$x_{i+1}$};
\draw (0,0) -- ({-1},{1}) node[anchor= south] {\small$x_{i-1}$};
\draw (-0.5,0.5) -- (0,1) node[anchor= south] {\small$j$};
\draw (0,0) -- ({1},{1}) node[anchor= south] {\small $j$};
\node at ({-0.25-0.2},{0.25-0.2}) {\small $x_{i}$};
\node at ({3},{0}) {$=\sum\limits_{\widetilde{x}_i}\Big(F_{x_{i-1}jj}^{x_{i+1}}\Big)_{x_i}^{\widetilde{x}_i}$};
\draw (6+0,0) -- (6+0,{-0.5*sqrt(2)}) node[anchor= north] {\small$x_{i+1}$};
\draw (6+0,0) -- ({6-1},{1}) node[anchor= south] {\small$x_{i-1}$};
\draw (6+0.5,0.5) -- (6+0,1) node[anchor= south] {\small$j$};
\draw (6+0,0) -- ({6+1},{1}) node[anchor= south] {\small $j$};
\node at ({6+0.25+0.2},{0.25-0.2}) {\small $\widetilde{x}_{i}$};
\end{tikzpicture}
\caption{The $F$-move applied on the anyonic chain}
\end{figure}
We specialize to the case $j=\frac{1}{2}$, where the spin-$\frac{1}{2}$ chain is given a Hamiltonian. The motivation of all these settings could be seen as a generalization of the Heisenberg model \cite{gils2013anyonic}. In the Heisenberg model, there exist a spin-spin nearest neighbor interaction given by the term 
$$\vec{S}_i.\vec{S}_{i+1}=P_i^1-\frac{3}{4}I_i=-P_i^0+\frac{1}{4}I_i,$$
where $P_i^s$ is the projection onto the total spin $s$ channel of two spins $\vec{S}_i$ and $\vec{S}_{i+1}$. This leads to the following Hamiltonian 
$$H=J\sum_j P_j^0,$$

where $J$ determines if the chain is antiferromagnetic ($J=-1$) or ferromagnetic ($J=1$). In order to generalize this, we first need to define the projection onto the total spin using the so called $F$-move in \hyperref[fig2]{Fig. 2}.

The next step would be to project onto the desired fusion which is $0$---the vacuum- and go back to the previous basis of fusion path by applying the inverse of the $F$-move (\cite[Fig. 2]{gils2013anyonic},\cite[Fig. 1(c)]{feiguin2007interacting}):
$$H=\sum_{i=1}^{L-1}F_i^{-1}P_i^0F_i,$$
where the antiferromagnetic coupling $J=-1$ has been chosen in order to obtain UMMs in the scaling limit. In the case of spin-$\frac{1}{2}$ chain, letting $d=2\cos(\frac{\pi}{k+2})$, the quantum dimension of $\frac{1}{2}$, 
$$F_i^{-1}P_i^0F_i=-\frac{1}{d}X_i \implies H=-\frac{1}{d}\sum_{i=1}^{L-1} X_i.$$
The operators $X_i$ satisfy the following relations \cite[eq. (3)]{feiguin2007interacting}:
$$X_i^2=dX_i, \ \ X_iX_{i\pm1}X_i=X_{i}, \ \ [X_i,X_j]=0, \text{ for } |i-j|>1.$$
These are the same operators $e_i$ of the Temperley-Lieb (TL) algebra. Thus,
$$H=-\frac{1}{d}\sum_{i=1}^{L-1} e_i.$$
For the special case of Ising anyonic chain, if non-periodic, there are several possibilities $(a,b)$ for the boundaries as $a,b \in \{0,\frac{1}{2},1\}$. For example, the chain $(\frac{1}{2},\frac{1}{2})$ has odd length $L=2n+1$ due to the fusion rules and the Hamiltonian is $H=\frac{-1}{\sqrt{2}}\sum_{i=1}^{2n-1} e_i$. However, the periodic chain has always even length $2n$.

Back to the general case, the operator $e_i=e[i]$ acts non-trivially on the $i$-th particle according to its neighbor particles
$$e_i\ket{j_{i-1}j_ij_{i+1}}=\sum_{j_i'}(e[i]_{j_{i-1}}^{j_{i+1}})_{j_i}^{j_i'}\ket{j_{i-1}j_i'j_{i+1}},$$
where $(e[i]_{j_{i-1}}^{j_{i+1}})_{j_i}^{j_i'}$ is determined by the $S$-matrix entries of the MTC
$$(e[i]_{j_{i-1}}^{j_{i+1}})_{j_i}^{j_i'}=\delta_{j_{i-1},j_{i+1}}\sqrt{\frac{S_{j_i}^0S_{j_i'}^0}{S_{j_{i-1}}^0S_{j_{i+1}}^0}}, \ \ S_j^{j'}=\sqrt{\frac{2}{k+2}}\sin\Big(\frac{\pi(2j+1)(2j'+1)}{k+2}\Big).$$
In brief, from the MTC point of view, one can think of the (open) AC as a diagram inside $\text{Hom}(x_1 \otimes x_L, (\frac{1}{2})^{\otimes (L-1)})$ on which $e_j$ with the above entries act. Diagrammatically, the TL algebra acts by annihilating (cap) and then creating (cup) an adjacent pair of spin-$\frac{1}{2}$ particles. Then $H$ can be defined as above. 

Numerical experiments suggest that the scaling limit of the ACs of $SU(2)_k$ give chiral or full CFT (for open boundary condition or periodic chains, respectively). These results are outlined in \cite{gils2013anyonic,feiguin2007interacting}, and show that depending on the boundary condition, we obtain different chiral CFTs, i.e. different irreducible modules of UMMs with central charge $c=1-\frac{6}{(k+1)(k+2)}$. As emphasized before, this happens for the antiferromagnetic chain, and it is expected that one obtains the parafermion CFT with central charge $c=\frac{2(k-1)}{k+1}$ for the ferromagnetic chain.

Exact diagonalization numerically solves the anyonic chain model by finding the excitation spectra. Conformal dimensions of the predicted CFT limit are extracted from the energy levels for a length $L$ chain given by 
$$ E=E_1L+\frac{2\pi v}{L}(-\frac{c}{12}+h_L+h_R)+O(\frac{1}{L^2}).$$

The scaling limit CFT is stable under symmetry-preserving perturbation. More precisely, by symmetry, we mean the \textit{topological symmetry} that the periodic chain has. One can imagine a loop \cite[Fig. 3]{gils2013anyonic} inside the chain and repeatedly use the $F$-move until it gets removed. As demonstrated in \cite{gils2013anyonic,feiguin2007interacting}, this provides a topological symmetry and any perturbation preserving such symmetry will not change the scaling limit.

As a final note, an important connection between AC model and the RSOS lattice model provides a physical proof that the scaling limits of ACs are CFTs. One can show that the Hamiltonian derived from the logarithmic derivative of the transfer matrix, coincides with the AC Hamiltonian \cite{feiguin2007interacting}. This lattice model has been studied for a long time and the literature has similar numerical results for this model (see \cite{baxter2016exactly}, \cite{jones03}, and the references in \cite{feiguin2007interacting}).  While there is no doubt that the two approaches are equivalent in the end, mathematically it seems easier to obtain CFTs as scaling limits in the AC approach.  As comparison, we recall the recovery of Ising CFT in the 2d classical Ising model \cite{smirnov07}.  One aspect of the difference is manifested in the order of enforcing Jones-Wenzl projectors in $SU(2)_k$:  in the ACs, the Jones-Wenzl projector $p_{k+1}$ is implemented first, whereas in spin chains the projector, which is non-local, will be implemented at a later time.

\subsection{Outline of main results}\label{1.5}

\hfill \break
In this paper, we provide a mathematical definition of a low energy scaling limit of a sequence of general non-relativistic quantum theories in any dimension, and apply our formalism to ACs.  Similar ideas for defining related scaling limits for lattice models and spin chains appeared earlier in the physics and mathematics literature (see the next section).  Of utmost importance to our future applications are the rate of convergence to the scaling limit, and the recovery of all algebras of local observables in the scaling limit.  We emphasize those points for the scaling limits of the Ising ACs. 

We formulate \hyperref[cnj4.3rep]{\textbf{Conjecture 4.3}} reproduced below on conditions when a chiral unitary CFT would arise as such a limit and verify the conjecture for the Ising minimal model $M(4,3)$ using Ising ACs.  Part of the conjecture is a precise relation between Temperley-Lieb generators $\{e_i\}$ and finite versions of Virasoro generators $\{L_m+L_{-m}\}$ and $\{i(L_m-L_{-m})\}$ for UMMs $M(k+2,k+1)$.  Our approach is supported by extensive numerical simulation and physical proofs in the physics literature.

In \hyperref[2.]{section 2}, we define the low energy (strong) scaling limit of quantum \textit{theories} --- Hilbert spaces $\mathcal{W}_n$ with Hamiltonians $H_n$ and algebras of observables $\mathcal{A}_n$--- as a Hilbert space ${\mathcal{V}}$ with Hamiltonian $H$, and address the issues that come up with our definition. We define the scaling limit of observables $O_n \xrightarrow{SL} O$ when $O_n$'s low energy behavior converges to that of $O$.  Those $O$'s defined on $\mathcal{V}$ generate the \textit{vector space} $\mathcal{A}$ of observables on $\mathcal{V}$.

On an important related issue, we propose a definition of locality with respect to both space and energy. As an example, intuitively, local energy observables are those that do not shift the energy by more than a constant. We explore the (space and energy) local operators in \hyperref[4]{section 4} in greater details (\hyperref[thm4.1]{\textbf{Theorem 4.1}} and \hyperref[thm4.2]{\textbf{Theorem 4.2}}).

In \hyperref[3]{section 3}, we obtain the scaling limits of Ising ACs with all kinds of boundary conditions. Proving the limits is a rather computationally involved procedure, where the same technique (\cite{fendley2014free}) is applied to each case. While the variation of details from case to case is small, the details of the proofs are necessary for our future discussion. Part of the proofs for some cases has been done in the physics literature with different or similar approaches (see \cite{koo1994representations} as an example), but we could not locate a mathematically rigorous proof for all Ising ACs in the literature using one consistent method, and with explicit estimate of the convergence rate for the limits. Furthermore, we need to set up notations and use some details of the proofs in later sections.  In the end, we find enough reasons to provide a thorough and detailed mathematical proof of the scaling limit of all Ising ACs in one place, estimate the rate of convergence, and show how one can obtain the Virasoro algebra generators in the scaling limit. The final result is a long \hyperref[thm3.1]{\textbf{Theorem 3.1}} which motivates the conjecture below:

\begin{customcnj}{4.3}\label{cnj4.3rep}
For any unitary minimal model VOA $\mathcal{V}=\mathcal{V}_{c,0}$ and a chiral representation $\mathcal{V}_{c,h}$, there is a sequence of quantum theories $(\mathcal{W}_n, H_n, \mathcal{A}_n)$ with \textbf{strong} scaling limit $(\mathcal{V}_{c,h},L_0)$ such that for each Virasoro generator $L_m$, we have a sequence $\widetilde{L}_m(n) \in \mathcal{A}_n$ with the following properties:
\begin{itemize}
    \item $\widetilde{L}_m(n)$ is a space local observable such that the hermitian operators $a\widetilde{L}_m+\overline{a}\widetilde{L}_{-m} \in \mathcal{A}_n^H$ for any complex number $a$, where $\mathcal{A}_n^H$ is the generating set for $\mathcal{A}_n$ consisting of hermitian observables.
    \item  $\widetilde{L}_m(n)$ shifts the energy no more than $|m|$ for each $n$.
    \item There are positive constants $d_\omega,g_\omega,e_\omega$ such that when $\widetilde{L}_m(n)$ restricted to energy at most $n^{d_\omega}$, it has the following approximation by $L_m|_{n^{
    d_\omega}}$ with remainder $R_{n}^m$:
    $$\widetilde{L}_m=L_m|_{n^{d_\omega}}+O(\frac{1}{n^{g_\omega}})+R_{n}^m,$$
    and the operator norm of the remainder $R_{n}^m$ is bounded by $O(n^{e_\omega})$.
\end{itemize}
\end{customcnj}

In \hyperref[4]{section 4}, using the results in \hyperref[3]{section 3}, we obtain the observables of different types in the scaling limit. 

Obtaining the action of the VOA on its modules is the next step. As the VOA is generated by $L_n$s applied to the primary fields, the ability to obtain the Virasoro generators as scaling limits is of utmost importance. In other words, realizing the smeared conformal field $Y(\omega,f)$ should be the top priority. Then, we would need to get the smeared primary fields. For UMMs, the VOAs are generated from only the conformal vector $\omega$ applied on the vacuum (the \textit{only} primary field). Thus, obtaining $Y(\omega,f)$ is close to obtaining all the operators $Y(a,f)$ for all fields $a$ and that is due to the Borcherds identity. Assuming the above conjecture for UMMs, most theorems in \hyperref[4]{section 4} notably \hyperref[thm4.8]{\textbf{Theorem 4.8}} mentioned below, hold for all UMMs as well (the exact results are mentioned in \hyperref[rmk8]{\textit{Remark 8}} and \hyperref[rmk9]{\textit{Remark 9}}). So although our discussion is for the Ising ACs, everything below about the conformal field and $Y(a,f)$s is conjectured to hold for higher level UMMs as well.

Consider a non periodic Ising chain placed on the upper half-circle $S^1_+$, what is the ``finite version'' $\widetilde{\omega}$ of $\omega$? Informally, the answer is $\widetilde{\omega}=e$ (the TL operator) where subtleties are explained below. For example, we will show that for any function $f$ with a Fourier series where coefficients of $\sin(n\theta)$s are zero, we have (informally)
$$\int_{S^1_+} f(e^{i\frac{\pi j}{n}})e_j \xrightarrow{SL} Y(\omega,f)=\oint Y(\omega,z)f(z)z^2 \frac{\text{d}z}{2\pi i z},$$
where the integral on the left is an integral over a ``finite'' space, in other words, a summation. Hence, as $\omega$ can be regarded either as a vector or a field, so does $e_j$, which can be seen as a diagram or as an operator (stacking up diagrams). Here $Y$ (the vertex operator) is the analog of the stacking at infinity.

For the opposite situation, i.e. the functions with a Fourier series where coefficients of $\cos(n\theta)$s are zero, we could consider this as a \textit{derivative} of a series with \textit{nonzero} coefficients for only $\cos(n\theta)$. In this case, we are basically looking at the derivative of $Y(\omega,z)$ or in other words $\frac{\text{d}}{\text{d}z}Y(\omega,z)=Y(L_{-1}\omega,z)$. So it is necessary to find the corresponding operator for the finite version of $L_{-1}\omega$ which should be the \textit{derivative} (as it is the interpretation of $L_{-1}$) of $e_j$. The first candidate that comes to mind is $[e_{j},e_{j+1}]$ and informally, we have
$$i\int_{S^1_+} f(e^{i\frac{\pi j}{n}})[e_j,e_{j+1}] \xrightarrow{SL} Y(\omega,f)=\oint Y(\omega,z)f(z)z^2 \frac{\text{d}z}{2\pi i z}.$$
For general functions $f$, we need a linear combination of $e_j$ and $[e_j,e_{j+1}]$ to get $Y(\omega,f)$. Once this is achieved, an application of the Borcherds identity gives all $Y(a,f)$s and the action of the VOA on its module can be recovered as an algebra in the scaling limit (SL-algebra):
\begin{customthm}{4.8}
The set of operators $\{Y(a,f) | \ a \in \mathcal{V}, f \in C^\infty(S^1)\} \subset \mathcal{A}$ generate an SL-algebra. 
\end{customthm}

The importance of this result lies in the required recovery of operators as an algebra in the scaling limit, not just as isolated operators generating a vector space as explained in the beginning. 

Next goal is the algebra of observables in LCNs. We define finite versions of local conformal nets in the obvious way by thinking of intervals in the interval set $\mathcal{I}_+$ on the upper-half circle and the algebra of observables on them. Then, we show that there is a net of bounded observables $\mathcal{A}_b$ obtained in the scaling limit \textit{exactly} matching the one in an LCN, i.e. $\mathcal{A}_{lcn}$, at least on those intervals touching the boundary of the upper half-circle:

\begin{customthm}{4.13}
We have $\mathcal{A}_b(I)=\mathcal{A}_{lcn}(I \cup j(I))$ for $I \in \mathcal{I}_+$, with $j(I)$ being $I$'s reflection in the lower half-circle and $|I \cap \partial S^1_+|=1$
\end{customthm}

We also propose a method to recover the point-like fields (\hyperref[thm4.14]{\textbf{Theorem 4.14}}). In that case, we are unable to show that they form an algebra in the scaling limit. In fact, the sequence of operators that we are proposing to identify them in the scaling limit is probably not the suitable one. We will discuss these more in the relevant section.

In the final section, conjectures and problems that needs to be addressed to fully recover all the structures of CFTs in the scaling limit are listed. Lastly, we make an attempt to formulate the problem of simulating CFTs using quantum computers. 

\subsection{Previous works}\label{1.6}

\hfill \break
We discuss briefly some prior works in the literature on the mathematically rigorous definition of a scaling limit in the quantum mechanics approach, and the recovery of algebras of observables. As previously mentioned, there is a vast literature on the subject of scaling limits in statistical mechanics \cite{baxter2016exactly}, and substantial progress has been made in the case of Ising model proving the correlation functions in the limit are conformal invariant (see \cite{smirnov07,chelkak2012conformal} and the references therein). Statistical mechanics approach could also provide techniques with which one could compute the conformal weights present at the scaling limit without actually diagonalizing the Hamiltonian \cite{aasen2016topological}. A recent program to construct CFTs from subfactors is in \cite{jones17, jones14, jones16}, where the inductive limit of Hilbert spaces is clearly discussed based on planar algebras, which have the same Hilbert spaces of states as ACs (spin chains in these papers are better interpreted as generalized spin chains as in \cite{jones03}). Our work focuses on the quantum mechanics approach to scaling limits of ACs enriching the inductive limits \cite{jones14} with explicit Hamiltonians and algebras of local observables.

A scaling limit of spin chains close to our Ising AC was analyzed earlier in \cite{koo1994representations} starting with the idea of how to take the scaling limit of the Hamiltonians of the chains and also obtain the Virasoro modes $L_n$ from Fourier transforms of the TL generators $e_i$. More recently, in the first paper of the series \cite{gainutdinov2013continuum,gainutdinov2013bimodule,gainutdinov2016associative} on the $\mathfrak{gl}(1|1)$ (free) model, the authors proposed a potentially rigorous definition for the scaling limit \cite[section 4.3]{gainutdinov2013continuum}, obtained operators like our $\widetilde{L}_m$s and computed their commutators to check their convergence to the commutators of the Virasoro modes. Such computations are commonly pursued after one obtains some operators $\widetilde{L}_m \xrightarrow{SL} L_m$ and have been done in different models both rigorously and numerically (\cite{gainutdinov2014lattice},\cite[p.19 and references therein]{gainutdinov2013logarithmic}). We go beyond the convergence of commutators and further pin down the conditions necessary (\hyperref[cnj4.3]{\textbf{Conjecture 4.3}}) to prove the same theorems for higher level UMMs. 

In the third paper of the series (\cite{gainutdinov2016associative}), the authors gave a rigorous definition of scaling limit based on their previous ideas while working on the scaling limit of JTL algebra (with $d=0$) as it acts on a $\mathfrak{gl}(1|1)$ periodic spin-chain model (the scaling limit is the $c=-2$ Logarithmic CFT---symplectic fermions theory).  Even though the context and the type of model (on Logarithmic CFTs) are quite different from ours (unitary CFTs), our definitions closely mirror theirs.  But there are some differences due to our different motivation, emphasis and applications.

As defined in \cite[Appendix C]{gainutdinov2016associative}, our scaling limit is also dictated by the low energy behavior of Hamiltonians, Hilbert spaces, and observables.  In \cite{gainutdinov2016associative}, the primary focus is on the \textit{algebraic} scaling limit of $JTL_N$'s action ($\pi_{\mathfrak{gl}}(JTL_N)$). However, we focus on the analytic side of scaling limits motivated by our goal of simulating CFTs as we need to know how computations in the finite stages converge. Especially, the unitary evolution and correlation functions involve unbounded operators for which we desire a clear description on how they are obtained in the scaling limit. In fact, even when restricted to the bounded observables, not all bounded operators can be obtained through the algebraic approach (for example the unitary operators $e^{iL(f)}$). Related to this, the analytic approach provides a more direct picture at how the LCNs emerge (\hyperref[4.2]{section 4.2}) since we still keep the JTL operators $e_i$ as our operators of interests and mostly, do not switch to fermionic fields. This enables us to obtain theorems with proofs general enough for higher UMMs assuming \hyperref[cnj4.3]{\textbf{Conjecture 4.3}}.

The algebraic approach, and algebraic-numerical techniques \cite{read2007enlarged,gainutdinov2015periodic}, has been used to obtain more information about the algebraic structure of the Hilbert space and the algebra of observables in the scaling but to our knowledge, a mathematically rigorous procedure has been applied mainly for free models like $\mathfrak{gl}(1|1)$. 

Recently, emergence of conformal symmetry has been numerically investigated using the Koo-Saleur generators (KSGs) \cite{milsted2017extraction,koo1994representations}.  To compare our version of KSGs with those of \cite{milsted2017extraction},  first recall our notation
$\mathbb{L}_n=L_n+\overline{L}_n.$
Our counterparts of the KSGs are operators $\widetilde{\mathbb{L}}_n\pm \widetilde{\mathbb{L}}_{-n}$ on the ACs that give us $\mathbb{L}_n\pm \mathbb{L}_{-n}$ in the scaling limit. On the other hand, using a different diagonalization of the Hamiltonian in \cite{milsted2017extraction} (same as that in \cite{koo1994representations}), the authors found their KSG operators, different from ours, in the AC notation to be  $\widetilde{H}_n=\frac{-N}{2\pi}\sum_{j=1}^{2N} e^{\frac{2n(j+1)\pi}{2N}}e_j$, which converge to $L_n+\overline{L}_{-n}$. Taking the sum and difference of $\widetilde{H}_n$ and $\widetilde{H}_{-n}$ respectively, we obtain $\mathbb{L}_n+ \mathbb{L}_{-n}$ from the sum and $(L_n-L_{-n}) -(\overline{L}_n-\overline{L}_{-n})$ from the difference, which does not have a counterpart in our version.

The difference stems from different diagonalizations of the same Hamiltonian, which illustrates the potential importance of connecting maps in our definition of scaling limits. In \cite{koo1994representations}, the diagonalization of the Hamiltonian is accomplished by constructing creation and annihilation operators from the \textit{usual} Fourier transforms of the Majorana operators. While in our version, the creation and annihilation operators are obtained as $\sin()$ and $\cos()$ transforms for the left and right moving sectors, which implies  that going from one diagonalization to the other requires a mixing of the right and left moving sectors of the full CFT. 

It follows that the scaling limit of $\widetilde{H}_n$ from our diagonalization will have an \textit{interchiral} part which mixes left and right moving sectors which is clearly different from $H_n=L_n+\overline{L}_{-n}$. The method in \cite{milsted2017extraction} works well numerically, and for the NS sector  $\rchi_0\overline{\rchi}_0+\rchi_\frac{1}{2}\overline{\rchi}_\frac{1}{2}$, the resulting scaling limit (see e.g. \cite{koo1994representations} for a proof) gives rise to a full CFT isomorphic to ours by a not necessarily local isomorphism that connects the two different sets of creation and annihilation operators.

Finally, while not directly related, the paper \cite{konig2017matrix} serves as a conceptual inspiration for our work and the techniques introduced there address analytic problems of similar nature to ours.

\section{Scaling limit of quantum theories}\label{2.}

It is commonly believed that QFTs are low energy effective theories such as WCS TQFTs are the low energy effective theories for two dimensional fractional quantum Hall liquids.  In this section, we define mathematically a low energy limit of a sequence of quantum theories.  Our formalism is closely related to the definition of topological phases in \cite{RW17} and ideas in \cite{gainutdinov2013continuum}.  

We start with the definition of quantum theories by imagining quantum theories that describe a collection of interacting quantum particles. The theories considered have a discrete energy spectrum in the scaling limit like all CFTs. Notice this is different than the energy spectrum given by the primary fields. In the context of CFTs, there are non unitary Virasoro representations with continuous spectrum of primary fields, while still having a discrete energy spectrum in each sector. The definition below is for a finite dimensional theory, with the next sections defining what the scaling limit (infinite dimensional) is. 
\begin{dfn}
A quantum theory is $(\mathcal{W}, H, \mathcal{A})$ where
\begin{itemize}
    \item $\mathcal{W}$ is the Hilbert space of states,
    \item $H$ is the Hamiltonian and hermitian,
    \item $\mathcal{A}$ is the algebra of observables,
\end{itemize} 
\end{dfn}
\label{dfn2}
\begin{rmk}\label{rmk1}
We can also add a number of notions to the definition above. For example, The space information of the system can be thought of a graph $G$, which is usually the $1$-skeleton of a triangulation of the space. In the following text, $G$ is always a chain. There are also different notions of locality based on the basis we choose. As an example, considering the space information given the graph $G$, the Hamiltonian $H$ is $r$-local for some constant $r>0$ if $H=\sum_{i=1}^p H_i$ such that each local hermitian term $H_i$ is trivial outside the ball $B_r(v_i)$ of distance $r$ at some vertex $v_i$ of $G$. If $p=1$, then $H$ is ultra $r$-local. There will also be a notion of \textit{energy}-local operators as discussed later in this section.
\end{rmk}
\subsection{Low energy limit of quantum theories}\label{2.1}

\hfill \break 
The first part of a limit theory is a Hilbert space and a Hamiltonian, which are constructed from the low energy spectra of a sequence of quantum theories $(\mathcal{W}_n,H_n)$ with strictly increasing dimensions.

Assume a sequence of quantum theories $(\mathcal{W}_n,H_n)$ with $H_n$'s eigenvalues being ordered as $\lambda_1^{(n)} \le \ldots \le \lambda_{d(n)}^{(n)}$, where $d(n)=\text{dim}(\mathcal{W}_n)$. The Hilbert spaces $\mathcal{W}_n$ decompose into the corresponding one-dimensional eigenspaces
$$\mathcal{W}_n=E_{\lambda_1^{(n)}} \oplus \cdots \oplus E_{\lambda_{d(n)}^{(n)}}.$$ 
Denote by $\mathcal{W}_n^M$ the Hilbert space $\mathcal{W}_n$ restricted to energies at most $M$, i.e. $\mathcal{W}_n^M=\bigoplus_{\lambda_i^{(n)}\leq M}E_{\lambda_i^{(n)}}$. Assume the following set of properties (\textbf{P}\label{P})
\begin{itemize}
    \item $\lambda_i=\lim\limits_{n \to \infty} \lambda_i^{(n)}$ exists for all $i \in \mathbb{N}$ with the convention $\lambda_i^{(n)}=0$ for $i>d(n)$, and $\lim\limits_{i\to \infty}\lambda_i=\infty$,
    \item (\textit{connecting maps}) for all $M > \lambda_1$ where $M \neq \lambda_j$ for all $j$, there exist \textit{connecting} unitary maps $\phi_n^M: \mathcal{W}_n^M \to \mathcal{W}_{n+1}^M$ for all $n>N_M$ for some $N_M$ depending on $M$,
    \item (\textit{extension}) $\phi_n^M$ is an extension of $\phi_n^{M'}$ when $M \ge M'$, i.e. $\phi_n^M|_{\mathcal{W}_n^{M'}}=\phi_n^{M'}$.
\end{itemize}
Consider the sequence $(\mathcal{W}_n^M,\phi_n^M)$ with $M > \lambda_1$ and $M\neq \lambda_j$. We note that this sequence eventually \textit{stabilizes} due to the existence of \textit{unitary} maps for large enough $n$.

The reason for $M \neq \lambda_j$ for all $j$ in the first property is that energies oscillating around their limit points would make the stabilization of the low-energy spectrum impossible for a cut-off $M=\lambda_j$. From now on, any cut-off will be implicitly assumed to be not equal to any $\lambda_j$.

Taking the colimit of the sequence $(\mathcal{W}_n^M,\phi_n^M)$ gives a finite dimensional vector space, called $\mathcal{V}^M$, along with the unitary maps $\rho_n^M : \mathcal{W}_n^M \to \mathcal{V}^M$. It follows that  $\mathcal{V}^M$ has a \textit{natural} Hilbert space structure. Further, due to the first property, one can easily see that for all $M\in (\lambda_j,\lambda_{j+1})$, the space $\mathcal{V}^M$ is the same as $\mathcal{W}_n^M$ will have the same dimension (for large enough $n$). The space $\mathcal{V}^M$ can be also conveniently called $\mathcal{V}^{\lambda_j}$. So there are only countably many different $\mathcal{V}^M$s. Next we add the following property to \textbf{P} on the convergence of $H_n^M$, the restriction of $H_n$ to $\mathcal{W}_n^M$: 
\begin{itemize}
\item (\textit{convergence}) The push-forward of $H_n^M$ on $\mathcal{V}^M$ given by $\rho_n^M$ converges to some operator $H^M$:
$$\rho_n^MH_n^M(\rho_n^M)^{-1} \to H^M.$$
\end{itemize}

Obviously, $H^M$ will be hermitian. Furthermore, the above property is equivalent to the following diagram ``commuting up to $\epsilon_{n}^M$ in the norm operator'', which goes to zero as $n \to \infty$:
\[ 
\begin{tikzcd}
\mathcal{W}_n^M \arrow{r}{\rho_n^M} \arrow[swap]{d}{H_n^M} & \mathcal{V}^M \arrow{d}{H^M} \\
\mathcal{W}_n^M \arrow{r}{\rho_n^M}& \mathcal{V}^M
\end{tikzcd}.
\]
  
The construction of the \textbf{scaling limit} $(\mathcal{V},H)$ of the sequence is not hard from here. We spell out the formal details although future constructions in the case of Ising will be much more straightforward.

Properties of the colimit imply that the set $\{(\mathcal{V}^M,H^M)\}_{M > \lambda_1}$ is unique up to unique isomorphism. We would like to construct these spaces in such a way that $\{(\mathcal{V}^M,H^M)\}$ are restrictions of a single Hilbert space and its Hamiltonian $(\mathcal{V},H)$. 

Due to the \textit{extension} property, $N_{M} \ge N_{M'}$ for $M\ge M'$. The existence of connecting maps for $n>N_M$ ensures that one can build $\mathcal{V}^M$ using the orthonormal basis
\begin{gather*}
\scalebox{0.93}{$
\{(v_i^{(N_M+1)},\phi_{N_M+1}^M(v_i^{(N_M+1)}),\phi_{N_M+2}^M(\phi_{N_M+1}^M(v_i^{(N_M+1)})),\ldots)\}_{i=1}^{\dim \mathcal{V}^M}$}
\end{gather*}
where $v_i^{(N_M+1)}$ is an orthonormal basis of $\mathcal{W}_{N_M+1}^M$. Each sequence represents an actual vector $\textbf{v}_i^M$ and addition is component-wise and inner product is given by the inner product on any component, which due to isometry of connecting maps gives the same number. To be really precise, we will have to take $\textbf{v}_i^M$ not exactly as that sequence, but as the colimit of that sequence and every sequence that is a truncation of the sequence from the left. This choice will soon become clear. Now consider making the above construction (inductively) for the sequence $M=\lambda_1',\lambda_2',\ldots$ where $\lambda_j'\in (\lambda_{j},\lambda_{j+1})$ in such a way that for $M=\lambda_j'$ for $j>1$, the orthonormal basis used in the previous case is \textit{extended}. As an \textit{example}, assuming $\lambda_1<\lambda_2<\lambda_3$, for $M=\lambda_2'$, we take the orthonormal basis provided by $M'=\lambda_1'$  on $\mathcal{W}_{N_{M}+1}^{M'}$ which is 
$$\{ \phi_{N_M}^{M'}(\phi_{N_M-1}^{M'}( \ldots (\phi_{N_{M'}+1}^{M'}(v_i^{(N_{M'}+1)})) \ldots )) \}_{i=1}^{\dim \mathcal{V}^{M'}}$$
and extend it to an orthonormal basis for $\mathcal{W}_{N_{M}+1}^{M}$. Notice that the colimit will be having the vectors $\{\textbf{v}_i^{M'}\}_{i=1}^{\dim \mathcal{V}^{M'}}$ which are the colimit of the sequences starting by the vectors of basis of $\mathcal{W}_{N_{M}+1}^{M'}$ mentioned above; this is why we had to consider truncations from the left as the sequence does not start with basis of $\mathcal{W}_{N_{M'}+1}^{M'}$. It is formal diagram chasing, using the extension property of the connecting maps, that this construction gives a well-defined colimit $\mathcal{V}^M$. By the choice of $\lambda_j'$s, these $\mathcal{V}^M$s are $\{\mathcal{V}^{\lambda_j}\}_j$ which means we have obtained $\mathcal{V}^M$ for \textbf{all} possible cut-off $M$. This is also compatible with the action of $H^M$ as defined in the convergence property. This means for the set $\{(\mathcal{V}^M,H^M)\}_{M>\lambda_1}$, the embedding $\mathcal{V}^{M'} \rightarrow \mathcal{V}^{M}$ is by identity for $M' \le M$ and the following diagram commutes
\[ 
\begin{tikzcd}
\mathcal{V}^{M'} \arrow{r}{} \arrow[swap]{d}{H^{M'}} & \mathcal{V}^M \arrow{d}{H^M} \\
\mathcal{V}^{M'} \arrow{r}{}& \mathcal{V}^M
\end{tikzcd}.
\]
Taking a second colimit of the set $\{(\mathcal{V}^M,H^M)\}_{M > \lambda_1}$ leads to the desired scaling limit $(\mathcal{V},H)$ where $\{(\mathcal{V}^M,H^M)\}$ are restrictions of $(\mathcal{V},H)$. To see this, the (unique) operator $H$, which would make the diagram below commute \textbf{for all} $M$, is the desired Hamiltonian. Notice the embedding $\mathcal{V}^M \rightarrow \mathcal{V}$ is by identity since $\mathcal{V}$ is a union of all $\mathcal{V}^M$s as they have a nested structure.
    \[ 
    \begin{tikzcd}
    \mathcal{V}^{M} \arrow{r}{} \arrow[swap]{d}{H^{M}} & \mathcal{V} \arrow{d}{H} \\
    \mathcal{V}^{M} \arrow{r}{}& \mathcal{V}
    \end{tikzcd}.
    \]

Since two colimits are taken to obtain the scaling limit (similar to the construction in \cite[section 4.3]{gainutdinov2013continuum}), the above process is called the \textit{double colimit construction}, allowing the following definition
\begin{dfn}\label{dfn3}
Given a sequence of quantum theories $(\mathcal{W}_n,H_n)$ with given connecting maps $\phi_n^M$ satisfying properties \hyperref[P]{(\textbf{P})}, the \textit{scaling limit} $(\mathcal{V},H)$ is the result of the double colimit construction. This limit will be written as $(\mathcal{W}_n,H_n) \xrightarrow{SL} (\mathcal{V},H)$.
\end{dfn}

We emphasize that as long as the connecting maps are specified the scaling limit process is unique up to unique isomorphism due to the nature of colimit. From now on, whenever a sequence of quantum theories is given with a scaling limit, implicitly, there is a given set of connecting maps. We do not discuss the issue of uniqueness any further and for a relevant example, we refer to the previous discussion in \hyperref[1.6]{section 1.6} on different diagonalization in the case of the Ising full CFT.

Notice that $\mathcal{V}$ is separable but not complete, so not yet a Hilbert space. The completion of $\mathcal{V}$ will be $\overline{\mathcal{V}}$. For notational easiness, The scaling limit will be written as $(\mathcal{V},H)$ with the understanding that one needs to take a completion whenever the context requires so.

We would like to think of the scaling limit as the result of \textit{stacking up} the low energy spectra of $H_n$s, and the double colimit construction indeed fulfills this expectation. Let  $\textbf{E}_{\lambda_1}$ be the eigenspace of the limit Hamiltonian $H$ corresponding to $\lambda_1$ and $\lambda_k$ be some larger eigenvalue of $H$ for some $k$. Choose some $M$ such that $\lambda_1<M<\lambda_k$, then the above construction builds $\textbf{E}_{\lambda_1}$ from the spaces $\mathcal{W}_n^M$ with large enough $n$, which contains all the vectors whose energy converges to $\lambda_1$ in the limit. The same holds for the other eigenspaces. 

Although our definition does not assume an embedding of the whole space $\mathcal{W}_n$ into $\mathcal{W}_{n+1}$, we expect this to be the case for all physical models. Indeed, scaling limit should be after all a physical process in which a whole system is embedded into another one when some new particles are added.

Our discussions in \hyperref[4]{section 4} will be based on this assumption, hence the need for a more refined definition:
\begin{dfn}\label{dfn4}
$(\mathcal{W}_n,H_n)$ have a \textit{strong} scaling limit $(\mathcal{V},H)$ if in addition to properties \textbf{P}, for all $n$ and $M$, the connecting maps $\phi_n^M$ are the restriction up to energy $M$ of an isometry
$$\phi_n : \mathcal{W}_n \hookrightarrow \mathcal{W}_{n+1},$$
for large enough $n$.
\end{dfn}

Given the above, the colimit of the sequence of embeddings $\mathcal{W}_n \hookrightarrow \mathcal{W}_{n+1}$ gives $\mathcal{V}$ directly.

Usually, the chosen basis for $\mathcal{W}_n$ closely relates to a notion of space, and locality in this space basis is supposed to represent locality in space. Finding the embedding $\phi_n$, though an isometry, is not  trivial based on this basis. In the scaling limit, the \textbf{space} embedding is not the ``trivial'' embedding, in contrast to the thermodynamical limit \cite[Appendix A]{fiedler2017jones}. In the scaling limit, the \textbf{energy} embedding is the trivial one as shown in the definition. As a result of this trivial energy embedding, the space local operators in $\mathcal{W}_n$, like $e_i$s in ACs, are generally space non-local when their actions are push-forwarded. This will become clearer in next few sections.

Finding the ``energy basis'' requires an understanding of the energy local degrees of freedom (EL-DOFs), which comes from an exact diagonalization of the Hamiltonian. Even numerical exact diagonalization is very limited for interacting models. In the Ising AC case, exact diagonalization analytically gives us the creation and annihilation operators, which are the EL-DOFs.  This, in turn, provides us the energy basis, which allows us to construct the scaling limit at each energy eigenspace. For all the models with known CFT limits, only free theories have mathematical descriptions of their EL-DOFs so far (see \cite{fendley2014free, gainutdinov2016associative, gainutdinov2014lattice} for some recent examples).

Another difficulty with a scaling limit is the description of observables in the limit. In the scaling limit, a ``space'' description of the operators in the limit is hard to find. For example, if we look at any observable in a CFT, the description which allows us to \textit{compute} with, is in terms of mode operators, which are more naturally described as energy shifting operators while their space action is obscure. Indeed, the $Y(a,z)$s are considered to be space local observables, yet their description is a Fourier series of mode operators $\sum a_nz^{-n-1}$. This, along with the fact that the chosen basis for $\mathcal{W}_n$ is closely related to the notion of space and not energy, complicates the process of finding a description of observables in the scaling limit. In the case of ACs, the $e_i$s are space local operators. Therefore, having a general definition of Fourier transform on the $e_i$s is essential, especially one that relates to the mode operators $L_n$. Alternatively, one will have to find and work with some space description of $Y(a,z)$.

\subsection{Scaling limit of observables}\label{2.2}

\hfill \break
Given a sequence of quantum theories $\{(\mathcal{W}_n,H_n,\mathcal{A}_n\}_{n=1}^{\infty}$ with the scaling limit $(\mathcal{V},H)$, by definition, Hilbert spaces $\mathcal{W}_n$ have strictly increasing dimensions,  Hamiltonians $H_n$, and algebras of observables $\mathcal{A}_n$. Recall that the algebra $\mathcal{A}_n$ is generated by an underlying real vector space of hermitian observables called $\mathcal{A}_n^H$, and $H_n \in \mathcal{A}_n^H$. In the examples of ACs, the space $\mathcal{A}_n^H$ is spanned by $\{e_j,i[e_j,e_{j+1}]\}$. This choice of generating set is motivated on one hand from including the local terms of interaction of the system, and on the other hand to recover the Virasoro algebra in the scaling limit; see \hyperref[thm4.2]{\textbf{Theorem 4.2}} and \hyperref[rmk12]{\textit{Remark 12}}. 

To build the observables of $\mathcal{V}$ from the observables in $\mathcal{W}_n$, the low energy behavior of the observables has to be taken into account.
\begin{dfn}\label{dfn5}
Let $O_n \in \mathcal{A}_n$ be \textit{any} sequence of observables. For a given $M$ and $u,v \in \mathcal{V}^M$, denote by $u_n,v_n \in \mathcal{W}_n^M$ the vectors $(\rho_n^M)^{-1}u,(\rho_n^M)^{-1}v$, which are defined for sufficiently large $n$. The scaling limit of $O_n$ is a partially-defined (defined on a subset of $\mathcal{V} \times \mathcal{V}$) sesquilinear form $O(\cdot,\cdot)$, where $O(u,v)$ is defined as $\lim\limits_{n \to \infty}(u_n,O_nv_n)$ when it exists. We will denote the scaling limit by $O_n \xrightarrow{SL} O$.
\end{dfn}
The idea is that the operator $O$ is constructed to exactly store the information in the expectation values of $O_n$.
\begin{dfn}\label{dfn6}
Consider the set of sesquilinear forms (which will also be called operators) in \hyperref[dfn5]{\textbf{Definition 5}}. Define
\begin{itemize}
    \item $\overline{\mathcal{A}}$: the set of observables in \hyperref[dfn5]{\textbf{Definition 5}},
    \item $\mathcal{A}$: the vector space of sesquilinear forms that are scaling limit of observables in $\mathcal{A}_n$ and defined on $\mathcal{V}\times\mathcal{V}$,
    \item $\mathcal{A}^H$: the real vector space consisting of all hermitian operators defined on $\mathcal{V}$ which are scaling limits of hermitian observables in $\mathcal{A}_n^H$.
\end{itemize}
\end{dfn}

\begin{rmk}\label{rmk2}
One can ask whether $\mathcal{A}$ \textit{generates} $\overline{\mathcal{A}}$?  This is true for Ising and any other model with what would be called an \textit{algebraic} scaling limit (defined after the next remark). The algebraic scaling limit gives a \textit{copy} of each observable of $\mathcal{A}_n$ inside $\mathcal{A}$. Therefore, any operator $O_n \xrightarrow{SL} O \in \overline{\mathcal{A}}$ can be seen as an operator obtained as scaling limit of the copies of $O_n$ inside $\mathcal{A}$, implying that $\mathcal{A}$ in a sense generates $\overline{\mathcal{A}}$.
\end{rmk}
\begin{rmk}\label{rmk3}
Another simple observation is that (again by some standard diagonal argument) $\mathcal{A}$ is closed under the obvious  ``weak limit''. In fact, we can consider the semi-norms $||\cdot||_n$ on $\mathcal{A}$ which is defined by $||O||_n=||P^nOP^n||$ where $P^n$ is the restriction up to energy $\lambda_n$, and $||P^nOP^n||$ is the usual norm of a \textit{linear} operator (linear, as it is nonzero only on a finite-dimensional space). Then, it is not hard to see that $\mathcal{A}$ is a Fr\'echet space with respect to these (separated) countably many semi-norms providing the \textit{scaling limit metric} $d_{\text{SL}}$. Indeed, to show completeness of $\mathcal{A}$, assuming operators $O^{(n)}\in \mathcal{A}$ forming a Cauchy sequence, one can easily construct their limit $O$. But to prove this limit is in $\mathcal{A}$, we need a sequence $O_{i}^{(n_i)} \in \mathcal{A}_i$ having scaling limit $O$. This sequence is constructed by a standard diagonal argument from the sequences $O^{(n)}_i \in \mathcal{A}_i$ giving $O^{(n)}$s. 
\end{rmk}
Adding to the above remark, in the case of a strong scaling limit, assume there also exist embeddings $\tau_n: \mathcal{A}_n \hookrightarrow \mathcal{A}_{n+1}$ compatible with the embeddings $\phi_n$, i.e. $\phi_n \circ O_n = \tau_n(O_{n})|_{\phi_n(\mathcal{W}_n)}, \ \forall O_n \in  \mathcal{A}_n$. Then scaling limit becomes convergence in the metric $d_{\text{SL}}$. In fact, the closure with respect to $d_{\text{SL}}$, of the colimit of the sequence of embeddings $\mathcal{A}_n \hookrightarrow \mathcal{A}_{n+1}$ is precisely $\mathcal{A}$, which we could call \textit{analytic} scaling limit. The colimit can also be called the \textit{algebraic} scaling limit and it contains a copy of each $O_n$ by the sequence $O_n,\tau_n(O_n),\tau_{n+1}(\tau_n(O_n)),\ldots$ which by algebraic construction (or scaling limit as $\tau_n$ is compatible with $\phi_n$) gives a copy of $O_n \in \mathcal{A}$ defined on $\mathcal{V}$. By going through the definitions, we have the same picture presented in \hyperref[rmk2]{\textit{Remark 2}}. The embeddings $\tau_n$ exist in the case of study in \cite[Theorem 4.4]{gainutdinov2016associative}. A similar theorem can be established for the Ising ACs as both algebras are simply the even algebra generated by Dirac operators. We conjecture that it holds for higher level anyonic chains. 

When one looks at the different set of observables in the different frameworks for unitary CFTs, there is always an underlying set of hermitian observables generating the whole set. Indeed, as proved in \cite{carpi2015vertex}, the \textit{hermitian fields} (more strongly, hermitian quasi-primary fields) generate the VOAs. As for LCNs, since the algebra corresponding to an interval $I$ is a Von Neumann algebra, it is trivially true that it can be generated by hermitian observables. But does $\mathcal{A}^H$ generate $\mathcal{A}$ in any way? We do not know the general answer.

We wish to identify some subsets of $\overline{\mathcal{A}}$ that may be algebras. Since some operators are not linear, it is not clear how one can have an algebraic structure. Note that by definition, there might not be a linear operator which gives the sesquilinear form $O$. But if such an operator exists, it will be called $O$ as well.

In some cases, these operators can be \textit{almost linear}. Consider for example the case of $Y(a,z)$ in VOAs. Their expectation values are defined, while none of them is actually defined on the VOA. They are \textit{almost} linear operators since there is a grading of the VOA. For an almost linear observable, one can formally set $Ov=\sum_i v_i$ where $v_i$ are the well-defined degree $\lambda_i$ component of $Ov$, and for any $u \in \mathcal{V}_i$, we have $O(v,u)=(v_i,u)$. This motivates us to call all operators inside $\mathcal{A}$ \textit{almost linear} operators. If the formal sum is always finite, $O$ is a linear operator.

The definition for the product of such operators is exactly in the same spirit of the correlation function $$(u,Y(a_1,z_1)\ldots Y(a_k,z_k)v).$$
\begin{dfn}\label{dfn7}
Given almost linear observables $O^{(1)},\ldots,O^{(k)} \in \mathcal{A}$, we define their product as a partially-defined sesquilinear form $F$ by using the formal sum interpretation. If the result is \textbf{absolutely} convergent for some $u,v \in \mathcal{V}$,
$$F(u,v):=(u,O^{(1)}\ldots O^{(k)}v),$$
then the above is considered to be well-defined.
\end{dfn}
We discuss one basic obstacle to get an algebraic structure by an example; observables $O_n$ that have a \textit{significant mix} of the low and high energy states. For example, the two sequences below where $v^{(i)}_n \in E_{\lambda_i^{(n)}}$ are pull-back of some $v^{(i)} \in E_{\lambda_i}$:

\begin{itemize}
    \item $O_{n,1}=v_n^{(1)}(v^{d(n)}_n)^{\dagger}+v^{d(n)}_n(v_n^{(1)})^{\dagger}$,
    \item $O_{n,2}=0$.
\end{itemize} 

Both sequences converge to zero while being quite different. The significant (non-decaying) mix of low-high energy states in the $O_{n,1}$s  manifests itself not in the expectation values of the observables at low energies, but the higher powers of the observables. 

Looking at the expectation values of powers, $O_{n,2}^k \xrightarrow{SL} O^k=0$ while  $O_{n,1}^2 \xrightarrow{SL} v^{(1)}(v^{(1)})^{\dagger} \neq 0$. Next example shows that just the decay of this low-high energy mix is not enough:
\begin{itemize}
    \item $O_{n,1}=v_n^{(1)}(v_n^{(1)})^{\dagger}+2^{d(n)}v^{d(n)}_n(v^{d(n)}_n)^{\dagger},$
    \item $O_{n,2}=v_n^{(1)}(v_n^{(1)})^{\dagger}+ \sum_{i=1}^{d(n)}\frac{1}{i^2}\Big(  v_n^{(i)}(v_n^{(1)})^{\dagger}+v_n^{(1)}(v_n^{(i)})^{\dagger}\Big)$.
\end{itemize}
It is not hard to check that 
$$O_{n,1}^k \xrightarrow{SL} O_1^k, \ \ \text{where} \ \ O_1=v^{(1)}(v^{(1)})^{\dagger},$$ 
and
$$O_{n,2}^k \xrightarrow{SL} O_2^k, \ \ \text{where} \ O_2=v^{(1)}(v^{(1)})^{\dagger}+ \sum_{i=1}^{\infty}\frac{1}{i^2}\Big(  v^{(i)}(v^{(1)})^{\dagger}+v^{(1)}(v^{(i)})^{\dagger}\Big).$$
Both sequences would be regarded as well-behaved but the first one has a significant high-high energy mix while the second one has a decaying low-high energy mix. One can check that $||O_{n,1}O_{n,2}v_n^{(1)}|| \not \rightarrow ||O_1O_2 v^{(1)}||$, i.e.
$$(v_n^{(1)},O_{n,2}O_{n,1}O_{n,1}O_{n,2}v_n^{(1)}) \not \rightarrow (v^{(1)},O_2O_1O_1O_2v^{(1)}).$$
So $O_{n,2}O_{n,1}O_{n,1}O_{n,2}$ does not have $O_2O_1O_1O_2$ as a scaling limit. The reason behind this is an \textit{imbalance} between the low-high energy mix decay rate and the rate of high-high energy mix. We note that it is possible to have a collection of observables with high-high energy mix, which is even increasing, and yet have an algebra, as will be shown in the case of Virasoro operators $\widetilde{L}_n \xrightarrow{SL} L_n$. For the discussion of algebra structures in scaling limit, a natural definition is 
\begin{dfn} \label{dfn8}
Given a set of almost linear observables $\{O^{(i)}\}_{i \in I}$, and the algebra of operators generated by this set. If this algebra is inside $\mathcal{A}$, we call the resulting algebra a scaling limit algebra (SL-algebra).
\end{dfn}
There could be many overlapping and yet different and maximal sets of observables forming an SL-algebra. Some of these are special in the sense that each observable has a \textit{nice} sequence associated to:
\begin{dfn}\label{dfn9}
Given an SL-algebra as in \hyperref[dfn8]{\textbf{Definition 8}}, assume each $O^{(i)}$ is associated a sequence $O^{(i)}_n$  such that for any $i_1,i_2,\ldots,i_k$,
$$\lim\limits_{n\to\infty}(u,O^{(i_1)}_n\cdots O^{(i_k)}_nv)=(u,O^{(i_1)}\cdots O^{(i_k)}v), \ \ \forall u,v \in \mathcal{V}.$$
Then the algebra generated by $\{O^{(i)}\}_{i \in I}$ is called a \textbf{strong} SL-algebra.
\end{dfn}

For example, in the case of the Ising model, $\{L(f)| f \in C^{\infty}(S^1)\}$ gives a strong SL-algebra. The above definition assumes a strong property which is sometimes not easy to show; in \hyperref[4]{section 4}, it is shown that $\{Y(a,f)| f \in C^{\infty}(S^1)\}$ gives only an SL-algebra.

\subsection{Locality in scaling limit}\label{2.3}

\hfill \break
First we review the terminology when it comes to the meaning of \textit{local} observables. 

In LCNs (or more generally for QFTs in Haag Kastler's axioms), a local \textbf{net} $A$ of Von Neumann algebras refers to the locality axiom: If $I_1$ and $I_2$ are spacelike separated, then elements in $A(I_1)$ and $A(I_2)$ commute. So local is used for the \textbf{net} when it satisfies the locality axiom. But also elements inside the \textbf{local observables} algebra $\mathcal{A}(I)$ are called local observables \cite{haag2012local,carpi2015vertex}.

For the VOA or more generally for Wightman's axioms in QFT, observables are \textit{(primary) fields} or \textit{distribution of operators} $\Phi$ and \textit{limits of observables localized at a point $x$,} $\Phi(x)$ \cite{fredenhagen1981local}. In addition, there are \textit{local} smeared fields $\Phi(f)$ with functions $f$ having support in some region $O$ \cite[II.4.1]{haag2012local} (if $f$ is a test function, then $\Phi(f)$ is ``almost local''). We also have a similar locality axiom: Let $\Phi_1$ and $\Phi_2$ be two observables and functions $f_1$ and $f_2$ be space-like separated in their supports, then $[\Phi_1(f_1), \Phi_2(f_2)]= 0.$ 

The conclusion is that there is a notion of locality in all frameworks as an axiom and the elements of the sets satisfying those axioms are called local observables.

Our definition of locality (in space and in energy) turns out to be more restrictive.

\subsubsection{Energy-local observables}

\hfill \break 
One goal in this work is to find out the constraints on observables in the scaling limit that will force them to be a specific type of observables (Wightman's, bounded as in LCN or point-like fields). Locality is one of these fundamental constraints.

We propose a definition of energy local operators without using any explicit knowledge of the EL-DOFs. Therefore, it might not be the most refined definition. Still, our notion of energy locality, which is intrinsic, together with space locality put enough constraint on operators so that they are easier to work with (see \hyperref[thm4.1]{\textbf{Theorem 4.1}}).

All smeared operators $Y(a,f)$ where $f$ has finite Fourier series do not shift the energy of any eigenvector by more than a constant. This is a motivation for the definition of \textit{energy locality} and to analyze energy local observables in general.

\begin{dfn}\label{dfn10}
The sequence $(O_n)_n$ is $\Lambda$-energy local for $\Lambda \in \mathbb{N}$, if for any $n$ and for all $u\in E_{\lambda_i},v \in E_{\lambda_j}$ with $|i-j|>\Lambda$, and any sequences $u_n \in \textbf{E}_{\lambda_i^{(n)}},v_n \in \textbf{E}_{\lambda_j^{(n)}}$ with $u_n \xrightarrow{SL} u, v_n \xrightarrow{SL} v$,
$$(u_n,O_nv_n)=0.$$
Any observable $O \in \mathcal{A}$ which is the scaling limit of such a sequence is also called a $\Lambda$-energy local observable.
\end{dfn}

It turns out that any almost linear observable which is $\Lambda$-energy local is a linear operator, as the formal sum $Ov$ is a finite sum with no more than $2\Lambda$ terms. The important observation is
\begin{thm}\label{thm2.1}
The set of all $\Lambda$-energy local observables for all $\Lambda$ forms a strong SL-algebra.
\end{thm}
\begin{proof}
Consider $\Lambda_i$-energy local observables $O^{(i)},1\le i\le k$ and corresponding sequences $(O_n^{(i)})_n$. Note $(u,O^{(k)}\ldots O^{(1)}v)$ is well-defined; indeed, if $v \in E_{\lambda_t}$ for some $t$, then every multiplication by some $O^{(i)}$ makes a vector in a space enlarged by adding $\Lambda_i$ to the energy level. This means taking projections onto $\mathcal{L}=\bigoplus_{i=t-\sum \Lambda_i}^{t+\sum \Lambda_i}E_{\lambda_i}$ called $P_{\mathcal{L}}$, all operators $O^{(i)}$ in the product can be replaced with the linear operator $P_{\mathcal{L}}O^{(i)}P_{\mathcal{L}}$ without changing the result.

Similarly for the corresponding expectation values $(u_n,O^{(k)}_n\ldots O^{(1)}_nv_n)$, everything is also happening in a finite dimensional Hilbert space. In fact, the limit can be taken with restriction to $\mathcal{W}_n^M \backslash \mathcal{W}_n^{M'}$, with $\lambda_{t+\sum \Lambda_j}<M<\lambda_{1+t+\sum \Lambda_j}$ and $\lambda_{t-1-\sum \Lambda_j}<M'<\lambda_{t-\sum \Lambda_j}$, which is a finite dimensional Hilbert space stabilizing for large enough $n$ and becoming isometric to $\mathcal{L}$. This means for large enough $n$, we might as well assume that all operators $O^{(i)}_n$ are acting on $\mathcal{L}$, by using the connecting maps followed by the projection $P_\mathcal{L}$ like the previous case. In this setting, we have a sequence of operators weakly convergent, but all acting on a finite dimensional Hilbert space. This implies norm convergence and the convergence of their product as a $(\sum \Lambda_j)$-energy local operator.
\end{proof}

This is our first example of an algebraic structure which is preserved under the scaling limit. One can ask whether it is truly necessary for a \textbf{constant} $\Lambda$ to be present in order to define energy locality. One might think of the possibility to enlarge the set of all $\Lambda$-energy local observables to include those operators that are scaling limits of $\Lambda(n)$-energy local observables where $\Lambda(n)$ is a function of $n$. 

The motivation for this modification again comes from the smeared operators $Y(a,f)$ where $f$ has \textit{infinite} Fourier series. Any product of these operators is defined on the VOA (\hyperref[thm4.8]{\textbf{Theorem 4.8}}), so it is possible that they form a strong SL-algebra. They are not energy local by themselves, but it is clear that the higher shift of energies happen with ever smaller magnitude which depends on the Fourier coefficients $|\hat{f}_n|$, a rapidly decaying sequence. 

There is also another motivation. In quantum computation, a space local operator is defined to be a sum of operators, each acting on no more than $O(\log(n))$ particles for a system with $n$ particles. But this is a discrete way of characterizing locality and equivalently, one could define space locality as an action that has exponential decay when one gets away from a specific particle. A similar picture exists for $Y(a,f)$s. The extent to which an operator can be called \textit{energy local} could therefore be more than just shifting the energy by a constant. But we need to keep in mind that no matter how one extends this definition, the algebraic structure has to be preserved under the scaling limit. This issue will be explored further for the Ising AC.

\subsubsection{Space-local observables}

\hfill \break 
Another property of the smeared operators is that they are considered to be space-local. In order to have a notion of space, some notion of adjacency for particles in $\mathcal{W}_n$ is needed. In the case of anyonic chains, the notion of space locality is clear.
\begin{dfn}\label{dfn11}
The $r$-\textit{space local} operators in ACs are a sequence of operators $O_n \in \mathcal{A}_n$ that are the sum of $r$-\textit{ultra space local} operators. An $r$-\textit{ultra space local} operator acts on $r$ many of adjacent particles. 
\end{dfn}
A typical example is the $3$-ultra space local operator $e_i$. Notice the difference between space-locality in our sense and locality in quantum computation. In quantum computation, a sequence of observables like $O_n=e_1e_{\lfloor \frac{n}{2} \rfloor} \in \mathcal{A}_n$ is considered to be local, while it is clearly not space-local. Therefore, space locality is a stronger locality than the one in quantum computation.

Still, the picture we hope to obtain for $Y(a,f)$ in finite settings is that of a quantum system with a large number of equidistant particles, and some \textit{ultra space local} operator $\widetilde{a}$, which is supposed to be the \textit{finite} version of $a$, applied with weight $f$ on each particle and constantly many of its close neighbors. Informally,
$$\sum_{j=1}^n f(e^{i\frac{2\pi j}{n}})\widetilde{a}_j,$$
will have the scaling limit $Y(a,f)$. This will be explored in \hyperref[4]{section 4}.

\section{Scaling limit of Ising anyonic chains}\label{3}

The main theorem of the section will be written in its entirety as a reference for the next sections. The proof will be given in the \hyperref[A.1]{appendix}. We shall use the notations in \hyperref[1.4]{section 1.4}, especially $(a,b)$ which will be used to denote the Hilbert space given by the anyonic chain with the two ends of the chain being $a$ and $b$.

\begin{thm}\label{thm3.1}
1- The following \textbf{strong} scaling limits hold, \textbf{up to some scalings} of the Hamiltonians (explained below)
\begin{enumerate}
    \item $\mathcal{W}_{n}=(\frac{1}{2},\frac{1}{2})$, $H_{n}=-\sum_{j=1}^{2n-1}e_j$. Then $(\mathcal{W}_{n},H_{n}) \xrightarrow{SL} (\rchi_0+\rchi_{\frac{1}{2}},L_0)$. 
    \item $\mathcal{W}_{n}=(0,0)$ or $(1,1)$, $H_{n}=-\sum_{j=2}^{2n-2}e_j$. Then $(\mathcal{W}_{n},H_{n})\xrightarrow{SL}(\rchi_{0},L_0)$. 
    \item $\mathcal{W}_{n}=(0,1)$ or $(1,0)$, $H_{n}=-\sum_{j=2}^{2n-2}e_j$. Then $(\mathcal{W}_{n},H_{n})\xrightarrow{SL}(\rchi_{\frac{1}{2}},L_0)$. 
    \item $\mathcal{W}_{n}=(\frac{1}{2},1)$ or $(\frac{1}{2},0)$, $H_{n}=-\sum_{j=1}^{2n-2}e_j$. Then $(\mathcal{W}_{n},H_{n})\xrightarrow{SL}(\rchi_{\frac{1}{16}},L_0)$. 
    \item $\mathcal{W}_{n}$ be the periodic chain of size $2n$, and $H_{n}=-\sum_{j=1}^{2n}e_j$. Then $$(\mathcal{W}_{n},H_{n})\xrightarrow{SL}(\rchi_0\overline{\rchi}_0+\rchi_{\frac{1}{2}}\overline{\rchi}_{\frac{1}{2}}+\rchi_{\frac{1}{16}}\overline{\rchi}_{\frac{1}{16}},L_0+\overline{L}_0)$$ if $n$ is even. 
\end{enumerate}
Furthermore, the rate of convergence of each scaling limit is $O(\frac{1}{n})$ while we have restriction of energies up to $O(\sqrt[3]{n})$. 

2- For the corresponding higher Virasoro generators action, with the same rate of convergence as above, given a fixed $m \neq 0$, we have (up to some scalings)
\begin{enumerate}
    \item $-\sum_{j=1}^{2n-1}\cos(\frac{m(j+\frac{1}{2})\pi}{2n+1})e_j\xrightarrow{SL}L_m+L_{-m}$, \\
    $i\sum_{j=1}^{2n-2}\sin(\frac{m(j+1)\pi}{2n+1})[e_j,e_{j+1}]\xrightarrow{SL}i(L_m-L_{-m})$
    
    \item $-\sum_{j=2}^{2n-2}\cos(\frac{m(j+\frac{1}{2})\pi}{2n-1})e_j\xrightarrow{SL}L_m+L_{-m}$, \\
    $i\sum_{j=2}^{2n-3}\sin(\frac{m(j+1)\pi}{2n-1})[e_j,e_{j+1}]\xrightarrow{SL}i(L_m-L_{-m})$
    
    \item $-\sum_{j=2}^{2n-2}\cos(\frac{m(j+\frac{1}{2})\pi}{2n-1})e_j\xrightarrow{SL}L_m+L_{-m}$, \\
    $i\sum_{j=2}^{2n-3}\sin(\frac{m(j+1)\pi}{2n-1})[e_j,e_{j+1}]\xrightarrow{SL}i(L_m-L_{-m})$
    
    \item $-\sum_{j=1}^{2n-2}\cos(\frac{m(j+\frac{1}{2})\pi}{2n})e_j\xrightarrow{SL}L_m+L_{-m}$,\\
    $i\sum_{j=1}^{2n-3}\sin(\frac{m(j+1)\pi}{2n})[e_j,e_{j+1}]\xrightarrow{SL}i(L_m-L_{-m})$
    
    \item    $-\sum_{j=1}^{2n}\cos(\frac{2m(j+\frac{1}{2})\pi}{2n})e_j\xrightarrow{SL}\mathbb{L}_m+\mathbb{L}_{-m}$ \\
    $i\sum_{j=1}^{2n}\sin(\frac{2m(j+1)\pi}{2n})[e_j,e_{j+1}]\xrightarrow{SL}i(\mathbb{L}_m-\mathbb{L}_{-m})$   

\end{enumerate}
If $m \le \sqrt[4]{n}$, we have a rate of convergence of $O(\frac{1}{n})$ for energies up to $\sqrt[4]{n}$.
\end{thm}
\textit{Notation and $\widetilde{L}_m^{c,s}$ identities.}\label{identities} For the Hamiltonians, assuming an $n$ which will always be obvious from the context, we choose the notation $\widetilde{L}_0^c$ as a \textit{scaling} of it, which has scaling limit $L_0$. The notations and scalings for the case 1(a), i.e. $\rchi_0+\rchi_\frac{1}{2}$, are
$$\widetilde{L}_0^c=\alpha_n^cH_n + \beta_n^{0,c}\textbf{1} \xrightarrow{SL} L_0,$$
where $\alpha_n^c=\frac{(2n+1)\sqrt{2}}{8\pi}$ and $\beta_n^{0,c} \in \mathbb{R}$. For the higher Virasoro generators, the first observable is $O_n^c$ (superscript $c$ because of $\cos$), and the second $O_n^s$ with the following similar notation and identities for 2(a)
$$\frac{\widetilde{L}_m^c+\widetilde{L}_{-m}^c}{2}=\alpha_n^cO_n^c+\beta_n^{m,c}\textbf{1} \xrightarrow{SL} \frac{L_m+L_{-m}}{2},$$
$$ \frac{i(\widetilde{L}_m^s-\widetilde{L}_{-m}^s)}{2}=\alpha_n^sO_n^s+\beta_n^{m,s}\textbf{1} \xrightarrow{SL} \frac{i(L_m-L_{-m})}{2},$$  
where $\alpha_n^s=\frac{(n+\frac{1}{2})(\sqrt{2})^2}{8\pi}$ and the scalars $\beta_n^{m,c},\beta_n^{m,s} \in \mathbb{R}$. Similarly for the full CFT, $\widetilde{\mathbb{L}}_m^c+\widetilde{\mathbb{L}}_{-m}^c$ and $i(\widetilde{\mathbb{L}}_m^s-\widetilde{\mathbb{L}}_{-m}^s)$ can be defined. It will turn out that such a splitting is possible so that $\widetilde{L}_{\pm m}^c$ and $\widetilde{L}_{\pm m}^s$, has scaling limit $L_{\pm m}$. The proof of the above theorem is provided in the \hyperref[A.1]{appendix} and one can easily recover the scaling factors by following the proof. We will only need the rate of growth of these scaling factors which will be at best $O(n^2)$ and $\alpha_n^s,\alpha_n^c$ do not depend on $m$ while $\beta_n^{m,c}$ and $\beta_n^{m,c}$ do.

\section{Scaling limit algebras in \texorpdfstring{$\overline{\mathcal{A}}$}{TEXT}}\label{4}

We would like to obtain the observables of each of these three types and prove they form an SL-algebra:
\begin{enumerate}
    \item Wightman's observables or smeared fields $Y(a,f)$,
    \item LCN observables $O \in \mathcal{A}(I)$.
    \item VOA observables or fields $Y(a,z)$,
\end{enumerate}
It is not hard to show that they are all in $\mathcal{A}$ as a vector space, i.e. all in a single framework. This fact tells us two things known before. First, that they are all physical as they describe some computable convergent sequence. And second, although they are all related and each one is believed to store all the information of the CFT by itself, by definition of scaling limit, they have to be in our set of observables simultaneously.

We will first obtain (a), as a result, recover the observables of (b), and lastly, some comments will be made on (c). The nonperiodic chains or in other words the chiral cases will be handled first. Due to its simplicity, only the case $\mathcal{V}=\rchi_0+\rchi_\frac{1}{2}$ will be analyzed, but all theorems can be similarly stated for the other chiral cases. At the end, there will be some comments on similar results for the full CFT.

\subsection{Wightman's observables}\label{4.1}

\subsubsection{Smeared vertex operator $Y(a,f)$}

\hfill \break
We will try to identify when hermitian observables of the form 
$$\eta\textbf{1}+\sum_j t_je_j \ \ \& \ \ \eta\textbf{1}+i\sum_j t_j [e_j,e_{j+1}]$$
that are already space-local, are also energy-local. A trigonometric interpolation of the $t_j$s with $\cos(\frac{m(j+\frac{1}{2})\pi}{2n+1})$ or $\sin(\frac{m(j+1)\pi}{2n+1})$ is performed. Afterwards, previous results can be used to write down the observable in terms of $(\widetilde{L}_m^c+\widetilde{L}_{-m}^c)$ or $i(\widetilde{L}_{m}^s-\widetilde{L}_{-m}^s)$, where $\widetilde{L}_{\pm m}^{c},\widetilde{L}_{\pm m}^{s}$ are the operators with scaling limit $L_{\pm m}$.

For the observable $O_n=\eta_n\textbf{1}+\sum_{j=1}^{2n-1} t_je_j$, a trigonometric interpolation using $\cos(\frac{m(j+\frac{1}{2})\pi}{2n+1})$ for $0 \le m \le 2n-2$ gives
$$t_j= \alpha_n^c\sum_{m=0}^{2n-2} a_m\cos(\frac{m(j+\frac{1}{2})\pi}{2n+1})$$         
$$\implies O_n=\gamma_n\textbf{1}+a_0\widetilde{L}_0^c+\sum_{m=1}^{2n-2} a_m\frac{\widetilde{L}_m^c+\widetilde{L}_{-m}^c}{2},$$
where $\gamma_n$ is some multiple of identity. Next, suppose $O_n$ does not shift the energy more than some given $\Lambda$.

An analysis of $\widetilde{L}_m^c$ formula given in \hyperref[eq32]{(32)} and \hyperref[eq33]{(33)}, shows two distinct parts
\begin{gather*}
\scalebox{0.99}{$\Big(\sum\limits_{k+m \le 2n} \cos\Big(\frac{(k+\frac{m}{2})\pi}{2n+1}\Big) \Psi_{k+m}\Psi_{k}^\dagger-\sum\limits_{k+m>2n} \cos\Big(\frac{(k+\frac{m}{2})\pi}{2n+1}\Big)\Psi_{2(2n+1)-k-m}\Psi_{k}^\dagger \Big).$}
\end{gather*}
The first part provides an energy shift of exactly $-m$. The second part provides an energy shift of $2(k-(2n+1))+m$ when $k+m > 2n+1$ (if $k+m=2n+1$, since $\Psi_{2n+1}=0$, that term is irrelevant). This energy shift is between $(-m,m)$ and it has the same parity as $m$. The same holds for $\widetilde{L}_{-m}$.

After the appropriate relabelling $\Psi_k \rightarrow \Psi_{\frac{n}{2}+1-k}$ (explained after \hyperref[eq17]{(17)}), the term $\Psi_{-(n+\frac{1}{2})}\Psi_{-(n+\frac{3}{2})}$ provides an energy shift of $-(2n-2)$ and it is only in $\widetilde{L}_{2n-2}^c$ due to the observation in the previous paragraph. Since $O_n$ is energy local $(1,O_n\Psi_{n-\frac{1}{2}}\Psi_{n-\frac{3}{2}}1) = 0$ implying $a_{2n-2}=0$.

It is easy to see how inductively each $a_m$ is zero; for $a_{2n-3}$, taking the term $\Psi_{-(n-\frac{1}{2})}\Psi_{-(n-\frac{5}{2})}$ leading us to the similar conclusion $a_{2n-3}=0$ and so on.

The case $\eta\textbf{1}+i\sum t_j[e_j,e_{j+1}]$ can also be done by using the trigonometric interpolation
$$t_j=\alpha_n^s \sum_{m=0}^{2n-3} b_m\sin(\frac{m(j+1)\pi}{2n+1}).$$
By mixing both cases $(\widetilde{L}_m^c+\widetilde{L}_{-m}^c)$ and $i(\widetilde{L}_{m}^s-\widetilde{L}_{-m}^s)$, we get
\begin{thm}\label{thm4.1}
$O_n$ is a $\Lambda$-energy local observable made from a linear combination of $e_j$ and $[e_j,e_{j+1}]$s and the identity if and only if it is of the form
$$O_n=\gamma_n \textbf{1}+ a_0\widetilde{L}_0+\sum_{m=1}^{\Lambda} \Big(a_m\widetilde{L}_m^c+ib_m\widetilde{L}_m^s\Big)+\sum_{m=1}^{\Lambda}\Big(a_m\widetilde{L}_{-m}^c-ib_m\widetilde{L}_{-m}^s\Big),$$
where $a_m,b_m \in \mathbb{R}$.
\end{thm}
\begin{rmk}\label{rmk4}
An operator $\widetilde{L}_m$ is desired which has scaling limit $L_m$ so that expressions like $\sum \hat{f}_m\widetilde{L}_m  \xrightarrow{SL} \sum \hat{f}_mL_m$ can be used where $\hat{f}_m=a_m+ib_m \in \mathbb{C}$. Dealing with $a_m\widetilde{L}_m^c+ib_m\widetilde{L}_{m}^s$ every time can become inefficient and the choice below resolves this issue
$$\widetilde{L}_m:=\Big(\frac{\widetilde{L}_m^c+\widetilde{L}_m^s}{2}+\frac{\widetilde{L}_{-m}^c-\widetilde{L}_{-m}^s}{2}\Big) \ \ \forall m \neq 0, \ \ \widetilde{L}_0=\widetilde{L}_0^c.$$
The above is a definition for an operator for which $\widetilde{L}_m \xrightarrow{SL} L_m$ and it satisfies the properties for convergence as it inherits those from the two operators. Indeed, $\frac{\widetilde{L}_{-m}^c-\widetilde{L}_{-m}^s}{2}$ when restricted to $\sqrt[4]{n}$ energy, will be an operator with a norm at most $O(\frac{1}{n})$ and so will become part of the error of the approximation. The rest of the operator acting on energy higher than $\sqrt[4]{n}$ will join that of $\frac{\widetilde{L}_m^c+\widetilde{L}_m^s}{2}$.  
\end{rmk}
\textit{Notation}. $O|_E$ denotes the restriction to energy at most $E$, i.e. $OP^E$, and $O|_{>E}:=O(\textbf{1}-P^E)$.

\textit{Notation}. From now on, $n$ will not be used for the virasoro mode operators, but for the sequence index which will be related to the size of the chain $2n+1$. For example
$$\widetilde{L}_m=L_m|_{\sqrt[4]{n}}+O(\frac{1}{m})+R_n^m,$$
where $R_m^n=\widetilde{L}_m|_{>\sqrt[4]{n}}$.

We can now state our first result for the scaling limit of observables.
\begin{thm}\label{thm4.2}
The energy local scaling limit of the sequence of hermitian observables $\mathcal{A}_n^H$ spanned by $e_j,i[e_j,e_{j+1}]$ and the identity as a real vector space is 
$$\{L(f)+\gamma\textbf{1} \ | \ f \ \text{has finite Fourier series},\ \gamma\in \mathbb{R} \}.$$
\end{thm}
\begin{proof}
One can remove the space local condition as all observables in $\mathcal{A}_n^H$ are space-local. Assume a sequence of $\Lambda$-energy local operators
$$O_n=\gamma_n\textbf{1}+\sum_{j=-\Lambda}^{\Lambda} \hat{f}_j^n\widetilde{L}_j,$$
where $\hat{f}_{-j}^n=\overline{\hat{f}_j^n}$ and $O_n \xrightarrow{SL} O$. To show that $O=L(f)+\gamma\textbf{1}$ for some function $f$ with finite Fourier series, restrict $O_n$ to some energy $M>2\Lambda$, 
$$O_n|_M=\gamma_n\textbf{1}+\sum_{j=-\Lambda}^{\Lambda}\hat{f}_j^n\widetilde{L}_j|_M.$$
According to the properties of $\widetilde{L}_j$s, for large enough $n$, 
$$O_n|_M=\gamma_n\textbf{1}+\sum_{j=-\Lambda}^{\Lambda}\hat{f}_j^nL_j|_M+\hat{f}_j^nO(\frac{1}{n}).$$
Since $O_n|_M$ has a limit in the operator norm to $O|_M$, $\hat{f}_j^n$s must have a limit. To prove that, we compute the inner product below for the vacuum $1$:
$$(L_{-\Lambda}|_M1,O_n|_M1)=f_{-\Lambda}^n\ ||L_{-\Lambda}1||+ (L_{-\Lambda}1,(\sum_j \hat{f}_j^nO(\frac{1}{n}))1)  \to (L_{-\Lambda}|_M1,O1),$$
where $|_M$ is dropped as it is no longer needed. Notice all the errors $O(\frac{1}{n})$ corresponding to $\widetilde{L}_j$ give at most $|j|$ energy shift. This mean only the errors corresponding to $\widetilde{L}_{\pm \Lambda}$ have to be handled
$$f_{-\Lambda}^n\ ||L_{-\Lambda}1||+(L_{-\Lambda}1,(f_\Lambda^nO(\frac{1}{n})+f_{-\Lambda}^nO(\frac{1}{n}))1).$$
$f_{-\Lambda}^n\ ||L_{-\Lambda}1||$ can be exactly computed and is of order  $f_{-\Lambda}^n \Lambda^\frac{3}{2}$. The rest can have norm at most $O(\frac{1}{n})|f_{-\Lambda}^n|$ as $f_\Lambda^n=\overline{f_{-\Lambda}^n}$. It is easy to see from here that in order for the above to have some limit,  $f_{-\Lambda}^n$ must have some limit $f_{-\Lambda}$.  

Next step is to subtract $f_{\Lambda}^n\widetilde{L}_\Lambda+\overline{f_{\Lambda}^n}\widetilde{L}_{-\Lambda}$ from $O_n$ and repeat the procedure. For the special case of $j=0$,  $\gamma_n\textbf{1}+f_0^n\widetilde{L}_0$ can be seen to give the same conclusion. Denoting $\lim_{n\to \infty} \hat{f}_j^n= \hat{f}_j, \lim_{n\to \infty} \gamma_n=\gamma$,
$$O=\gamma\textbf{1}+\sum_{j=-\Lambda}^\Lambda \hat{f}_jL_j.$$
\end{proof}
\begin{rmk}\label{rmk5}
By \hyperref[thm2.1]{\textbf{Theorem 2.1}}, we have a strong SL-algebra.
\end{rmk}
We would like to have our theorems as general as possible. For UMMs, higher level ACs \cite{gils2013anyonic} is conjectured to give the same results as in \hyperref[thm3.1]{\textbf{Theorem 3.1}}, implying the above theorem for UMMs. But a relaxed version of that theorem for UMMs would still give us the results in this section:
\begin{cnj}\label{cnj4.3}
For any UMM VOA $\mathcal{V}=\mathcal{V}_{c,0}$ and chiral representation $\mathcal{V}_{c,h}$, there is a sequence of quantum theories with \textbf{strong} scaling limit $(\mathcal{V}_{c,h},L_0)$ such that for each $L_m$, we have a sequence $\widetilde{L}_m \in \mathcal{A}_n$ with the following properties:
\begin{itemize}
    \item It is a space local observable with hermitian operators $a\widetilde{L}_m+\overline{a}\widetilde{L}_{-m} \in \mathcal{A}_n^H$.
    \item It shifts the energy no more than $|m|$.
    \item Restricted to energy at most $n^{d_\omega}$ it has the following approximation by $L_m|_{n^{
    d_\omega}}$ with the rest being $R_{n}^m$:
    $$\widetilde{L}_m=L_m|_{n^{d_\omega}}+O(\frac{1}{n^{g_\omega}})+R_{n}^m,$$
    where $d_\omega,g_\omega$ are positive constants.
    \item Its norm is bounded by $O(n^{e_\omega})$ for some constant $e_\omega$.
\end{itemize}
\end{cnj}
\begin{rmk}\label{rmk6}
It should be noted that the second and third item above have a meaning after the ``push-forward'' of the map $\widetilde{L}_m$ acting on $\mathcal{V}_{c,h}$ is assumed. This is done by the natural embedding $\rho_n: \mathcal{W}_n \hookrightarrow \mathcal{V}_{c,h}$ from the strong scaling limit and the map $\rho_n\widetilde{L}_m(\rho_n)^{-1}$ which acts on the copy of $\mathcal{W}_n$ inside $\mathcal{V}_{c,h}$ and extended by zero on the orthogonal complement. This ``push-forward'' will be implicitly assumed whenever it is necessary. Also, notice that this is not the ``natural'' embedding but it will work for our purposes (in Ising, the natural embedding, is described in the \hyperref[A.1]{appendix}).
\end{rmk}
\begin{rmk}\label{rmk7}
The last assumption is true for the Ising chain as $\widetilde{L}_m^{c}$ and $\widetilde{L}_m^s$ are after all a sum of $2n$ terms of $e_j$s which have norm order one. Taking the scaling factors $\alpha_n^c$ and $\beta_n^{m,c},\beta_n^{m,s}$ and their norm into account
$$||\widetilde{L}_m ||\le  O(n^2).$$
\end{rmk}
\begin{rmk}\label{rmk8}
Assuming the above conjecture, the \hyperref[thm4.2]{\textbf{Theorem 4.2}} is true for all UMMs with the exception that the statement should change to: the scaling limit of space energy local \textit{contains} the set $\{L(f)| \ f \ \text{finite Fourier series}\}$. Therefore, except for \hyperref[thm4.1]{\textbf{Theorem 4.1}}, \hyperref[thm4.2]{\textbf{Theorem 4.2}} and \hyperref[thm4.15]{\textbf{Theorem 4.15}}, all other theorems in sections \hyperref[4.1]{4.1} and \hyperref[4.3]{4.3} will hold the way they are stated for all UMMs. For all theorems in \hyperref[4.2]{section 4.2}, the stronger \hyperref[cnj5.5]{\textbf{Conjecture 5.5}} which tells us exactly how to recover the higher Virasoro modes for UMMS has to be assumed. 
The theorems below will be proved using the Ising AC, but by replacing some of the powers by appropriate constants ($d_\omega,$ etc), the results hold for UMMs assuming \hyperref[cnj4.3]{\textbf{Conjecture 4.3}}.
\end{rmk}
\begin{rmk}\label{rmk9}
It is conjectured that all the VOAs we care about (as described in \hyperref[1.1]{section 1.1}) satisfy energy boundedness \cite[Conjecture 8.18]{carpi2015vertex}. A generalization of the \hyperref[cnj4.3]{\textbf{Conjecture 4.3}} to all \textit{chiral} CFTs which satisfy energy boundedness is possible. Sequences in the same fashion of the Virasoro modes have to exist for all elements inside a minimal quasi-primary hermitian field generator set of the VOA. Then, all theorems in section \hyperref[4.1]{4.1} and \hyperref[4.3]{4.3} except \hyperref[thm4.1]{\textbf{Theorem 4.1}}, \hyperref[thm4.2]{\textbf{Theorem 4.2}} and \hyperref[thm4.15]{\textbf{Theorem 4.15}} can be recovered. In UMMs, the generator is only $\omega$ and in WZW models, the currents corresponding to the Lie algebra $\mathfrak{g}$ (see \cite{bondesan2015chiral} for a numerical demonstration and also for $W$-algebra currents see \cite{gainutdinov2014lattice}).
\end{rmk}

\textit{Notation}. Set $L(f)_{\le m}=\sum_{|j| \le m} \hat{f}_jL_j$ and similarly for $\widetilde{L}(f)$. Similarly define $L(f)_{> m}$ and $\widetilde{L}(f)_{> m}$. Also set
$$||f||_{s}^{\le E}=\sum_{|i|\le E} |\hat{f}_i|(|i|+1)^{s},$$
and 
$$|f|^{\le m}:=\sum_{|i| \le m}|\hat{f}_i|.$$
We wish to show that the choice of the ``natural'' sequence corresponding to $L(f)$ gives a strong SL-algebra. Some lemmas are needed.
\begin{lem}\label{lem4.4}
We have
$$\widetilde{L}(f):=\sum_{j=-\infty}^{\infty} \hat{f}_j\widetilde{L}_j \in \mathcal{A}_n^H, \ \ \text{for all } f \in C^\infty(S^1)$$
\end{lem}
\begin{proof}
Note that $\hat{f}_j$s are rapidly decreasing. Also, from \hyperref[rmk7]{\textit{Remark 7}},
\begin{align}\label{eq1}
||\widetilde{L}_j|| \le O(n^2).
\end{align}
The estimation does not depend on $j$. This gives an absolute convergence to an operator with norm bounded by $|f|O(n^2)$. On the other hand, for each $j$, we have $\hat{f}_j\widetilde{L}_j+\hat{f}_{-j}\widetilde{L}_{-j} \in \mathcal{A}_n^H$ implying $\widetilde{L}(f) \in \mathcal{A}_n^H$.
\end{proof}
The next step to establish a strong SL-algebra is to prove
\begin{lem}\label{lem4.5}
$\widetilde{L}(f) \xrightarrow{SL} L(f)$.
\end{lem}
\begin{proof}
The result shown here on the convergence behavior of $\widetilde{L}(f)$ will be useful in the next theorem. Take any $k \in \mathbb{N}$ and consider $N_{f,(10k)^3+1}$ for which $|\hat{f}_j|<\frac{1}{j^{(10k)^3+1}}$ for all $j>N_{f,(10k)^3+1}$. For $n$ large enough such that $\sqrt[10k]{n}>N_{f,(10k)^3+1}$, using \hyperref[eq1]{(1)},
\begin{align}\label{eq2}
||\widetilde{L}(f)_{>\sqrt[10k]{n}}||=||\sum_{|j|>\sqrt[10k]{n}} \hat{f}_j\widetilde{L}_j|| \le O(n^2)\sum_{|j|>\sqrt[10k]{n}} |\hat{f}_j|
\end{align}
$$\le O(n^2)\int_{\sqrt[10k]{n}}^\infty \frac{1}{x^{(10k)^3+1}} \text{d}x < O(n^2)\frac{(10k)^3+1}{n^{(10k)^2}}=O(n^{-(10k)^2+2}).$$
An same estimate for $L(f)_{>\sqrt[10k]{n}}$ via energy bounds is the next step:
$$||f||_\frac{3}{2}^{>\sqrt[10k]{n}}<2\sum_{j>\sqrt[10k]{n}}^{\infty} \frac{(j+1)^\frac{3}{2}}{j^{(10k)^3+1}}<\int_{\sqrt[10k]{n}}^\infty \frac{1}{x^{(10k)^3-10k+1}}=O(n^{-(10k)^2+1}),$$
therefore
\begin{align}\label{eq3}
||L(f)_{>\sqrt[10k]{n}}v|| <O(n^{-(10k)^2+1})||(L_0+\textbf{1})v||.
\end{align}
Next, given a vector $v \in \mathcal{V}$ and the embedding $\mathcal{W}_n \hookrightarrow \mathcal{V}$,
$$(\widetilde{L}(f)-L(f))v= (\widetilde{L}(f)_{\le \sqrt[10k]{n}}-L(f)_{\le \sqrt[10k]{n}}) v + (\widetilde{L}(f)_{>\sqrt[10k]{n}}-L(f)_{>\sqrt[10k]{n}}) v$$
The two estimations above imply that the second part vanishes. For the first part,
\begin{align}\label{eq4}
\widetilde{L}(f)_{\le \sqrt[10k]{n}}=L(f)_{\le \sqrt[10k]{n}}|_{\sqrt[4]{n}}+O\Big(\frac{|f|_{j \le \sqrt[10k]{n}}}{n} \Big)+R(f),
\end{align}
where $R(f)=\widetilde{L}(f)_{\le \sqrt[10k]{n}}|_{>\sqrt[4]{n}}$. Since $v$ has finite energy, for large enough $n$, $R(f)v=0$ and $L(f)_{\le \sqrt[10k]{n}}|_{\sqrt[4]{n}}v=L(f)_{\le \sqrt[10k]{n}}v$. This implies $||(\widetilde{L}(f)-L(f))v||\to 0,$ which is indeed a stronger result than $\widetilde{L}(f) \xrightarrow{SL} L(f)$.
\end{proof}
\begin{thm}\label{thm4.6}
The set $\{L(f) \ | \ f \in C^{\infty}(S^1)\}$ gives a strong SL-algebra with corresponding sequence $\widetilde{L}(f)$ to each $L(f)$.
\end{thm}
\begin{proof}
For the vacuum $1$ and $1_n=(\rho_n)^{-1}1$, the statement implies 
$$(1_n, \prod_{j=1}^k \widetilde{L}(f^{(j)}) 1_n) \to (1,\prod_{j=1}^k L(f^{(j)})1).$$
Proving the above is enough as this can be done similarly for any two vectors $u,v \in \mathcal{V}$. The fact that the right side is defined is shown in \cite[Lemma 3.2.1]{weiner2007conformal}. We will prove the above by using triangle inequality after estimating the \textit{intermediate} terms
$$ |(1,  \prod_{j=1}^{t-1} L(f^{(j)}) (L(f^{(t)})-\widetilde{L}(f^{(t)}))\prod_{j=t+1}^k \widetilde{L}(f^{(j)})1)|, \ \   1 \le t \le k,$$
where the embedding $\rho_n$ is used implicitly. For each $1\le j \le t$, $\widetilde{L}(f^{(j)})=\widetilde{L}(f^{(j)})_{\le \sqrt[10k]{n}}+\widetilde{L}(f^{(j)})_{> \sqrt[10k]{n}}$. Denote $y_t=\prod_{j=t+1}^k \widetilde{L}(f^{(j)})1$ and let $y_t=y_t^1+y_t^2$, where the first vector is inside $\mathcal{V}^{(k-t)\sqrt[10k]{n}} \subset \mathcal{V}^{k\sqrt[10k]{n}}$ of vectors with energies at most $k\sqrt[10k]{n}$, defined as
$$y_t^1=\prod_{j=t+1}^{k} \widetilde{L}(f^{(j)})_{\le \sqrt[10k]{n}}\  1.$$
To estimate the norm of $||y_t||$ and $||y_t^i||$s, the norm of the two operators decomposing $\widetilde{L}(f^{(j)})$ has to be bounded from above. Equation \hyperref[eq2]{(2)} gives $||\widetilde{L}(f^{(j)})_{> \sqrt[10k]{n}}|| <O(n^{-(10k)^2+2})$. As for $||\widetilde{L}(f^{(j)})_{\le \sqrt[10k]{n}}||$, there are two different estimations. One will be used to find an upper bound for $||y_t^1||$ and the other to bound $||y_t^2||$.

For $y_t^1$, as the product is applied on the vacuum, consider the restriction of each of those operators to energy $\le k\sqrt[10k]{n}$. By energy bounds
\begin{align}\label{eq5}
||\widetilde{L}(f^{(j)})_{\le \sqrt[10k]{n}}\Big|_{k\sqrt[10k]{n}}|| = ||L(f^{(j)})_{\le \sqrt[10k]{n}}\Big|_{k\sqrt[10k]{n}} + O\Big(\frac{|f^{(j)}|_{\le \sqrt[10k]{n}}}{n} \Big) || 
\end{align}
$$ \le 2\mathcal{C}_\omega||f^{(j)}||_{\frac{3}{2}}^{\le \sqrt[10k]{n}}(k\sqrt[10k]{n}+1)=O(\sqrt[10k]{n}) \ , \ \text{for large enough } n. $$
The second estimate is coming from \hyperref[eq1]{(1)}
\begin{align}\label{eq6}
||\widetilde{L}(f^{(j)})_{\le \sqrt[10k]{n}}|| \le O(n^2).
\end{align}
By using \hyperref[eq5]{(5)},
\begin{align}\label{eq7}
||y_t^1|| \le O((\sqrt[10k]{n})^{k-t})<O(\sqrt[10]{n}).
\end{align}
To estimate $||y_t^2||$, let us take the expansion of 
$$\prod_{j=t+1}^{k} \widetilde{L}(f^{(j)})1=\prod_{j=t+1}^{k} (\widetilde{L}(f^{(j)})_{\le \sqrt[10k]{n}}+\widetilde{L}(f^{(j)})_{> \sqrt[10k]{n}})1$$
and consider those terms that have at least one $\widetilde{L}(f^{(j)})_{> \sqrt[10k]{n}}$ in them. Those will be the ones contributing to $y_t^2$. Hence, as there are $2^{k-t}-1$ such terms,
\begin{align}\label{eq8}
||y_t^2|| < (2^{k-t}-1) O(n^{-(10k)^2+2})  O((n^2)^{k-t})) \le O(n^{-(10k)^2+2k+2}).
\end{align}
The estimates for $||y_t^1||,||y_t^2||$ give
\begin{align}\label{eq9}
||y_t|| < 2||y_t^1|| < O(\sqrt[10]{n}).
\end{align}
Let
$$x_t:=1^\dagger\prod_{j=1}^{t-1} L(f^{(j)})\ \ \& \ \ \max\limits_{t=1,\ldots,k}||x_t||=p \ \ \& \ \ \max\limits_{t=1,\ldots,k}||(L_0+\textbf{1})x_t^\dagger||=q.$$
It can be shown that (\cite[Lemma 3.2.1]{weiner2007conformal})
$$p=\max\limits_{t=1,\ldots,k}||x_t|| \le r||(L_0^k+\textbf{1})1||, $$
where $r$ depends on $f^{(j)}$s. $p$ depends on $k$, $f^{(j)}$s, and the degree of vector $v$ (which is chosen to be the vacuum here). Obviously, there is no dependence on $n$. We can derive a bound on $q$ using the above. Let us approximate
$$ |x_t(\widetilde{L}(f^{(t)})-L(f^{(t)}))(y_t^1+y_t^2)|.$$
Decomposing $(\widetilde{L}(f^{(t)})-L(f^{(t)}))$ as in \hyperref[lem4.5]{\textbf{Lemma 4.5}}:
$$(\widetilde{L}(f^{(t)})_{\le \sqrt[10k]{n}}-L(f^{(t)})_{\le \sqrt[10k]{n}})  + (\widetilde{L}(f^{(t)})_{>\sqrt[10k]{n}}-L(f^{(t)})_{>\sqrt[10k]{n}}). $$
For the second part, using the estimates \hyperref[eq9]{(9)} for $||y_t||$, \hyperref[eq2]{(2)} on $||\widetilde{L}(f^{(t)})_{>\sqrt[10k]{n}}||$, and finally \hyperref[eq3]{(3)} for $||x_tL(f^{(t)})_{>\sqrt[10k]{n}}||=||L(f^{(t)})_{>\sqrt[10k]{n}}x_t^\dagger||$,
$$ |x_t(\widetilde{L}(f^{(t)})_{>\sqrt[10k]{n}}-L(f^{(t)})_{>\sqrt[10k]{n}})y_t|<$$
$$pO(n^{-(10k)^2+2})O(\sqrt[10]{n})+ qO(n^{-(10k)^2+1})O(\sqrt[10]{n}) \xrightarrow{n \to \infty} 0.$$
For the first part, considering the approximation of $\widetilde{L}$s for energies up to $k\sqrt[10k]{n}$,
\begin{align}\label{eq10}
\widetilde{L}(f^{(t)})_{\le \sqrt[10k]{n}}-L(f^{(t)})_{\le \sqrt[10k]{n}}=
\end{align}
$$-L(f^{(t)})_{\le \sqrt[10k]{n}}\Big|_{>k\sqrt[10k]{n}}+O\Big(\frac{|f^{(t)}|_{\le \sqrt[10k]{n}}}{n} \Big)+R(f^{(t)}),$$
The first term $L(f^{(t)})_{\le \sqrt[10k]{n}}|_{>k\sqrt[10k]{n}}$ annihilates $y_t^1$ as the vector is inside $\mathcal{V}^{k\sqrt[10k]{n}}$. As for its action on $y_t^2$, instead of taking the norm of that multiplication, one can apply the energy bound on the left multiplication by $x_t$ and due to the smallness of $||y_t^2||$, it is easy to see that it vanishes when $n \to \infty$.

The second term, which is the only term where our approximation gets somewhat tight, when acting on $y_t$, has to compete with its norm. The estimation \hyperref[eq7]{(7)} tells us that the result is bounded by $O(\frac{\sqrt[10]{n}}{n})$ which still goes to zero.

Finally, the last term is the higher energy term $\widetilde{L}(f^{(t)})_{\le \sqrt[10k]{n}}|_{> k\sqrt[10k]{n}}$. When acting on $y_t^1$, this will give zero. Then, one can use the bound on the norm of $\widetilde{L}(f^{(t)})_{\le \sqrt[10k]{n}}$ (recall that this is a bounded operator like $\widetilde{L}(f^{(t)})$ with norm $O(n^2)$). As $||y_t^2||$ is much smaller, this will vanish as well.
\end{proof}
We note that $\mathcal{A}^H$ contains more than just the strong SL-algebra above:
\begin{thm}\label{thm4.7}
We have $\{L(f) \ | \ ||f||_{\frac{3}{2}}<\infty \} \subset \mathcal{A}^H$ which contains $\{L(f) | f \in C^{\infty}(S^1)\}$ as a maximal strong SL-algebra.
\end{thm}
\begin{proof}
For the maximality part, one has only to estimate the norm of $L_0^kL(f)1$ for any $k \in \mathbb{N}$. This would imply that the Fourier series of $f$ must be rapidly decreasing and therefore $f \in C^{\infty}(S^1)$.

It is also clear, by some analysis easier than \hyperref[lem4.5]{\textbf{Lemma 4.5}}, that any $L(f)$, with $||f||_{\frac{3}{2}}<\infty $, is obtainable as a sequence by choosing (e.g.) $O_n=\widetilde{L}(f)_{\le \log(n)}$. The rest was done in the previous theorem.
\end{proof}
The next theorem generalizes to all fields. We list three facts \cite{carpi2015vertex}
\begin{itemize}
    \item In a UMM, the descendants of $\omega$ span the VOA.
    \item Due to the Virasoro algebra identities, all descendants of $\omega$ can be obtained only by applying operators $L_n(n\ge-2)$.
    \item In a UMM, all fields are energy bounded:
    $$ ||Y(a,f)v|| \le C_{a}||f||_{r_a} ||(L_0+\textbf{1})^{s_a}v||$$
\end{itemize}
\begin{thm}\label{thm4.8}
$\{Y(a,f) | \ a \in \mathcal{V}, f \in C^\infty(S^1)\} \subset \mathcal{A}$ generates an SL-algebra.
\end{thm}
\begin{rmk}\label{rmk10}
The algebra is also local but in the quantum computation sense of locality (QC-locality) where product of a constant number of $e_i$s far apart from each other still counts as local. But it must be observed that if a field is space local, then its derivative ($Y(L_{-1}a,z)=[L_{-1},Y(a,z)]$) is also space local.
\end{rmk}
\begin{proof} It is not hard to show using the same approach in \cite[Lemma 3.2.1]{weiner2007conformal} that
$$||Y(a_1,f^{(1)})\ldots Y(a_k,f^{(k)})v|| \le \alpha||(L_0+\textbf{1})^{\sum s_{a_i}}v||,$$
where $\alpha$ depends on $f^{(j)}$s and $k$. Thus the set $Y(a,f)$ has all products defined on $\mathcal{V}$.
We will proceed by induction. 

Choose a basis with descendants. Then, for each field $L_{i_r}\ldots L_{i_1}\omega=a$ an induction will be performed on $r$. Hence, assume hypotheses have been shown to hold for the field $b$ and we wish to prove the same for $a=L_{-2}b$; the Borcherds identity shows that this is the hardest case and $L_rb$ for $r\ge -1$ are easier and will be described later.

For the field $a$, we want to obtain operators $\widetilde{\textbf{y}}_E(a)_m$ for energy $E$ and mode $m$ with the following hypotheses:

There exist $d_a$ such that for all $E \le n^{d_a}$ there are operators $\widetilde{\textbf{y}}_E(a)_m$ satisfying
\begin{itemize}
    \item $\widetilde{\textbf{y}}_E(a)_m$ is generated by the $e_i$s QC-locally; i.e there is some constant $p_a$ such that $\widetilde{\textbf{y}}_E(a)_m$ is $p_a$ QC-local and $p_a$ is independent of $m$ and $E$.
    \item $\exists v_a<d_a$ such that for any $m$ and $n^{v_a} \ge |m|$, $\widetilde{\textbf{y}}_E(a)_m|_E$ provides an energy shift at most $K_a(E+|m|)$ for some constant $K_a \ge 1$.
    \item $\widetilde{\textbf{y}}_E(a)_m$ has norm at most $n^{e_a}$ where $e_a$ depends on $a$.
    \item There exist a constant $g_a>0$ such that 
    \begin{align}\label{eq11}
    \widetilde{\textbf{y}}_E(a)_m=\textbf{y}(a)_m|_{E}+O(\frac{1}{n^{g_a}})+R_{E,n}^{a,m},\text{ for }n^{v_a} \ge |m|
    \end{align}
        where $\textbf{y}(a)_m|_{E}$ is the restriction of $\textbf{y}(a)_m$ in the VOA to energy at most $n^{d_a}$ but acting on $\mathcal{W}_n$ via pushback. $O(\frac{1}{n^{g_a}})$ should be regarded as the error in the approximation of $\widetilde{\textbf{y}}_E(a)_m|_{E}$ by ${\textbf{y}}(a)_m|_{E}$, and it has norm at most $O(\frac{1}{n^{g_a}})$. Finally, the last term is $R_{E,n}^{a,m}=\widetilde{\textbf{y}}_E(a)_m(\textbf{1}-P^E)$.
\end{itemize}
Notice the last hypothesis implies the same for restriction of energy to any $E' \le E$ since projection to energy $E'$ has norm at most $1$ and the rest will mix with $R_{E,n}^{a,m}$. Further, the base of induction $\omega$ is essentially done. For $E \le n^\frac{1}{4}$ and any $\sqrt[4]{n} \ge |m|$, as $\widetilde{\textbf{y}}_E(\omega)_m=\widetilde{L}_m$ provides an energy shift of at most $|m|$ for any mode $m$, in other words, at most $1 \times (E+|m|)$.

Suppose the hypotheses are true for $b$ and $a=L_{-2}b$. In \hyperref[1.1]{Borcherds identity}, putting $p=0$ and $q=-1$, and some index shifting gives
$${\textbf{y}}(a)_m={\textbf{y}}(L_{-2}b)_m=\sum_{j=0}^{\infty}(L_{-2-j}{\textbf{y}}(b)_{m+j+2}+{\textbf{y}}(b)_{m-j+1}L_{j-1}),$$
which is an infinite sum but when restricted to energy $E\le n^{d_a}$, where $d_a$ will be determined, the summation above will be finite and summed to some $j$. Indeed, as shown in \cite{carpi2015vertex}, consider the projection ${\textbf{y}}(a)_mP^{E}$. Then the first term is always zero when $E-(m+j+2)<0$ and the second term is always zero when $E-(j-1)<0$. So both are zero when $E<j+\max\{-1,m+2\}$. Hence
$${\textbf{y}}(a)_m|_E= \sum_{j=0}^{E-\max\{-1,m+2\}}(L_{-2-j}{\textbf{y}}(b)_{m+j+2}P^{E}+{\textbf{y}}(b)_{m-j+1}L_{j-1}P^{E}).$$
Putting redundant projections in the middle of the operators leads to ${\textbf{y}}(a)_m|_E=$
\begin{gather}\label{eq12}
\scalebox{0.86}{$\sum\limits_{j=0}^{E-\max\{-1,m+2\}}(L_{-2-j}|_{E+K_b(E+|m+j+2|)}{\textbf{y}}(b)_{m+j+2}|_E+{\textbf{y}}(b)_{m-j+1}|_{E+|j-1|}L_{j-1}|_{E}).$}
\end{gather}
This will be important as the last induction hypothesis for $\omega$ and $b$ will be applied separately. Based on the above identities, our choice for $\widetilde{\textbf{y}}_E(a)_m$ will be 
\begin{align}\label{eq13}
\widetilde{\textbf{y}}_E(a)_m=   \sum_{j=0}^{E-\max\{-1,m+2\}}(\widetilde{L}_{-2-j}\widetilde{\textbf{y}}_{E}(b)_{m+j+2}+\widetilde{\textbf{y}}_{E+|j-1|}(b)_{m-j+1}\widetilde{L}_{j-1})
\end{align}
where we recall that $\widetilde{\textbf{y}}_E(\omega)_j=\widetilde{L}_j$ for all $E$. 

The first hypothesis obviously holds as $p_a\le p_w+p_b$ for all $m$. One can easily see why only QC-locality can be proved.

For the second hypothesis, assume $v_a<d_a$ and smaller than $v_b$ ($<d_b$). This allows us to apply the hypothesis on $\widetilde{\textbf{y}}_{E}(b)_{m+j+2}$, i.e. to have
$$n^{v_b} \ge (2n^{v_a}+n^{d_a}+4) \ge (|m|+|j|+2).$$
This will be one of the restrictions on $v_a$ and $d_a$. At the end of the argument, choosing $v_a<d_a << d_b,g_b,v_b$ will be shown to be enough. Given $E \le n^{d_a}$, the hypothesis for $b$ implies that $\widetilde{\textbf{y}}_E(b)_{m+j+2}$ and $\widetilde{\textbf{y}}_{E+|j-1|}(b)_{m-j+1}$ will provide an energy shift at most  $K_b(|m+j+2|+E)$ and $K_b(|m-j+1|+E+|j-1|))$ which added to the energy shift of the Virasoro operators is at most  $|j+2|+K_b(|m+j+2|+E)$ and $K_b(|m-j+1|+E+|j-1|)+|j-1|$. Since $|j| \le E+|m|+2$, there is a constant $K_a$ such that the energy shift is at most $K_a(E+|m|)$.

The third hypothesis is easy to check as this rough estimate for all $m$ holds
$$|\widetilde{\textbf{y}}_E(a)_m| < n^{d_a+v_a+3}n^{e_b+e_\omega}.$$ 
Implying that $e_a$ can be chosen $d_a+v_a+e_b+e_\omega+3$.

For the last hypothesis, from the equations \hyperref[eq12]{(12)} and \hyperref[eq13]{(13)}, the sum of the approximations given by the hypothesis for $\omega$ and $b$ is ${\textbf{y}}(a)_m|_E$. But what about the other terms? Consider $\widetilde{L}_{-2-j}\widetilde{\textbf{y}}_{E}(b)_{m+j+2}$ which is the product
$$(L_{-2-j}|_{E+K_b(E+|m+j+2|)}+O(\frac{1}{n^{g_\omega}})+R_{E+K_b(E+|m+j+2|),n}^\omega)$$
$$({\textbf{y}}(b)_{m+j+2}|_{E}+O(\frac{1}{n^{g_b}})+R_{E,n}^b),$$
where the superscript for $R$s indicating the mode is dropped as it is clear from the context. In each parenthesis, the first two terms are $\widetilde{L}_{-2-j}|_{E+K_b(E+|m+j+2|)}$ and $\widetilde{\textbf{y}}_{E}(b)_{m+j+2}|_{E}$ respectively. The term that contributes to ${\textbf{y}}(a)_m|_E$ is precisely
$$L_{-2-j}|_{E+K_b(E+|m+j+2|)}{\textbf{y}}(b)_{m+j+2}|_{E}.$$
What contributes to $R_{E,n}^{a,m}$ is also clear
$$(L_{-2-j}|_{E+K_b(E+|m+j+2|)}+O(\frac{1}{n^{g_\omega}})+R_{E+K_b(E+|m+j+2|),n}^\omega)R_{E,n}^b,$$
which is indeed an operator with restriction to energy higher than $E$. 
Further, due to energy restrictions
$$R_{E+K_b(E+|m+j+2|),n}^\omega ({\textbf{y}}(b)_{m+j+2}|_{E}+O(\frac{1}{n^{g_b}}))=0.$$
The only terms remaining should contribute to $O(\frac{1}{n^{g_a}})$:
\begin{itemize}
    \item $O(\frac{1}{n^{g_\omega}})O(\frac{1}{n^{g_b}})$
    \item $O(\frac{1}{n^{g_\omega}}){\textbf{y}}(b)_{m+j+2}|_{E}$
    \item $L_{-2-j}|_{E+K_b(E+|m+j+2|)}O(\frac{1}{n^{g_b}})$
\end{itemize}
We need to show that while choosing $g_a$ appropriately. Also, in addition to this analysis, one has to analyze the product $\widetilde{\textbf{y}}_{E+|j-1|}(b)_{m-j+1}\widetilde{L}_{j-1}$ for all $0 \le j \le E-\max\{-1,m+2\}$.  But $E \le n^{d_a}$ and further, $|m| \le n^{v_a}$. It is not hard to see that although there are so many terms contributing to what should be $O(\frac{1}{n^{g_a}})$, by choosing $v_a,d_a$ small enough compared to $g_b,g_\omega$, the approximation will be of the form $O(\frac{1}{n^{g_a}})$ with norm at most $O(\frac{1}{n^{g_a}})$.

The first term $O(\frac{1}{n^{g_\omega}})O(\frac{1}{n^{g_b}})$ is $O(\frac{1}{n^{g_a+g_\omega}})$ with the obvious bounded norm. As mentioned in the last paragraph, this will show up many times and so, there is one restriction here on $d_a,v_a$.

As for the second term, ${\textbf{y}}(b)_{m+j+2}|_{E}$ is a bounded operator with norm at most 
$$C_{b}(|m+2+j|+1)^{r_b}(E+1)^{s_b} \le C_{b}(n^{d_a}+2n^{v_a}+4+1)^{r_b}(n^{d_a}+1)^{s_b}.$$
In other words, the norm is bounded by
$$O(\frac{(n^{d_a}+2n^{v_a}+5)^{r_b}(n^{d_a}+1)^{s_b}}{n^{g_\omega}}).$$
Hence small enough $d_a,v_a$ can deliver the desired result.
The story for the third term is similar, energy bound is used for $L_{-2-j}|_{E+K_b(E+|m+j+2|)}$ 
$$C_{b}(|2+j|+1)^{\frac{3}{2}}(E+K_b(E+|m+j+2|)+1) \le$$
$$C_{b}(n^{d_a}+n^{v_a}+4+1)^{\frac{3}{2}}(n^{d_a}+K_b(2n^{d_a}+2n^{v_a}+4)+1).$$
The norm is bounded by
$$O(\frac{(n^{d_a}+n^{v_a}+5)^{\frac{3}{2}}(n^{d_a}+K_b(2n^{d_a}+2n^{v_a}+4)+1)}{n^{g_b}})$$
which is another restriction on how small $d_a,v_a$ have to be.

One can handle the product $\widetilde{\textbf{y}}_{E+|j-1|}(b)_{m-j+1}\widetilde{L}_{j-1}$ in a similar way. The induction is finished when $a=L_{-2}b$.

If $a=L_rb$ for any $r>-2$, then the Borcherds identity would be
$$(L_{r}b)_m=\sum\limits_{j=0}^{\infty} (-1)^j \binom{r+1}{j}\Big(L_{r-j}b_{m+j}-(-1)^{r+1}b_{r+1+m-j}L_{j-1}  \Big)$$
and the treatment of this case is easier since the summation is finite; for $j>r+1\ge0$ we have $\binom{r+1}{j}=0$. All the properties described in the induction can be proved here as well. Therefore, the induction is fully proved.

It remains to show that $\{Y(a,f) | a \in \mathcal{V}, f \in C^\infty(S^1)\}$ generates an SL-algebra. Using the properties in the induction hypotheses, it can be seen that the proof is nothing but a more involved version of \hyperref[thm4.6]{\textbf{Theorem 4.6}}.

To get the product $\prod_{j=1}^k Y(a_j,f^{(j)})$ in the scaling limit, the operators $$\widetilde{Y}_{E_j}(a_j,f^{(j)})=\sum_m \hat{f}^{(j)}_m\widetilde{\textbf{y}}_{E_j}(a_j)_m \in \mathcal{A}_n$$
have to be chosen where $E_j$s need to be determined carefully by taking into account the constants in the energy bound inequalities for all $a_j$s, and also all other constants, notably $d_{a_j}$s and $g_{a_j}$s, so that we can use the approximation provided by the last hypothesis. It is clear that the choice of $E_j$s will not be universal and depends on the product. They will also not be equal due to the second hypothesis and will be very small compared to all other constants.
\end{proof}

\begin{rmk}\label{rmk11}
The reason we could not obtain smeared fields as a \textit{strong} SL-algebra generating set is the dependence of the energy shift on the energy itself. If somehow all vectors were obtained by only applying $L_r, r > -2$ (because of the finite sum) or if we knew that the base of induction $\widetilde{L}_m$ shifts the energy \textbf{exactly} by $m$, this issue would not be present.
\end{rmk}
\begin{rmk}\label{rmk12}
One would wish to get the hermitian fields giving self-adjoint $Y(a,f)$, as a scaling limit of hermitian observables generated by the $e_i$s. Descendants of even degree of $\omega$ are hermitian if and only if they are quasi-primary. Also, as \cite{carpi2015vertex} demonstrates, quasi-primary hermitian fields \textit{generate} (not span) any unitary VOA. Further, we could not find any exact formula or general description of these fields. But \textbf{a generating set} of quasi-primary hermitian fields can exist which have a corresponding sequence coming from (a generating set formed by hermitian observables, i.e.) $\mathcal{A}_n^H$. For UMMs, that generating set is $\{\omega\}$, for which there is a corresponding sequence from $\mathcal{A}_n^H$.
\end{rmk}

\subsection{Local conformal nets observables}\label{4.2}

\hfill \break
In this section, bounded operators in the LCN framework are recovered. Recall that for UMMs, the observables algebra on an interval $I$ is given by $\{e^{iL(f)}| \ \text{supp}(f) \subset I \}''$ \cite{kawahigashi2004classification}.

From results of the previous section, the following is immediate
\begin{cor}\label{cor4.9} 
The sequence of observables below give a strong SL-algebra:
$$e^{i\widetilde{L}(f)} \xrightarrow{SL} e^{iL(f)}.$$
\end{cor}
\begin{proof}
This is a direct application of the Trotter-Kato approximation theorem (see e.g. \cite{engel2006short}) on \hyperref[thm4.6]{\textbf{Theorem 4.6}}. The fact that the scaling limit is a strong SL-algebra is simply due to the uniform boundedness of the operators involved, all being unitary.
\end{proof}

\begin{cor}\label{cor4.10}
All operators in $\{e^{iL(f)}\}''$ are in $\mathcal{A}$ giving a strong SL-algebra.
\end{cor}
\begin{proof}
As all algebras here are generated by self-adjoint operators, we will be considering only self-adjoint operators (this will make applying Kaplansky's density theorem easier). Consider a sequence of self-adjoint operators $O^{(i)}$ in the algebra generated by $e^{iL(f)}$s with a strong limit to a self-adjoint bounded operator $O$. Each $O^{(i)}$ has a corresponding sequence of self-adjoint $(O^{(i)}_n)_n$ with scaling limit $O^{(i)}$ which can be thought of replacing any $e^{iL(f)}$ in $O^{(i)}$'s expression by $e^{i\widetilde{L}(f)}$. From these sequences, by a standard diagonal argument, one can get a sequence $O_n$ with scaling limit $O$.

As long as $O^{(i)}$s are uniformly norm bounded, there is the possibility of having a sequence $O_n$ that is uniformly bounded, giving a strong SL-algebra as in the previous theorem. This includes the case where $O$ is in the \textbf{norm-operator} closure of the algebra generated by $\{e^{iL(f)}\}$. So the $C^*$-algebra can be recovered. Then, Kaplansky's density theorem does the rest: one can apply it on the sequence $O^{(i)}$, such that it becomes uniformly bounded by $||O||$ and then apply the same theorem on each sequence associated to $O^{(i)}$ so that they become in turn uniformly bounded by $||O^{(i)}||$.
\end{proof}

Hence, all observables in LCN form a strong SL-algebra. The next question is whether there exist some definition of the algebra $\mathcal{A}_n(I)$ and how the \textit{bounded} scaling limit would compare to the LCN, called $\mathcal{A}_{lcn}(I)$. As we shall see, the anyons must be on the upper half-circle as in \hyperref[fig1]{Figure 1}.
\begin{dfn}\label{dfn12}
Consider the upper half-circle $S^1_+$ with its two points on the boundary. The set of intervals $\mathcal{I}_+$ are the connected sets in one of the following forms:
\begin{itemize}
    \item Open intervals $I$ inside $S^1_+$ for which $\partial I \cap \partial S^1_+=\emptyset$,
    \item Closed-open intervals $I$ where $|I \cap \partial S^1_+| = |\partial I \cap \partial S^1_+|= 1$,
    \item $S^1_+$.
\end{itemize}
\end{dfn}
On these sets, the following nets of observables are defined
\begin{dfn}\label{dfn13}
Given $I \in \mathcal{I}_+$, $\mathcal{A}_n(I)$ is generated by $e_j$s where $[\frac{j\pi}{2n+1},\frac{(j+1)\pi}{2n+1}] \in I$ and the identity. 
\end{dfn}

The definition can be seen to imply $[\mathcal{A}_n(I_1),\mathcal{A}_n(I_2)]=\{0\}$ which is locality. Isotony is obvious, i.e. $I_1 \subset I_2 \implies \mathcal{A}_n(I_1) \subset \mathcal{A}_n(I_2)$.
\begin{dfn}\label{dfn14}
Consider the set of self-adjoint \textit{bounded} linear operators $O$ in the scaling limit of the algebra of observables $\mathcal{A}_n(I)$ such that there exist a self-adjoint sequence $ O_n \in \mathcal{A}_n(I)$ with bounded norm and
$$\rho_n(O_n(\rho_n)^{-1}(u)) \rightarrow Ou,\ \ \ \forall u \in \mathcal{V},$$
i.e. there is sequence with \textit{strong SL} convergence to $O$ or the strong-operator convergence in $\mathcal{V}$. Define $\mathcal{A}_b(I)$ as the von Neumann algebra generated by the set.
\end{dfn}
Locality is the reason behind the above definition. Consider two sequences of operators $x_n \xrightarrow{SL} x$ and $y_n \xrightarrow{SL} y$ which are self-adjoint and commuting. In order to ensure $[x,y]=0$, it can be easily observed that the weak-limit offered by scaling limit is not enough and we need at least a strong type of that limit (which is the above definition). But that could not be enough as $x_ny_n\xi \to xy\xi$ for $\xi \in \mathcal{V}$ can \textit{not} be necessarily true yet:
$$(x_ny_n-xy)\xi=x_n(y_n-y)\xi+(x_n-x)y\xi.$$
The first and second part of the above summation are not guaranteed to go to zero unless $x_n$s are uniformly bounded and $x_n \to x$ in the strong-operator topology (of $\overline{\mathcal{V}}$ as $y\xi \in \overline{\mathcal{V}}$). It turns out that the strong SL convergence (which is strong-operator convergence in $\mathcal{V}$) \textbf{and} norm boundedness are in some way equivalent to convergence in the strong-operator topology (in $\overline{\mathcal{V}}$). One direction is clear and the other is the application of Kaplansky's density theorem to get such a sequence with norms uniformly bounded. The definition above imposes these properties and $\mathcal{A}_b(I)$ can be seen to satisfy locality and isotony. In fact similar to the procedure carried out in \hyperref[cor4.10]{\textbf{Corollary 4.10}}, it can be seen to be have a sequence associated to any of its elements which are norm bounded and converge strongly to that element. Therefore, it is a strong SL-algebra.

How does this ``net'' compare to $\mathcal{A}_{lcn}(I)$? Denote by $j(I)$ the reflection of the interval $I$ with respect to the $x$-axis where $j:z \rightarrow \bar{z}$.
\begin{thm}\label{thm4.11}
Given a function $f=\sum \hat{f}_me^{im\theta} \in C^\infty(S^1)$ with $\text{supp}(f) \subset I \cup j(I)$, and $\hat{f}_m=a_m+ib_m$, define $\widetilde{e}(f)=$
\begin{gather*}
\scalebox{0.9}{$\alpha_n^c\sum\limits_{j=1}^{2n-1} f_c\Big(\frac{\pi(j+\frac{1}{2})}{2n+1}\Big)e_j +i\alpha_n^s \sum\limits_{j=1}^{2n-2}f_s\Big(\frac{\pi(j+1)}{2n+1}\Big)[e_{j},e_{j+1}]+(\sum\limits_{m=-\infty}^\infty a_m\beta_n^{m,c}+b_m\beta_n^{m,s})\textbf{1}$}
\end{gather*}
which is inside $\mathcal{A}_n(I)$ (for large enough $n$), and
$$f_c(\theta)=\frac{f(\theta)+f(-\theta)}{2} \in C^\infty(S^1_+), \ \ \forall e^{i\theta} \in S^1_+,$$
$$f_s(\theta)=\frac{f(\theta)-f(-\theta)}{2} \in C^\infty(S^1_+), \ \ \forall e^{i\theta} \in S^1_+.$$
we have $\widetilde{e}(f) \xrightarrow{SL} L(f)$.
\end{thm}
\begin{proof}
In fact, $\widetilde{e}(f)=\widetilde{L}(f)$ and this implies the theorem (using \hyperref[thm4.6]{\textbf{Theorem 4.6}}). To show that equality, the formula for $\widetilde{L}_m$ gives
$$\hat{f}_m\widetilde{L}_m+\hat{f}_{-m}\widetilde{L}_{-m}=a_m\widetilde{L}_m^c+ib_m\widetilde{L}_m^s+a_{-m}\widetilde{L}_{-m}^c+ib_{-m}\widetilde{L}_{-m}^s$$
where $\hat{f}_m=\overline{\hat{f}_{-m}}$. Next, the identities for $\widetilde{L}_m^{c,s}$ gives $\widetilde{L}(f)=$
$$\alpha_n^c\sum_{j=1}^{2n-1} c_je_j + i\alpha_n^s\sum_{j=1}^{2n-2}s_j[e_{j},e_{j+1}]+(\sum_{m=-\infty}^{\infty} a_m\beta_n^{m,c}+b_m\beta_n^{m,s})\textbf{1},$$
where 
$$c_j=\sum_{m=-\infty}^{\infty}a_m\cos(\frac{m(j+\frac{1}{2})\pi}{2n+1}), s_j=\sum_{m=-\infty}^{\infty}b_m\sin(\frac{m(j+1)\pi}{2n+1}).$$
But $f_c(\theta)$ and $f_s(\theta)$ are precisely the $\cos()$ and $\sin()$ part of the Fourier series of $f$. Therefore, the above is precisely $\widetilde{e}(f)$. 
\end{proof}
As a corollary, by definition,
\begin{cor} \label{cor4.12}
$\{e^{iL(f)}\}'' \subset \mathcal{A}_b(I)$ for $\text{supp}(f) \subset I \cup j(I)$.
\end{cor}
This hints to the relation between $\mathcal{A}_b(I)$ and $\mathcal{A}_{lcn}$. Assume $I$ touches the boundary of upper half-circle. Then, $I \cup j(I)$ is some connected interval in the circle and so $\{e^{iL(f)}|\text{supp}(f) \subset I \cup j(I)\}'' = \mathcal{A}_{lcn}(I \cup j(I))$. By the corollary above,
$$\mathcal{A}_{lcn}(I \cup j(I)) \subset \mathcal{A}_b(I).$$
But due to Haag duality for the conformal net $\mathcal{A}_{lcn}$ and locality for $\mathcal{A}_b$, for the complement of $I$, called $J$, in $S^1_+$,
$$\mathcal{A}_b(J) \subset \mathcal{A}_b(I)' \subset \mathcal{A}_{lcn}(I \cup j(I))'=\mathcal{A}_{lcn}(J \cup j(J)) \subset \mathcal{A}_b(J) .$$
Therefore, one recovers exactly, no more and no less, the LCN by taking the bounded scaling limit.
\begin{thm}\label{thm4.13}
$\mathcal{A}_b(I)=\mathcal{A}_{lcn}(I \cup j(I))$ for $I \in \mathcal{I}_+$ with $|I \cap \partial S^1_+|=1$.
\end{thm}
The above theorem is true for all UMMs assuming \hyperref[cnj4.3]{\textbf{Conjecture 4.3}}.

\subsection{Vertex operators \texorpdfstring{$Y(a,z)$}{TEXT}}\label{4.3}

\hfill \break 
In \hyperref[thm4.8]{\textbf{Theorem 4.8}}, ${\textbf{y}}(a)_m$ was found to be in the scaling limit using QC-local operator. Therefore, $Y(a,z)$ should also be in the scaling limit as an almost linear operator. In fact, $Y(a,z)$ is the weak limit of a sequence $Y(a,f)$ where $f$ shrinks to the $\delta$ Dirac function. Then, according to \hyperref[rmk3]{\textit{Remark 3}}, $Y(a,z) \in \mathcal{A}$. Here, we wish to construct a concrete sequence for the observable.
\begin{thm} \label{thm4.14}
$Y(a,z) \in \mathcal{A}$ as an almost linear operator.
\end{thm}
\begin{proof}
Choose the sequence of observables
$$O_n=\sum_{|m| < \frac{g_a\log(n)}{2\log(|z|)}} \widetilde{\textbf{y}}_{\log(n)}(a)_m z^{-m-\text{wt }a}$$
where $\text{wt }a$ is the degree of $a$. Take $u,v \in \mathcal{V}$. Without loss of generality, assume $u,v$ are homogeneous with weight difference $-s$. For $(u,Y(a,z)v)$, only the term $(u,{\textbf{y}}(a)_sz^{-s-\text{wt }a}v)$ is nonzero. Thus, 
$$(u,\sum_{s< |m| < \frac{g_a\log(n)}{2\log(|z|)}} \widetilde{\textbf{y}}_{\log(n)}(a)_m z^{-m-\text{wt }a}v)$$
must go to zero. Once this is proved, the rest is the sum 
$$(u,\sum_{ |m| <s} \widetilde{\textbf{y}}_{\log(n)}(a)_m z^{-m-\text{wt }a}v)$$
is a finite sum of operators for which the scaling limit is known. By \hyperref[thm4.8]{\textbf{Theorem 4.8}}, we need to compute
$$\Big(u, ({\textbf{y}}(a)_m|_{\log(n)}+O(\frac{1}{n^{g_a}})+R_{\log(n),n}^a)z^{-m-\text{wt }a}v\Big).$$
For $n$ where $\log(n)>\text{wt }v$, the term $R_{\log(n),n}^av$ is zero. Similarly for $m \neq s$, the term $(u,{\textbf{y}}(a)_m|_{\log(n)}v)$ is zero. It remains to show
$$ \lim\limits_{n \to \infty} O(\frac{1}{n^{g_a}})(\sum_{s<|m|<\frac{g_a\log(n)}{2\log(|z|)}}z^{-m}) = 0.$$
Due to symmetry of the summation, assume $|z|> 1$, and the summation is not small only for positive powers. In that case, the summation has norm at most $\frac{|z|^{\frac{g_a\log(n)}{2\log(|z|)}}+1}{|z|-1}$ which vanishes when divided by $n^{g_a}$.
\end{proof}

We discussed the following intuition on the scaling limit of smeared field
$$\int f(e^{i\frac{\pi j}{n}})e_j \xrightarrow{SL} Y(\omega,f)=\oint Y(\omega,z)f(z)z^2 \frac{\text{d}z}{2\pi i z}.$$
Informally, one could think of this smooth function being a Gaussian distribution which goes to the $\delta$ Dirac function at some point corresponding to angle $\theta$. In that case, one would expect to get
$$e_\theta \xrightarrow{SL} Y(\omega,e^{i\theta}).$$
Of course, with $e_j$s, the ``$\cos()$'' part appears in the scaling limit. For the other part, the bracket $[e_j,e_{j+1}]$ must be used.

Finding some ultra local operator in $\mathcal{A}_n$ giving us the field operator in the scaling limit would be a ``proof'' that the field operator $Y(w,z)$ should not only be called a local observable, but an ultra local observable. Unfortunately, the \textit{natural} guess does not work.

\textit{Notations}.  $\widetilde{v}_x^c$ is the vector with entries $(\cos(\frac{m(x+\frac{1}{2})\pi}{2n+1}))_{0\le m\le 2n}$ and $\widetilde{v}_x^s=(\sin(\frac{m(x+1)\pi}{2n+1}))_{1\le m\le 2n}$. Define $\beta_n^c$ as the infinite vector with entries $\beta_n^{m,c}$ for all $m\ge 0$ and similarly define $\beta_n^s$. Extend $\widetilde{v}_x^c$ and $\widetilde{v}_x^s$ by zeros to have infinite entries for them as well.
\begin{thm}\label{thm4.15}
We do \textbf{not} have 
$$O_n=\alpha_n^c||\widetilde{v}_x^c||^2e_x+i\alpha_n^s||\widetilde{v}_x^s||^2[e_x,e_{x+1}]+(\beta_n^c.\widetilde{v}_x^c+\beta_n^s.\widetilde{v}_x^s)\textbf{1} \xrightarrow{SL} Y(\omega,z)z^2,$$
where $z=e^{i\theta}$ and we pick the unique $1 \le x\le 2n-1$ such that $\theta \in [\frac{x\pi}{2n+1},\frac{(x+1)\pi}{2n+1}]$.
\end{thm}
Notice sometimes $x$ can only be chosen for large enough $n$ as $\theta$ may be close to the boundaries. $O_n$ is exactly the expression for $e_x$ and $[e_x,e_{x+1}]$ one obtains by considering the \hyperref[identities]{\textbf{$\widetilde{L}_m^{c,s}$ identities}} of \hyperref[thm3.1]{\textbf{Theorem 3.1}}. That is why we believe this should be the first candidate for convergence to $Y(\omega,z)z^2$. 
\begin{proof}
By the \hyperref[identities]{\textbf{$\widetilde{L}_m^{c,s}$ identities}},
$$O_n=\sum_{m=0}^{2n}\cos(\frac{m(x+\frac{1}{2})\pi}{2n+1})\frac{\widetilde{L}_m^c+\widetilde{L}_{-m}^c}{2}+i\sum_{m=1}^{2n}\sin(\frac{m(x+1)\pi}{2n+1})\frac{\widetilde{L}_m^s-\widetilde{L}_{-m}^s}{2}.$$
In other words, the following should not hold
$$O_n=\sum_{m=-2n}^{2n}\Big(\cos(\frac{m(x+\frac{1}{2})\pi}{2n+1})+i\sin(\frac{m(x+1)\pi}{2n+1})\Big)\widetilde{L}_m \xrightarrow{SL} \sum e^{im\theta}L_m.$$
It is clear that any finite sum up to some $M$ for $O_n$ goes to $\sum_{|m| \le M}e^{im\theta}L_m$. Restrict to some finite energy $M$ from right \textbf{and} left. Notice the scaling limit is not supposed to be a linear operator so restriction needs to be made from both sides. It will be shown that the approximations to high Virasoro modes give something other than zero. This will be shown for the $\cos()$ part (in other words, the real part) of the summation which is provided by the operators $\widetilde{L}_m^c+\widetilde{L}_{-m}^c$. The $\sin()$ part  $\widetilde{L}_m^s+\widetilde{L}_{-m}^s$ (complex part) can be done similarly.

One can easily observe from the formula of $\widetilde{L}_m^c+\widetilde{L}_{-m}^c$ (\hyperref[eq32]{(32)} and \hyperref[eq33]{(33)}), that for $M<m<2n+1-2M$, the restriction from left and right is exactly zero for large enough $n$. This is easy to observe by considering the picture of the half-circle having the momenta on it.

Now notice that for $2n+1-2M\le |m| \le 2n$, the formula for $\widetilde{L}_m^c+\widetilde{L}_{-m}^c$ gives many fermion pairs which are distinct for different $m$. Indeed, after the restrictions, any term $\Psi_{k'}\Psi_k$ should have both $n+1-M \le k,k' \le n+M$ and as $\Psi_k$ appears only once for each $k$, there are around $O(M)$ possibilities. Also, the difference between $k$ and $k'$ for each one of these pairs is exactly $2n+1-m$ implying that the pairs are not repeated by different $m$s. Each of these terms will have a coefficient of order $O(n)$; Indeed there is a scaling provided by the $\alpha_n^c$ and further, there is also the coefficient $2\cos(\frac{(k+\frac{m}{2})\pi}{2n+1})$ which is close to $2$ as $\frac{\frac{m}{2}\pi}{2n+1}$ is close to $\frac{\pi}{2}$ and $n+1-M \le k \le n+M$. Therefore, although the part $|m|\le M$ of the summation cause no problem, there are terms corresponding to $2n+1-2M\le |m| \le 2n$ that blow up in norm. Hence, the scaling limit is certainly not an almost linear operator which is defined on $\mathcal{V}\times \mathcal{V}$ and the $\cos()$ part of $Y(\omega,z)z^2$ is not the scaling limit.

For the $\sin()$ and $\widetilde{L}_m^s$ case, with a similar reasoning, the $2n+1-2M\le |m| \le 2n$ part of the summation does not blow up in norm but it is non-zero and due to the divergent coefficient $\sin(\frac{m(x+1)\pi}{2n+1})$, does not have a limit and the $\sin()$ part of $Y(\omega,z)z^2$ is not the scaling limit.
\end{proof}

\begin{rmk}\label{rmk13}
As all problems emerge from the $m$ close to the both ends of the summation, i.e. $-2n$ and $2n$, one could take the observable $O_n$ as follows with the desired scaling limit
$$\sum_{m=-2n+\log(n)}^{2n-\log(n)}\Big(\cos(\frac{m(x+\frac{1}{2})\pi}{2n+1})+i\sin(\frac{m(x+1)\pi}{2n+1})\Big)\widetilde{L}_m \xrightarrow{SL} \sum e^{im\theta}L_m.$$
In fact, any function $f(n) \xrightarrow{n \to \infty} \infty$ instead of $\log(n)$ would work to avoid the discussed issues. But this is not the nice ultra local expression in terms of just $e_x$ and $[e_x,e_{x+1}]$ we desired for $O_n$. It is unknown whether there \textit{exists} ultra local observable $O_n$ with $O_n \xrightarrow{SL} Y(\omega,z)z^2$.
\end{rmk}

\section{Conjectures and future directions}\label{5}

\subsection{On scaling limit of anyonic chains}\label{5.1}

\hfill \break
In this section, we provide a list of problems that need to be addressed for a clearer picture of the structures in the scaling limit relevant to CFT.

After \hyperref[dfn6]{\textbf{Definition 6}}, it was asked whether or not $\mathcal{A}^H$ \textit{generates} $\mathcal{A}$ and whether that in turn \textit{generates} $\overline{\mathcal{A}}$. If the scaling limit is a CFT, such a conjecture becomes reasonable:
\begin{cnj}\label{cnj5.1}
The observables $\overline{\mathcal{A}}$ are ``generated'' by some means, like closure with respect to some topology, from the set of observables of the form $Y(a,f)$.
\end{cnj}

The smeared fields $Y(a,f)$ were obtained as QC-local operators.
\begin{cnj}\label{cnj5.2}
There is a spanning set $\mathcal{S}=\{a\}_{a \in \mathcal{V}}$ of the VOA such that for any $a \in \mathcal{S}$, the smeared field $Y(a,f)$ is space-local.
\end{cnj}

One obstacle to space locality is the \textit{absence} of commutators in the Borcherds identity. Otherwise, the terms involving products of far apart $e_j$s would disappear. If $Y(a,f)$ can be expressed in terms of commutators of $Y(\omega,f)$ with the Virasoro generators, the above conjecture would be true. Therefore, the obstacle may just be some simple lemma that is missing.

Closely related to conformal invariance, is the scaling limit of the product of unitaries $e^{i\widetilde{L}(f)}$ and the smeared field operators $\widetilde{L}(f)$.
\begin{cnj}\label{cnj5.3}
Prove that the algebras in \hyperref[cor4.10]{\textbf{Corollary 4.10}} and \hyperref[thm4.6]{\textbf{Theorem 4.6}} together generate a (strong) SL-algebra.
\end{cnj}

As a remark on the emergence of conformal invariance in the scaling limit, notice that due to \hyperref[cor4.9]{\textbf{Corollary 4.9}}, for $e^{iL(g)} \in \mathcal{A}_{lcn}(I)$, we have
$$e^{i\widetilde{L}(f)}e^{i\widetilde{L}(g)}e^{-i\widetilde{L}(f)} \xrightarrow{SL} e^{iL(f)}e^{iL(g)}e^{-iL(f)}.$$
Further, the scaling limit above is itself in $\mathcal{A}_{lcn}(\exp{(f)}(I))$ and it is this fact that is the \textit{reason} for conformal covariance (expressed below) in the LCN as $e^{-iL(f)}$s generate the local algebras:
$$U(\gamma)\mathcal{A}_{lcn}(I)U(\gamma)^\dagger=\mathcal{A}_{lcn}(\gamma(I)), \ \ \gamma \in \text{Diff}_+(S^1).$$
But it is easy to show that $\mathcal{A}_n(I)$ can also be generated by $e^{i\widetilde{L}(f)}$s. Due to \hyperref[cor4.9]{\textbf{Corollary 4.9}} and \hyperref[thm4.13]{\textbf{Theorem 4.13}}, this essentially implies that, loosely speaking, the two sets $e^{i\widetilde{L}(f)}\mathcal{A}_n(I)e^{-i\widetilde{L}(f)}$ and $\mathcal{A}_n(\exp{(f)}(I))$ become the same in the scaling limit (at least for $I$ and $\exp{(f)}(I)$ satisfying condition of \hyperref[thm4.13]{\textbf{Theorem 4.13}}). Therefore, conformal invariance emerges in the scaling limit. This may not be satisfying as it is not clear whether the group of operators $e^{i\widetilde{L}(f)}$ is the \textit{natural} choice for the group (\textit{sequence} of groups acting on the anyonic chains) that should recover the action of $\text{Diff}_+(S^1)$ in the scaling limit (see \cite{jones16} for a different candidate, the Thompson's group).

We could not get all types of observables as algebra. The point-like field operators $Y(a,z)$ were just given as a vector space in the scaling limit. 

\begin{cnj}\label{cnj5.4}
Field operators $Y(a,z)$ form a (strong) SL-algebra whenever their product is defined on $\mathcal{V} \times \mathcal{V}$.
\end{cnj}
The techniques used in \cite{konig2017matrix} may be useful to prove this as it involves some kind of truncations of the field operators. Finally, it would be useful for quantum simulation if one could obtain the field operators as ultra local operators because the product of ultra local hermitian operators can be simulated efficiently \cite[Claim 6.2]{aharonov2007polynomial} on a quantum computer.

Due to the numerical results on higher level ACs, it was conjectured (\hyperref[cnj4.3]{\textbf{4.3)}} that the theorems of \hyperref[4]{section 4} are true for higher level UMMs. Here, we emphasize the identities giving us the Virasoro algebra 
\begin{cnj}\label{cnj5.5}
The Hilbert space and the Virasoro algebra action of every chiral UMM with central charge $c=1-\frac{6}{(k+1)(k+2)}$ is obtainable as a scaling limit of some $SU(2)_k$ AC with some suitable boundary condition and the same theorems proved in \hyperref[4]{section 4} hold for them. Also denoting by $O_n^c$ and $O_n^s$ the following
$$O_n^c=-\sum_{j=1}^{2n-1}\cos(\frac{m(j+\frac{1}{2})\pi}{2n+1})e_j, \ \ O_n^s=i\sum_{j=1}^{2n-2}\sin(\frac{m(j+1)\pi}{2n+1})[e_j,e_{j+1}],$$
we have operators $\widetilde{L}_{\pm m}^{c},\widetilde{L}_{\pm m}^{s} \xrightarrow{SL} L_{\pm m}$ satisfying the properties in \hyperref[cnj4.3]{\textbf{4.3}} and,
$$\frac{\widetilde{L}_m^c+\widetilde{L}_{-m}^c}{2}=\alpha_n^cO_n^c+\beta_n^{m,c}\textbf{1} \xrightarrow{SL} \frac{L_m+L_{-m}}{2},$$
$$ \frac{i(\widetilde{L}_m^s-\widetilde{L}_{-m}^s)}{2}=\alpha_n^sO_n^c+\beta_n^{m,s}\textbf{1} \xrightarrow{SL} \frac{i(L_m-L_{-m})}{2},$$ 
where $\alpha_n^c,\alpha_n^s,\beta_n^{m,c},$ and $\beta_n^{m,s}$ are suitable scaling factors.
\end{cnj}
\begin{rmk}\label{rmk14}
In \cite[(7.5)]{koo1994representations}, there is a very similar conjecture although different. These will be compared in \hyperref[A.2]{appendix subsection 2}.
\end{rmk}
\begin{rmk}\label{rmk15}
The size of the chain was assumed to be $2n-1$. The chain size depends on the boundary condition which needs to be adjusted accordingly. 
\end{rmk}
One could then define the operators $\widetilde{L}_m$s as done before. Proving the above by direct diagonalization (as done for non-interacting theories), seems to be hard. As mentioned in \hyperref[1.6]{section 1.6}, one could hope to consider the commutators of the above observables and show that they satisfy similar relations as the Virasoro algebra. But note that by taking commutators, terms appear with non-vanishing norm and yet, what should be, vanishing in the scaling limit. These terms make it harder to recover the Virasoro algebra relations.

\subsection{Intertwiners and full CFTs}\label{5.2}

\hfill \break
Consider a rational VOA $\mathcal{V}$ with irreducible modules $A,B,C$ and corresponding conformal weights $h_A,h_B,h_C$. An \textit{intertwiner of type} $\binom{C}{A \ B}$ is a map 
$$\mathscr{Y}(\cdot,z): A \to \text{End}(B,C)[[z,z^{-1}]], \ \ \mathscr{Y}(a,z)=\sum_{n \in \mathbb{Z}}a_mz^{-\tau-m},$$ 
where $\tau=h_A+h_B-h_C$. It has the following notation for homogeneous $a \in A_k$
$$\mathscr{Y}(a,z)=\sum_{n\in \mathbb{Z}}\textbf{y}(a)_nz^{-n-k-\tau},$$
and it satisfies similar axioms as the vertex operator. 

Intertwiners are part of the fundamental features of a CFT as they describe the fusion rules. As an example, the fusion rules in \hyperref[1.2]{section 1.2} for the Ising model correspond to three different free fermionic fields
$$\psi_\frac{1}{2}^0(z)=\sum_{n \in \mathbb{Z}} z^{-(n-1)}\Psi_{n-\frac{1}{2}} $$
$$\psi_\frac{1}{2}^0(z)^\dagger=\psi_0^\frac{1}{2}(z)=\sum_{n \in \mathbb{Z}} z^{-n}\Psi_{n-\frac{1}{2}}$$
$$\psi_\frac{1}{16}^\frac{1}{16}(z)=\sum_{n \in \mathbb{Z}} z^{-n}\Psi_n,$$
where $\psi_i^j : \rchi_i \to \rchi_j$. Therefore, the natural question is how they emerge in the scaling limit and what the right framework to discuss them in finite settings should be. As an example, we could define some set of maps from $\mathcal{W}_n^1 \xrightarrow{SL} \rchi_\frac{1}{2}$ to $\mathcal{W}_n^2 \xrightarrow{SL} \rchi_0$ as the real vector space of observables 
$$\mathcal{A}^{h,\frac{1}{2} \to 0}_n=\{\sum_{m  =-n+\frac{1}{2}}^{n-\frac{1}{2}} \hat{f}_{m-1}\Psi_{m-\frac{1}{2}} \ |  \ \overline{\hat{f}_m}=\hat{f}_{-m} \in \mathbb{C} \},$$
and then define $\mathcal{A}^{i \to j}_n$, the algebra generated by odd numbers of Majorana operators (as in \hyperref[A.1]{Appendix subsection 1}), and also its closure $\overline{\mathcal{A}^{i \to j}_n}$. Still, this example clearly requires a very good idea of the finite versions of the primary fields. This is a hard problem in general (see \cite{mong2014parafermionic} as an example).

Although we will not provide an answer for this question,  independent of what the framework should be, one can try to find the fermionic free fields in the scaling limit as a QC-local operator. There are three types of free fermionic field $\psi_i^j$ as described in \hyperref[1.2]{section 1.2}.  As an example, one could use the basis provided in \cite[section 5]{bravyi2002fermionic} to localize the modes of the field $\psi_\frac{1}{2}^0(z)+\psi_0^\frac{1}{2}(z)$ on Hilbert spaces $\mathcal{W}_n \xrightarrow{SL} \rchi_0+\rchi_\frac{1}{2}$. Then, using the Borcherds identity (for intertwiners),
$$\mathscr{Y}(a,f)=\oint \mathscr{Y}(a,f)z^{k}\frac{\text{d}z}{2\pi i z}, \text{ where it can be seen that }\tau=0, $$
can be obtained similar to what was done in \hyperref[thm4.8]{\textbf{Theorem 4.8}}. 

Another important subject we did not discuss, was finding the algebra of observables for full CFTs. First, note that the scaling limit of the periodic anyonic chain should be assumed to be the full CFT on the torus. This changes the definition of the expectation values that one needs to measure in order to claim that a set generates an SL-algebra. The correlation functions for the torus are traces taken over the whole Hilbert space \cite{konig2017matrix}. The second issue is the presence of \textit{interchiral} observables \cite[section 6]{gainutdinov2016associative} for which there is no counterpart in the VOA picture. These observables can be obtained by using $\sum \sin()e_j$ and $\sum \cos()[e_j,e_{j+1}]$ as shown in the \hyperref[A.2]{appendix}. 

Still, with similar techniques as in \hyperref[4]{section 4}, it can be shown that the trace of the observables in finite spaces corresponding to the conformal field and the interchiral observable converge to what we expect:
$$\text{Tr}_\mathcal{V}((\mathbb{L}(f^{(1)})+\overline{\mathbb{L}}(g^{(1)}))\ldots (\mathbb{L}(f^{(k)})+\overline{\mathbb{L}}(g^{(k)}))r^{\mathbb{L}_0}), \ \ 0<r<1,$$
where $r$ is the diameter of the torus, $f^{(i)},g^{(i)} \in C^\infty(S^1)$, $\mathbb{L}(f^{(i)})$ is the smeared field for the conformal field and $\overline{\mathbb{L}}(g^{(i)})$ is the smeared field for the interchiral observable.

\subsection{Simulation of CFTs by quantum computers}\label{5.3}

\hfill \break 
The motivation of our work was an efficient quantum simulation of CFTs. What insights have we gained from this work? The first step is to define the problems that we want to solve. 

In each case, there will be \textbf{local} observables for which we ask their expectation values to be efficiently computed in polynomial time with respect to the inputs. By computing, we always mean approximating in polynomial time up to an error inverse polynomial with respect to the inputs. 

Informally speaking, we could also say that we are simulating efficiently (some of) the local observables themselves. Therefore, it is important to find out what those \textit{efficiently simulatable} local observables are. 

For example, in quantum computation, in the context of many problems like the simulation of the unitary evolution \cite{lloyd1996universal}, the \textit{efficient} $k$-Local Hamiltonians are a sum of \textit{polynomially} many $k$-ultra local operators.

When it comes to locality, quantum computation has its own precise definition. A fundamental aspect of the definition of locality is that explicitly or implicitly, there is a \textbf{sequence} of operators $O_n$ that are the sum of \textit{ultra-}local operators acting on at most $O(1)$ many particles. In this way, one could distinguish between local and nonlocal (sequence of) operators. This idea can not be applied directly to CFTs, simply because a CFT is a \textbf{single} Hilbert space with no sequences attached naturally to any observable. Since we believe that CFTs must be efficiently simulatable by a quantum computer, there are two paths for defining locality in CFTs. The first is to declare a subset of observables in CFTs to be ultralocal depending on what kinds of problems one wants to solve, and then show that one can simulate them efficiently using a quantum computer. The other, which is less problem oriented, is to associate a sequence that ``quantifies'' locality (just like in finite dimensions where we have $2$-local, $3$-local, etc.) for each observable.

In this work, the second strategy has been followed since the beginning. This strategy, as shown in \hyperref[4]{section 4}, by using a rigorous definition of locality inspired by quantum computation, demonstrates how locality in finite dimensions translates into that of infinite dimension. It also reconciles to a great extent with what mathematical physicists and physicists think of the notion of local observables, although not being exactly the same.

But there is no guarantee that this is the right path for the problems ahead and in fact, we will also point out the disadvantages of the anyonic chain approach when it comes to tackling these problems.

\subsubsection{Unitary evolution of CFTs}

\hfill \break
The first problem is the unitary evolution of CFTs. In TQFTs, using the functorial approach, this problem has been shown to be in \textbf{BQP} \cite{simulation02}. The important observation made is that the unitary evolution is a representation of the mapping class group and the mapping class groups is generated by braids and Dehn twists. Those are operators for which one can have a local expression. We seek the same picture in CFTs.

In CFTs, using the functorial approach, the unitary evolution is guided by unitary maps called $U(\gamma)$ and simulating
$$|(1,U(\gamma)1)|$$
is the goal, where $1$ is the vacuum and $U$ is a positive-energy projective unitary representation of $\text{Diff}_+(S^1)$ with $\gamma$ a diffeomorphism in the Lie group. It is well-known (\cite{goodman1985projective}) that the representation $U$ corresponds to a unitary positive energy representation of the Virasoro algebra. By a result of \cite{goodman1985projective}, simulating the above quantity is the same as simulating
\begin{align}\label{eq14}
|(1,\prod_{j=1}^k e^{iL(f^{(j)})}1)|,
\end{align}
where $f^{(j)} \in C^\infty(S^1)$ and $\gamma=\exp{(f^{(1)})} \circ \cdots \circ \exp{(f^{(k)})}$. Loosely speaking, if in TQFT, the \textit{complexity} of an evolution (a cobordism in the functorial point of view) arises at ultra local locations where braids happen, in CFT due to the continuous picture, one has to look at the diffeomorphism $\gamma$ infinitesimally, hence the decomposition of $U(\gamma)$ to finite products of $e^{iL(f)}$s.

Theoretically, the above quantity is obviously computable as long as functions $f^{(j)}$s are computable. But when there is the issue of efficiency, one needs to make sure to ask the \textit{right} question. Ideally speaking, one needs to know what nature does efficiently and ask whether a quantum computer can do that efficiently as well. In other words, what is the set $\mathcal{F}$ of operators $e^{iL(f)}$ that can be considered in \hyperref[eq14]{(14)}? One important observation is that in the same problem for other theories (TQFTs or usual quantum computation), the analog of the set $\mathcal{F}$ has always been given by a ``local generating'' set. In this case, the natural candidates are the $L_n$s and the fact that they are scaling limit of sums of $e_i$s, which are themselves the generators in the similar TQFT problem, is another evidence.

For example, the operator $e^{if_0L_0}$ which is the evolution by the Hamiltonian corresponding to the constant function $f \equiv f_0$ is certainly one of the operators in $\mathcal{F}$. And in general the Virasoro operators $L_n$ are thought to be local and $e^{i(\hat{f}_nL_n+\hat{f}_{-n}L_{-n})}$ corresponding to the function $f=\hat{f}_ne^{in\theta}+\hat{f}_{-n}e^{-in\theta}$ must be in $\mathcal{F}$. Therefore, it is reasonable to ask a finite combination of these to be simulated efficiently. This means $e^{iL(f)} \in \mathcal{F}$ for $f$ having finite Fourier series. 

The next question is which functions with infinite Fourier series can also be considered for the simulation problem. An analogy in quantum computation, would be to think of a hermitian matrix $H$ that may be nonlocal and acts on certain qubits but the norm of its action has an exponential decay away from those qubits. This translates to a unitary operator which is nonlocal but has an action exponentially close to identity except in some \textit{centers} of action. The Fourier coefficients $\hat{f}_n$ are rapidly decaying 
$$ \forall k , \exists N_k \text{  such that   } \forall |n| \ge N_k \implies |\hat{f}_n| \le \frac{1}{n^k}$$
but the rate of this decay or equivalently, what the rate of growth of $N_k$ should be is unclear. Perhaps, an exponentially decaying $\hat{f}_n$ or a polynomial growth for $N_k$ is the answer. Finding the exact form of dependence of the rate of convergence of the scaling limit in \hyperref[cor4.10]{\textbf{Corollary 4.10}} on $N_k$ will help to answer this question.

So far, we can safely assume that the set $\mathcal{F}$ has all operators corresponding to functions with finite Fourier series. We have the following definition for the CFT unitary evolution problem.
\begin{dfn}\label{dfn15}
(\textit{CFT UNITARY EVOLUTION)} Consider functions $f^{(1)},\ldots,f^{(k)}$ with finite Fourier series and coefficients nonzero up to $n_1,n_2,\ldots,n_k$ all given as inputs, find an approximation up to given error $\epsilon$, of the following quantity
$$|(1,\prod_{j=1}^k e^{iL(f^{(j)})}1)|$$
\end{dfn}
The conjecture in the same spirit of TQFT, is
\begin{cnj}\label{cnj5.6}
CFT UNITARY EVOLUTION is in \textbf{BQP}, i.e. there is a polynomial time quantum algorithm with respect to the inputs, namely $\{n_j\}_j \cup \{(\hat{f^{(j)}})_l\}_{j,l} \cup \{k,\frac{1}{\epsilon}\}$. Generically, the problem is \textbf{BQP}-complete.
\end{cnj}

It is worth mentioning that in contrast to TQFT, where unitary evolution on the vacuum is trivial (as it is a one dimensional space in the case of the sphere), in CFT due to the existence of descendents provided by $L_{-m}$ for $m>1$, we have a nontrivial problem in an infinite dimensional Hilbert space being the vacuum sector of the highest-weight representation with central charge $c$. Of course, one can generalize the above to other sectors and this is part of the next problem.

The AC approach provides evidence for which operators have to be in $\mathcal{F}$. It also provides insights as to what operators in the scaling limit can be called \textit{local}, which is a fundamental aspect of a theory from a computational point of view. Further, it gives a discretized picture of what unitary evolution looks like in a CFT. Consider many particles on a chain where it is allowed to have fusion between nearby particles with certain penalties for the undesired (nontrivial) fusion. The value of the penalties is what gives the function $f$. If we let the system evolve in this setting, the unitary evolution guided by those constraints is $e^{i\widetilde{L}(f)}$. If $f$ is the constant function, then it is the usual Hamiltonian.

But there is no guarantee that the AC approach is the right one for this question. In fact, in order to be able to approximate \hyperref[eq14]{(14)} using ACs, a proof for which AC gives the VOA in the scaling limit is needed. Furthermore, a bigger obstacle could be proving the \textbf{BQP}-completeness, as the expression for approximating the operators $e^{iL(f)}$ are exponentials of weighted sum of \textbf{all} $e_i$s and it is hard to build specific unitary operators using these. Lastly, one needs to prepare the vacuum which could be a hard problem (see \cite{murg2004adiabatic} for the case of Ising).

There are other possible approaches to this problem. First, one could look for combinatorial realization of the action of Virasoro generators. Path representations of the states of the Hilbert space and how the Virasoro algebra acts on these paths can be analyzed. We refer to \cite{kaufmann1995path} for nonunitary models minimal models $M(2,q)$ with $q$ odd where the action of every Virasoro generator is obtained, and \cite{feverati2003critical} for unitary minimal models where actions of higher Virasoro generators is not known yet. Another possible approach to prove that the problem is in \textbf{BQP}, would be to use large enough tensor power of the Fock space which contains many interacting models as subtheories (including all $\mathfrak{su}(2)_k$ WZW models and all minimal models; see \cite[section 4]{tener2016geometric} for a list). This free theory is essentially $(\rchi_0+\rchi_\frac{1}{2})^2$ and can be modelled using an ultra local realization of the Dirac operators $\Psi_k,\overline{\Psi}_k$s (as localized in \cite{bravyi2002fermionic}). We could then derive a local expression for the Virasoro generators of any subtheory using the Dirac operators. Then, by some energy truncation and taking the scaling limit, one should compute the rate of convergence and show that it is inverse polynomial with respect to the energy truncation. This could provide a faster convergence than AC; indeed, CFT has quantized energy but continuous spacetime. This approach focuses on the energy local degrees of freedom as the local basis for quantum computation, instead of the space local degrees of freedom as in AC. A third approach would be to first derive an exact expression for the quantity $|(1,U(\gamma)1)|$, and then try to simulate it. This exact expression can be obtained for all free models \cite[see Theorem 6.2.3 and section 7]{neretin1996categories} but no such \textit{closed formula} is known for higher minimal models. We hope to pursue these in future works.

\subsubsection{Correlation functions of CFT}
\hfill \break 
The second problem is in fact a generalization of the first. It is the simulation of the correlation functions of CFT which is very much like in TQFT and simulation of Jones polynomial \cite{tqc03} (which is the value (\cite{witten1989quantum}) of the TQFT correlation function); $2n$ fields are inserted, denoted by $n$ cups (inserting two fields dual to each other) as a TL diagram, a unitary evolution is applied and the probability of getting back to same state is measured (denoted by $n$ caps). The normalization in the case of Jones polynomial is also the norm of the state given by the $n$ cups. The direct analogy in CFT would be (not a chiral but) a full diagonal CFT where the cup inserting the dual pair is $\mathbb{Y}(a\otimes a',z,\overline{z})$ inserting the pair $a \in A$ and $a' \in A'$ (the contragredient module of $A$). But as will shown, it is not obvious that the similar quantity can be defined. Here, we make an attempt for the definition. The proof that the quantity is well-defined is not given. In future works, we will aim to fill the gaps in the arguments. Let us start with the chiral formulation.

We will assume a nice VOA: unitary, CFT-type, rational (including $C_2$-co-finite). The goal would be to approximate the following point-like chiral correlation function $C_{\text{chiral}}^{p}$ efficiently:
$$C_{\text{chiral}}^{p}=\frac{|\Big(\mathscr{Y}_n(a_n,\gamma(z_n))\ldots\mathscr{Y}_1(a_1,\gamma(z_1))1,U(\gamma)\mathscr{Y}_n(a_n,z_n)\ldots\mathscr{Y}_1(a_1,z_1)1\Big)|}{||\mathscr{Y}_n(a_n,z_n)\ldots\mathscr{Y}_1(a_1,z_1)1||. ||\mathscr{Y}_n(a_n,\gamma(z_n))\ldots\mathscr{Y}_1(a_1,\gamma(z_1))1||},$$
where $|z_i|=1$ are on the unit circle arranged as $0<\arg(z_1)<\ldots<\arg(z_n)<\pi$, and $a_i\in A_i$ are primary fields in the irreducible modules and $\mathscr{Y}_i$ are of type $\binom{B_{i-1}}{A_i \ B_{i}}$ with irreducible modules $B_i$. We note that fields inserted at $z_i$ move to $\gamma(z_i)$ by conformal covariance of primary fields and therefore, this is \textit{where} we should measure the amplitude of getting back the same configuration.

There are multiple \textit{issues} with this definition. First, it is not clear that the numerator or denominator exist. For the numerator to exist, it makes sense to impose the condition of intertwiners having energy bounds. Indeed, notice that one can first perform the evolution by $U(\gamma)$ and then insert the fields at $\gamma(z_i)$ (due to conformal covariance). As $U(\gamma)$ operates inside the common domain $\cap_{k=1}^\infty \mathcal{D}((L_0+1)^k)$ called \textit{smooth vectors} and denoted by $\mathcal{V}^\infty$ or $B^\infty$ for module $B$, the vector $U(\gamma)1$ would be inside $\mathcal{V}^\infty$ and generally not inside $\mathcal{V}$. Even with this condition one needs to prove that the expectation value for those insertion points exist which brings us to the second issue.

The insertion points have the same norm and they are \textbf{not} distinct. To see this, by taking the adjoint of the fields on the left side and using conformal covariance it can be shown that computing the numerator is the same as computing
$$|(U(\gamma)^\dagger1,\prod_{i=1}^n \mathscr{Y}_i(\eta_{A_i}(a_i),z_i)\prod_{i=n}^1\mathscr{Y}_i(a_i,z_i)1)|,$$
where $\eta_{A_i}$ is the anti-linear involution corresponding to the unitary structure of $A_i$ (in minimal models $\eta_{A_i}(a_i)=a_i$). It is known that correlation functions $(u,\prod \mathscr{Y}_i(b_i,w_i)v)$ for $u,v$ with finite energy, can be evaluated at distinct insertions points with the same norm by an analytic extension of the region $|w_1|<\ldots<|w_{2n}|$ to the configuration space of $\mathbb{C}^{2n}$; This is \cite[Theorem 3.5]{huang2005adifferential} for chiral and \cite[Proposition 2.8]{huang2005bdifferential} for full CFT. We first need to prove similar theorems for $u,v$ smooth vectors (this could be accomplished by proving that the corresponding correlation function satisfies the same ODE as in \cite[Theorem 1.6]{huang2005adifferential}). But another issue would remain, which is that the insertion point $z_i$ is repeated. There is a singularity when insertion points are the same which makes it \textit{impossible} to define the above quantity. 

But in $C_{\text{chiral}}^p$, we could see that the singularities cancel each other. Notice that for any correlation function of the form $(u,\prod \mathscr{Y}_i(b_i,w_i)v)$, the order of the singularity $w_i-w_j$ only depends on the fields $b_i,b_j$; see \cite[Proposition 3.5.1]{frenkel1993axiomatic} for $\mathscr{Y}_i=Y$ the vertex operator (this can be easily generalized to intertwiners). For the denominator, the first norm is: 
$$||\prod_{i=n}^1\mathscr{Y}_i(a_i,z_i)1||=|(1,\prod_{i=1}^n\mathscr{Y}_i(\eta_{A_i}(a_i),z_i)\prod_{i=n}^1\mathscr{Y}_i(a_i,z_i)1)|^{\frac{1}{2}}.$$
The above is not defined but if we consider the insertions generically at $z_i',z_i$, it should give us a meromorphic function 
$$\prod \frac{1}{(z_i'-z_i)^{s_{ii'}}}\prod \frac{1}{(z_i-z_j)^{s_{ij}}}\prod \frac{1}{(z_i'-z_j')^{s_{i'j'}}}F(z_1',\ldots,z_n',z_1,\ldots,z_n)$$ 
where $F$ is a polynomial in $z_i,z_i'$. For the second norm in the denominator, by conformal invariance, it can be evaluated at $z_i$ instead of $\gamma(z_i)$ and the same singularities will appear with the same order. Therefore, one can see that the order of $(z_i-z_i')$ is the same in the numerator and denominator. In the end,  $C_{\text{chiral}}^p$ would be of the form $\frac{F_\gamma}{F}$ for two functions in $z_i$ and the dependence of the normalization on $\gamma$, as mentioned, can be avoided by conformal invariance of vacuum to vacuum correlation functions.

Even if $C_{\text{chiral}}^p$ is defined, it is not obvious that $C_{\text{chiral}}^p \le 1$, which is crucial for quantum computation. We believe that one should be able to obtain $C_{\text{chiral}}^p$ as a limit of the smeared version of the problem called $C_{\text{chiral}}^s$, which will be at most one simply due to Cauchy inequality. Define the smeared correlation function as follows:
$$C_{\text{chiral}}^s=\frac{|\Big(\prod_{i=n}^1\mathscr{Y}_i(a_i,\beta_{d_a}(\gamma)(f_i))1,U(\gamma)\prod_{i=n}^1\mathscr{Y}_i(a_i,f_i)1\Big)|}{||\prod_{i=n}^1\mathscr{Y}_i(a_i,\beta_{d_a}(\gamma)(f_i))1||. ||\prod_{i=n}^1\mathscr{Y}_i(a_i,f_i)1||},$$
where $f_i$ are smooth functions on $S^1$ and $\beta_{d_a}(f)=\gamma'(\gamma^{-1}(z))^{\text{wt}\  a-1}f(\gamma^{-1})$ (this function appears in the conformal covariance of primary fields; see e.g. \cite[Proposition 6.4]{carpi2015vertex}). It can be shown that $C_{\text{chiral}}^s$ is defined due to energy-boundedness, see e.g. \cite[Proposition 3.9]{gui2017unitarity}. The next step is to take a limit by considering sequence of functions converging to the Dirac delta at $z_i$, i.e. $\delta_{z_i}$. To do so rigorously, a similar theorem to \cite[Proposition 3.12]{gui2017unitarity} when the expectation value is taken for smooth vectors could be helpful. This would relate the smeared and the point-like version:
\begin{gather*}
\scalebox{0.95}{$(u,\prod_{i=1}^n\mathscr{Y}_i(\eta_{A_i}(a_i),g_i)\prod_{i=n}^1\mathscr{Y}_i(a_i,f_i)v)$}
\end{gather*}
\begin{gather*}
\scalebox{0.95}{$=\int_{-\pi}^\pi \cdots \int_{-\pi}^\pi \Big(u,\prod_{i=1}^n\mathscr{Y}_i(\eta_{A_i}(a_i),e^{i\phi_i})\prod_{i=n}^1\mathscr{Y}_i(a_i,e^{i\theta_i})v\Big)\prod_{i=1}^n g_i(e^{i\phi_i}) $}
\end{gather*}
\begin{gather*}
\scalebox{0.95}{$\prod_{i=n}^1f_i(e^{i\theta_i})\prod_{i=1}^n \slashed{d}\phi_i \prod_{i=n}^1 \slashed{d}\theta_i$}
\end{gather*}
where $\slashed{d} \theta=\frac{e^{i\theta}}{2\pi}$, and $f_i,g_i$ can be thought to be distributions centered on $z_i,z_i'$ with $\int_{S^1} f_i=\int_{S^1} g_i=1$. These functions must have disjoint support (this is required to apply \cite[Proposition 3.12]{gui2017unitarity}). Then, one has to analyze the limit when $g_i, f_i \to \delta_{z_i'}=\delta_{z_i}$. This limit should be taken on the fraction $C_{\text{chiral}}^s$ as the terms in the numerator and denominator diverge individually. It is worthy to note that on distinct points the limit works well. Assuming continuity of
$$\Big(u,\prod_{i=1}^n\mathscr{Y}_i(\eta_{A_i}(a_i),e^{i\phi_i})\prod_{i=n}^1\mathscr{Y}_i(a_i,e^{i\theta_i})v\Big)$$
with respect to $\phi_i,\theta_i$ (this is already true if $u,v$ have finite energy but we need it for smooth vectors), it is not hard to show that if $g_i \to \delta_{z_i'},f_i\to \delta_{z_i}$ for distinct $z_i,z_i'$ we have 
\begin{gather*}
\scalebox{0.87}{$|(u,\prod_{i=1}^n\mathscr{Y}_i(\eta_{A_i}(a_i),g_i)\prod_{i=n}^1\mathscr{Y}_i(a_i,f_i)v)| \to |(u,\prod_{i=1}^n\mathscr{Y}_i(\eta_{A_i}(a_i),z_i')\prod_{i=n}^1\mathscr{Y}_i(a_i,z_i)v)|.$}
\end{gather*}

Recall that as mentioned at the beginning, the direct analog of the TQFT case is the full CFT point-like correlation function. Assuming that one can prove that $C_{\text{chiral}}^p$ is well-defined, it is not hard to show that the point-like and smeared version of the problem for the full CFT case can also be defined. In fact, for the smeared version, just like the chiral case, there is nothing to prove and it is already well-defined. For the point-like case, the full CFT intertwiner is a finite sum of pair of chiral intertwiners, hence the correlation function will be a finite sum of chiral correlation functions. We have the following analogous quantities for the full CFT:

\begin{gather*}
\scalebox{0.9999}{$C_{\text{full}}^p= \frac{|\Big(\prod_{i=n}^{1}\mathbb{Y}(a_i\otimes a_i',\gamma(z_i),\overline{\gamma({z_i})})1,U_L(\gamma) U_R(j \circ \gamma \circ j)\prod_{i=n}^{1}\mathbb{Y}(a_i\otimes a_i',z_i,\overline{z_i})1\Big)|}{||\prod_{i=n}^{1}\mathbb{Y}(a_i\otimes a_i',\gamma(z_i),\overline{\gamma({z_i})})1||.||\prod_{i=n}^{1}\mathbb{Y}(a_i\otimes a_i',z_i,\overline{z_i})1||},$}
\end{gather*}
\begin{gather*}
\scalebox{0.98}{$C_{\text{full}}^s= \frac{|\Big(\prod_{i=n}^{1}\mathbb{Y}(a_i\otimes a_i',\beta_{d_a}(\gamma)(f_i),\beta_{d_a}(\gamma)(f_i)\circ j)1,U_L(\gamma) U_R(j \circ \gamma \circ j)\prod_{i=n}^{1}\mathbb{Y}(a_i\otimes a_i',f_i,f_i\circ j)1\Big)|}{||\prod_{i=n}^{1}\mathbb{Y}(a_i\otimes a_i',\beta_{d_a}(\gamma)(f_i),\beta_{d_a}(\gamma)(f_i)\circ j)1||.||\prod_{i=n}^{1}\mathbb{Y}(a_i\otimes a_i',f_i,f_i\circ j)1||},$}
\end{gather*}
where $a_i' \in A_i'$ are primary fields from the contragredient module of $A_i$, $j$ is the conjugation map, and $U_L,U_R$ are the unitary evolution for the left and right moving part, respectively.

Finally, as we know how to take the unitary evolution as input (outlined in previous section), formulating the correlation function problem in all four versions is possible; the inputs are the primary fields from a nice fixed VOA, the insertion points $z_i$ (or Fourier coefficients of the smearing functions) and the decomposition of $\gamma$ as $\exp(f)$s. 

It is not entirely clear how the AC approach would help solve this question for general minimal models as it is hard to realize finite version of fields on a lattice or a spin chain as mentioned in \hyperref[5.2]{section 5.2}. Similar to unitary evolution, the correlation function problem is expected to be in \textbf{BQP} and typically \textbf{BQP}-complete.

\section*{Acknowledgements}

We would like to thank James Tener for many helpful discussions, in particular pointing out the relevance of Kaplansky's theorem for the proof of \hyperref[cor4.10]{\textbf{Corollary 4.10}}. We thank P. Fendley, V. Jones, R. Koenig, H. Saleur, and anonymous referees for helpful comments to improve our paper and their encouragements.  The second author is partially supported by NSF grant DMS-1411212.

\bibliographystyle{apa}
\bibliography{main}

\begin{thebibliography}{}

\bibitem[\protect\astroncite{Aasen et~al.}{2016}]{aasen2016topological}
Aasen, D., Mong, R.~S., and Fendley, P. (2016).
\newblock Topological defects on the lattice: I. the {Ising} model.
\newblock {\em Journal of Physics A: Mathematical and Theoretical},
  49(35):354001.

\bibitem[\protect\astroncite{Aharonov et~al.}{2007}]{aharonov2007polynomial}
Aharonov, D., Arad, I., Eban, E., and Landau, Z. (2007).
\newblock Polynomial quantum algorithms for additive approximations of the
  {Potts} model and other points of the {Tutte} plane.
\newblock {\em arXiv preprint quant-ph/0702008}.

\bibitem[\protect\astroncite{Baxter}{2016}]{baxter2016exactly}
Baxter, R.~J. (2016).
\newblock {\em Exactly solved models in statistical mechanics}.
\newblock Elsevier.

\bibitem[\protect\astroncite{Bondesan et~al.}{2015}]{bondesan2015chiral}
Bondesan, R., Dubail, J., Faribault, A., and Ikhlef, Y. (2015).
\newblock Chiral {$SU(2)_k$} currents as local operators in vertex models and
  spin chains.
\newblock {\em Journal of Physics A: Mathematical and Theoretical},
  48(6):065205.

\bibitem[\protect\astroncite{Bravyi and Kitaev}{2002}]{bravyi2002fermionic}
Bravyi, S.~B. and Kitaev, A.~Y. (2002).
\newblock Fermionic quantum computation.
\newblock {\em Annals of Physics}, 298(1):210--226.

\bibitem[\protect\astroncite{Carpi et~al.}{2015}]{carpi2015vertex}
Carpi, S., Kawahigashi, Y., Longo, R., and Weiner, M. (2015).
\newblock From vertex operator algebras to conformal nets and back.
\newblock {\em arXiv preprint arXiv:1503.01260}.

\bibitem[\protect\astroncite{Chelkak et~al.}{2012}]{chelkak2012conformal}
Chelkak, D., Hongler, C., and Izyurov, K. (2012).
\newblock Conformal invariance of spin correlations in the planar {Ising}
  model.
\newblock {\em arXiv preprint arXiv:1202.2838}.

\bibitem[\protect\astroncite{Engel and Nagel}{2006}]{engel2006short}
Engel, K.-J. and Nagel, R. (2006).
\newblock {\em A short course on operator semigroups}.
\newblock Springer Science \& Business Media.

\bibitem[\protect\astroncite{Feiguin et~al.}{2007}]{feiguin2007interacting}
Feiguin, A., Trebst, S., Ludwig, A.~W., Troyer, M., Kitaev, A., Wang, Z., and
  Freedman, M.~H. (2007).
\newblock Interacting anyons in topological quantum liquids: {The} golden
  chain.
\newblock {\em Physical review letters}, 98(16):160409.

\bibitem[\protect\astroncite{Fendley}{2014}]{fendley2014free}
Fendley, P. (2014).
\newblock Free parafermions.
\newblock {\em Journal of Physics A: Mathematical and Theoretical},
  47(7):075001.

\bibitem[\protect\astroncite{Feverati and Pearce}{2003}]{feverati2003critical}
Feverati, G. and Pearce, P.~A. (2003).
\newblock Critical {RSOS} and minimal models: fermionic paths, {Virasoro}
  algebra and fields.
\newblock {\em Nuclear Physics B}, 663(3):409--442.

\bibitem[\protect\astroncite{Fiedler et~al.}{2017}]{fiedler2017jones}
Fiedler, L., Naaijkens, P., and Osborne, T.~J. (2017).
\newblock Jones index, secret sharing and total quantum dimension.
\newblock {\em New Journal of Physics}, 19(2):023039.

\bibitem[\protect\astroncite{Francesco et~al.}{2012}]{francesco2012conformal}
Francesco, P., Mathieu, P., and S{\'e}n{\'e}chal, D. (2012).
\newblock {\em Conformal field theory}.
\newblock Springer Science \& Business Media.

\bibitem[\protect\astroncite{Fredenhagen and
  Hertel}{1981}]{fredenhagen1981local}
Fredenhagen, K. and Hertel, J. (1981).
\newblock Local algebras of observables and pointlike localized fields.
\newblock {\em Communications in Mathematical Physics}, 80(4):555--561.

\bibitem[\protect\astroncite{Freedman et~al.}{2003}]{tqc03}
Freedman, M., Kitaev, A., Larsen, M., and Wang, Z. (2003).
\newblock Topological quantum computation.
\newblock {\em Bulletin of the American Mathematical Society}, 40(1):31--38.

\bibitem[\protect\astroncite{Freedman et~al.}{2002}]{simulation02}
Freedman, M.~H., Kitaev, A., and Wang, Z. (2002).
\newblock Simulation of {Topological} {Field} {Theories} by {Quantum}
  {Computers}.
\newblock {\em Communications in Mathematical Physics}, 227(3):587--603.

\bibitem[\protect\astroncite{Frenkel et~al.}{1993}]{frenkel1993axiomatic}
Frenkel, I., Huang, Y.-Z., and Lepowsky, J. (1993).
\newblock {\em On axiomatic approaches to vertex operator algebras and
  modules}, volume 494.
\newblock American Mathematical Soc.

\bibitem[\protect\astroncite{Gainutdinov
  et~al.}{2013a}]{gainutdinov2013logarithmic}
Gainutdinov, A., Jacobsen, J., Read, N., Saleur, H., and Vasseur, R. (2013a).
\newblock Logarithmic conformal field theory: a lattice approach.
\newblock {\em Journal of Physics A: Mathematical and Theoretical},
  46(49):494012.

\bibitem[\protect\astroncite{Gainutdinov
  et~al.}{2013b}]{gainutdinov2013bimodule}
Gainutdinov, A., Read, N., and Saleur, H. (2013b).
\newblock Bimodule structure in the periodic {$\mathfrak{gl}(1|1)$} spin chain.
\newblock {\em Nuclear Physics B}, 871(2):289--329.

\bibitem[\protect\astroncite{Gainutdinov
  et~al.}{2013c}]{gainutdinov2013continuum}
Gainutdinov, A., Read, N., and Saleur, H. (2013c).
\newblock Continuum limit and symmetries of the periodic {$\mathfrak{gl}(1|1)$}
  spin chain.
\newblock {\em Nuclear Physics B}, 871(2):245--288.

\bibitem[\protect\astroncite{Gainutdinov
  et~al.}{2016}]{gainutdinov2016associative}
Gainutdinov, A., Read, N., and Saleur, H. (2016).
\newblock Associative {Algebraic} {Approach} to {Logarithmic} {CFT} in the
  {Bulk}: {The} {Continuum} {Limit} of the {$\mathfrak{gl}(1|1)$} {Periodic}
  {Spin} {Chain}, {Howe} {Duality} and the {Interchiral} {Algebra}.
\newblock {\em Communications in Mathematical Physics}, 341(1):35--103.

\bibitem[\protect\astroncite{Gainutdinov
  et~al.}{2015}]{gainutdinov2015periodic}
Gainutdinov, A., Read, N., Saleur, H., and Vasseur, R. (2015).
\newblock The periodic $s\ell(2|1)$ alternating spin chain and its continuum
  limit as a bulk logarithmic conformal field theory at {$c=0$}.
\newblock {\em Journal of high energy physics}, 2015(5):114.

\bibitem[\protect\astroncite{Gainutdinov et~al.}{2014}]{gainutdinov2014lattice}
Gainutdinov, A., Saleur, H., and Tipunin, I.~Y. (2014).
\newblock Lattice {W}-algebras and logarithmic {CFTs}.
\newblock {\em Journal of Physics A: Mathematical and Theoretical},
  47(49):495401.

\bibitem[\protect\astroncite{Gils et~al.}{2013}]{gils2013anyonic}
Gils, C., Ardonne, E., Trebst, S., Huse, D.~A., Ludwig, A.~W., Troyer, M., and
  Wang, Z. (2013).
\newblock Anyonic quantum spin chains: {Spin}-1 generalizations and topological
  stability.
\newblock {\em Physical Review B}, 87(23):235120.

\bibitem[\protect\astroncite{Ginsparg}{1988}]{ginsparg1988applied}
Ginsparg, P. (1988).
\newblock Applied conformal field theory.
\newblock {\em arXiv preprint hep-th/9108028}.

\bibitem[\protect\astroncite{Goodman and Wallach}{1985}]{goodman1985projective}
Goodman, R. and Wallach, N.~R. (1985).
\newblock Projective unitary positive-energy representations of
  {$\text{Diff}(S^1)$}.
\newblock {\em Journal of functional analysis}, 63(3):299--321.

\bibitem[\protect\astroncite{Gui}{2017}]{gui2017unitarity}
Gui, B. (2017).
\newblock Unitarity of the modular tensor categories associated to unitary
  vertex operator algebras, {I}.
\newblock {\em arXiv preprint arXiv:1711.02840}.

\bibitem[\protect\astroncite{Haag}{2012}]{haag2012local}
Haag, R. (2012).
\newblock {\em Local quantum physics: {Fields}, particles, algebras}.
\newblock Springer Science \& Business Media.

\bibitem[\protect\astroncite{Huang}{2005a}]{huang2005adifferential}
Huang, Y.-Z. (2005a).
\newblock Differential equations and intertwining operators.
\newblock {\em Communications in Contemporary Mathematics}, 7(03):375--400.

\bibitem[\protect\astroncite{Huang}{2005b}]{huang2005bdifferential}
Huang, Y.-Z. (2005b).
\newblock Differential equations, duality and modular invariance.
\newblock {\em Communications in Contemporary Mathematics}, 7(05):649--706.

\bibitem[\protect\astroncite{Jones}{2003}]{jones03}
Jones, V.~F. (2003).
\newblock In and around the origin of quantum groups.
\newblock {\em arXiv preprint math/0309199}.

\bibitem[\protect\astroncite{Jones}{2014}]{jones14}
Jones, V.~F. (2014).
\newblock Some unitary representations of {Thompson}'s groups {F} and {T}.
\newblock {\em arXiv preprint arXiv:1412.7740}.

\bibitem[\protect\astroncite{Jones}{2016}]{jones16}
Jones, V.~F. (2016).
\newblock A no-go theorem for the continuum limit of a periodic quantum spin
  chain.
\newblock {\em arXiv preprint arXiv:1607.08769}.

\bibitem[\protect\astroncite{Jones}{2017}]{jones17}
Jones, V.~F. (2017).
\newblock Scale invariant transfer matrices and {Hamiltionians}.
\newblock {\em arXiv preprint arXiv:1706.00515}.

\bibitem[\protect\astroncite{Jordan et~al.}{2012}]{jordan12}
Jordan, S.~P., Lee, K.~S., and Preskill, J. (2012).
\newblock Quantum algorithms for quantum field theories.
\newblock {\em Science}, 336(6085):1130--1133.

\bibitem[\protect\astroncite{Kaufmann}{1995}]{kaufmann1995path}
Kaufmann, R.~M. (1995).
\newblock Path space decompositions for the {Virasoro} algebra and its {Verma}
  modules.
\newblock {\em International Journal of Modern Physics A}, 10(07):943--961.

\bibitem[\protect\astroncite{Kawahigashi and
  Longo}{2004}]{kawahigashi2004classification}
Kawahigashi, Y. and Longo, R. (2004).
\newblock Classification of local conformal nets. {Case} {$c<1$}.
\newblock {\em Annals of mathematics}, pages 493--522.

\bibitem[\protect\astroncite{K{\"o}nig and Scholz}{2017}]{konig2017matrix}
K{\"o}nig, R. and Scholz, V.~B. (2017).
\newblock Matrix product approximations to conformal field theories.
\newblock {\em Nuclear Physics B}, 920:32--121.

\bibitem[\protect\astroncite{Koo and Saleur}{1994}]{koo1994representations}
Koo, W. and Saleur, H. (1994).
\newblock Representations of the {Virasoro} algebra from lattice models.
\newblock {\em Nuclear Physics B}, 426(3):459--504.

\bibitem[\protect\astroncite{Lloyd}{1996}]{lloyd1996universal}
Lloyd, S. (1996).
\newblock Universal quantum simulators.
\newblock {\em Science}, 273(5278):1073.

\bibitem[\protect\astroncite{Milsted and Vidal}{2017}]{milsted2017extraction}
Milsted, A. and Vidal, G. (2017).
\newblock Extraction of conformal data in critical quantum spin chains using
  the {Koo-Saleur} formula.
\newblock {\em Physical Review B}, 96(24):245105.

\bibitem[\protect\astroncite{Mong et~al.}{2014}]{mong2014parafermionic}
Mong, R.~S., Clarke, D.~J., Alicea, J., Lindner, N.~H., and Fendley, P. (2014).
\newblock Parafermionic conformal field theory on the lattice.
\newblock {\em Journal of Physics A: Mathematical and Theoretical},
  47(45):452001.

\bibitem[\protect\astroncite{Murg and Cirac}{2004}]{murg2004adiabatic}
Murg, V. and Cirac, J.~I. (2004).
\newblock Adiabatic time evolution in spin systems.
\newblock {\em Physical Review A}, 69(4):042320.

\bibitem[\protect\astroncite{Neretin}{1996}]{neretin1996categories}
Neretin, Y.~A. (1996).
\newblock {\em Categories of symmetries and infinite-dimensional groups}.
\newblock Number~16. Oxford University Press.

\bibitem[\protect\astroncite{Pasquier and Saleur}{1990}]{saleur90}
Pasquier, V. and Saleur, H. (1990).
\newblock Common structures between finite systems and conformal field theories
  through quantum groups.
\newblock {\em Nuclear Physics B}, 330(2-3):523--556.

\bibitem[\protect\astroncite{Read and Saleur}{2007}]{read2007enlarged}
Read, N. and Saleur, H. (2007).
\newblock Enlarged symmetry algebras of spin chains, loop models, and
  {S}-matrices.
\newblock {\em Nuclear Physics B}, 777(3):263--315.

\bibitem[\protect\astroncite{Rowell and Wang}{2012}]{RW12}
Rowell, E.~C. and Wang, Z. (2012).
\newblock Localization of unitary braid group representations.
\newblock {\em Communications in Mathematical Physics}, 311(3):595--615.

\bibitem[\protect\astroncite{Rowell and Wang}{2017}]{RW17}
Rowell, E.~C. and Wang, Z. (2017).
\newblock Mathematics of {Topological} {Quantum} {Computing}.
\newblock {\em arXiv preprint arXiv:1705.06206}.

\bibitem[\protect\astroncite{Schweigert et~al.}{2006}]{frs06}
Schweigert, C., Fuchs, J., and Runkel, I. (2006).
\newblock Categorification and correlation functions in conformal field theory.
\newblock In {\em Proceedings of the ICM}, pages 443--458. European
  Mathematical Society Z{\"u}rich.

\bibitem[\protect\astroncite{Seiberg}{2015}]{seibergQFT}
Seiberg, N. (2015).
\newblock What is a quantum field theory?
\newblock {\em Simons center talk:
  \url{http://scgp.stonybrook.edu/video_portal/video.php?id=389}}.

\bibitem[\protect\astroncite{Smirnov}{2007}]{smirnov07}
Smirnov, S. (2007).
\newblock Towards conformal invariance of {2D} lattice models.
\newblock In {\em Proceedings of the International Congress of Mathematicians
  Madrid, August 22--30, 2006}, pages 1421--1451.

\bibitem[\protect\astroncite{Tener}{2016}]{tener2016geometric}
Tener, J.~E. (2016).
\newblock Geometric realization of algebraic conformal field theories.
\newblock {\em arXiv preprint arXiv:1611.01176}.

\bibitem[\protect\astroncite{Tener and Wang}{2016}]{tener16}
Tener, J.~E. and Wang, Z. (2016).
\newblock On classification of extremal non-holomorphic conformal field
  theories.
\newblock {\em arXiv preprint arXiv:1611.04071}.

\bibitem[\protect\astroncite{Weiner}{2007}]{weiner2007conformal}
Weiner, M. (2007).
\newblock Conformal covariance and related properties of chiral {QFT}.
\newblock {\em arXiv preprint math/0703336}.

\bibitem[\protect\astroncite{Witten}{1989}]{witten1989quantum}
Witten, E. (1989).
\newblock Quantum field theory and the jones polynomial.
\newblock {\em Communications in Mathematical Physics}, 121(3):351--399.

\end{thebibliography}

\address{\textsuperscript{1\label{1}}Dept of Mathematics, University of California,
Santa Barbara, CA 93106-6105, U.S.A.}
\address{\textsuperscript{2\label{2}}Microsoft Station Q and Dept of Mathematics, University of California,
Santa Barbara, CA 93106-6105, U.S.A.} 
\addresseshere

\appendix
\section*{Appendix: Scaling limit of Ising anyonic chains}\label{A}

\subsection*{1. Obtaining Virasoro representations and their actions}\label{A.1}

\hfill \break 
In each case, we start with some operator that is supposed to become the \textit{desired} one converging to $L_m$, and it will undergo some changes (all being some scalings) before becoming the desired operator. As an example, for the Hamiltonian, we will always start with $-\sum e_j$ but during the process, it will change by some scaling which can be easily obtained by following the procedure until it produces the actual Hamiltonian $\widetilde{L}_0^c$ that converges to $L_0$. These scalings are the scaling factors mentioned in \hyperref[identities]{$\widetilde{L}_m^{c,s}$ identities} called $\alpha_n^c,\alpha_n^s,\beta_n^{m,c},$ and $\beta_n^{m,s}$. 

\subsubsection*{\textbf{Case 1(a): $(\frac{1}{2},\frac{1}{2}) \xrightarrow{SL} \rchi_0+\rchi_{\frac{1}{2}}$}}
\hfill \break
The method and the notations used in \cite{fendley2014free} will be followed closely and we will apply it case by case on Ising ACs to obtain the Virasoro modes throughout this section. It is therefore necessary to review the general procedure described for the Hamiltonian diagonalization of 1(a) in \cite{fendley2014free}. 

Consider the operator $-\sum_{j=1}^{2n-1} t_je_j$ which after a scaling due to the equalities
$$e_{2j}=\frac{1}{\sqrt{2}}(1+\sigma_{j}^z\sigma_{j+1}^z), \ \ e_{2j-1}=\frac{1}{\sqrt{2}}(1+\sigma_j^x),$$
becomes
$$H= -\sum_{j=1}^{n} t_{2j-1} \sigma_j^{x}-\sum_{j=1}^{n-1} t_{2j} \sigma_{j}^z \sigma_{j+1}^z,$$
where the coefficients $t_j$ are fixed. With this expression of $H$, it is easy to see the famous $\mathbb{Z}_2$ symmetry provided by the spin-flip operator, called
$$(-1)^F:= \prod_{j=1}^{2n} \sigma_j^x.$$
As detailed in \cite{fendley2014free}, in order to diagonalize this Hamiltonian, the \textbf{Majorana} operators should be defined as
$$ \psi_{2j-1}=\Big( \prod_{k=1}^{j-1} \sigma_k^x  \Big) \sigma_j^z, \ \ \ \ \psi_{2j}= i\Big(\prod_{k=1}^{j} \sigma_k^x  \Big) \sigma_j^z $$
which satisfy the ACR (Anticommutative Canonical Relations):
$$\{\psi_a,\psi_b\}=2\delta_{ab} , \ \  \forall a,b=1, \ldots,2n.$$
It is a well-known fact that these operators and their monomials are linearly independent and this representation of the Clifford algebra is faithful. By using
$$e_a=\frac{1}{\sqrt{2}}(1+i\psi_{a}\psi_{a+1}),$$
we rewrite the Hamiltonian 
$$H=i \sum_{a=1}^{2n-1} t_a \psi_{a+1}\psi_a.$$
Next, raising (creation) and lowering (annihilation) operators are introduced, i.e. \textbf{Dirac} operators for which
$$[H,\Psi]=2\epsilon \Psi.$$
Notice that for any operator linear in the Majorana operators, the commutator with $H$ is also linear in the Majorana operators. Let us choose the following form for $\Psi$
$$\Psi = \sum_b i^b \mu_b \psi_b,$$
where $\mu_b$ are numbers that will turn out to be real. The $i^b$'s factor will ensure that the matrix in \hyperref[eq15]{(15)} is hermitian and not skew-hermitian, thus making the computations easier. Computing $\mu_a'$s, 
$$\Psi'=[H,\Psi]= \sum_{a} i^a\mu_a' \psi_a,$$
is same as the following matrix equation
\begin{align}\label{eq15}
   \begin{pmatrix}
           \mu_1' \\
           \mu_{2}' \\
           \vdots \\
            \\
           \mu_{2n}'
         \end{pmatrix}
         = 2\begin{pmatrix}
         0 & t_1 & 0 & \ldots &  \\
         t_1 & 0 & t_2 & &  \\
         0 & t_2 & 0 & &   \\
         \vdots &  &  & & t_{2n-1} \\
         &  & & t_{2n-1} & 0 
         \end{pmatrix}
         \begin{pmatrix}
         \mu_1 \\
         \mu_{2} \\
         \vdots \\
            \\
         \mu_{2n}
         \end{pmatrix}.
\end{align}
This hermitian matrix has determinant $(-1)^n\prod_{j=1}^{n}t_{2j-1}^2$. The eigenvectors of this matrix give the Dirac operators and each corresponding eigenvalue is the energy that is raised or lowered. Specializing the values of $t_j$s will give the different boundary conditions. $(\frac{1}{2},\frac{1}{2})$ can be seen to correspond to the case $t_j=1$ for all $j$. Therefore, we will work with the matrix \hyperref[eq15]{\hyperref[eq15]{(15)}} assuming $t_j=1$. 

\textit{Notation}. for $n \in \mathbb{N}$, set $[n]:=\{1,\ldots,n\}$. E.g. $[2n]-[n]=\{n+1,\ldots,2n\}$. Similarly define $[-n]:=\{-1,\ldots,-n\}$.

The Dirac operators $\Psi_{k}$ for $k \in [2n]$, are given by the eigenvectors $\mu_{a,k}=\sin(\frac{ak\pi}{2n+1})$ with corresponding energy $\epsilon_k=4\cos(\frac{k\pi}{2n+1})$, satisfying (\cite{fendley2014free})
\begin{align}\label{eq16}
[H,\Psi_{\pm k}]=2 \epsilon_{\pm k} \Psi_{\pm k}, \ \{\Psi_{\pm k}, \Psi_{\pm k'} \}=0, \ \{\Psi_{\pm k}, \Psi_{\mp k'} \} = N_k \delta_{k,k'}\textbf{1},
\end{align}
where $\Psi_{-k}:=\Psi_{2n+1-k}$, and $N_k=2 \sum_{a} \abs{\mu_{a,k}}^2$. The relations are obtained using the identities
$$\{\Psi,\rchi\}=\sum_{a,b} i^{a+b} \mu_a \nu_b \{\psi_a,\psi_b\}=2\sum_a (-1)^a \mu_a\nu_a,$$
for any two linear Majorana forms $\Psi = \sum_b i^b \mu_b \psi_b,\rchi=\sum_b i^b \nu_b \psi_b$. As a hermitian matrix has orthogonal eigenvectors, and for any eigenvector $(\mu_{a,k})_a$ giving eigenvalue $\epsilon_k$, there is a corresponding eigenvector  $((-1)^{a+1}\mu_{a,k})_a$ giving eigenvalue $\epsilon_{-k}:= -\epsilon_k$, equations \hyperref[eq16]{(16)} follow including the fact that $\Psi_{k}^{\dagger}=\Psi_{-k}$. We will always work with the normalization of $\Psi_k$ by $\sqrt{N_k}$, hence $\{\Psi_{\pm k}, \Psi_{\mp k'} \} = \delta_{k,k'}\textbf{1}$.

From now on, the Dirac operators $\Psi_{k}$ for $k \in [n]$ will be called the \textit{raising or creation} operators and the Dirac operators $\Psi_{k}$ for $k \in [2n]-[n]$ will be called the \textit{lowering or annihilation} operators. This terminology will similarly apply for future cases. Further, at the end of each case, there will be a renumbering of the operators indices which will make the creation operators have negative index while the annihilation operators will have positive index.

Therefore, $\Psi_{k}$s satisfy the ACR while the dimension of $\mathcal{W}_n$ (the Hilbert space) is $2^n$. This implies the existence of an orthonormal basis of $\mathcal{W}_n$ given by
$$\prod\limits_{i \in S} \Psi_{i} 1_n, \ \forall S \subset [n],$$
all of which will turn out to be eigenvectors of $H$, where $1_n$ is the \textit{vacuum} or \textit{ground state} annihilated by the annihilation operators. As mentioned in \cite{fendley2014free}, the energy symmetry of $H$ and well-known properties of the representations of the algebra generated by the $\Psi_k$s, can be used to prove this. Let us recall these general facts.

\textit{Notation.} Denote by $\mathcal{F}_n$ the algebra generated by the $\Psi_k$s and $\mathcal{F}_n^+$ the sub-algebra generated by the creation operators. Similarly define $\mathcal{F}_n^-$. We will use $S$ as any subset of the indices of creation operators.

\begin{fct}\label{fct1}
Let $\mathcal{W}$ be a representation of $\mathcal{F}_n$ which is a Hilbert space with $\dim \mathcal{W}=2^s$  where $s \ge n$ and $\Psi_k^\dagger=\Psi_{-k}$ with respect to the inner product of $\mathcal{W}$. Consider the image $\mathcal{W}_0$ of the product of all annihilation operators. For any vector $v \in \mathcal{W}_0$, by definition of $\mathcal{W}_0$ and ACR relations, in particular $\Psi_k^2=0$, we get $\mathcal{F}_n^-(v)=\{0\}$. Further, the space $\mathcal{W}_v=\mathcal{F}_n^+(v)$ generated by the creation operators acting on $v$ has dimension $2^n$ with, assuming $v$ is a unit vector, an orthonormal basis $\{\prod\limits_{i \in S} \Psi_{i} v| \ \forall S\}$. The fact that this is an orthonormal basis can also be checked directly by computing the inner products using ACR, $\Psi_k^\dagger=\Psi_{-k}$ and that $\mathcal{F}_n^-(v)=\{0\}$. Finally, with the same direct calculations, for any two orthonormal vectors $u,v \in \mathcal{W}_0$, we have $\mathcal{W}_v \perp \mathcal{W}_u$. This implies that for any chosen orthonormal basis for $\mathcal{W}_0$, a direct sum of irreducible representations with dimension $2^n$ of $\mathcal{F}_n$ is obtained. We claim that this decomposition exhausts $\mathcal{W}$, or equivalently $\dim \mathcal{W}_0=2^{s-n}$. Assume $\exists v \in \mathcal{W}$ which is orthogonal to the decomposition. One needs to find a sequence of $\Psi_k$s acting on $v$ such that $v$ is sent to a vector in $\mathcal{W}_0$ and we will reach a contradiction. This is done by noticing that for any non-zero vector $u$, if $\Psi_k u=0$ for $\Psi_k \in \mathcal{F}_n^{+}$ then $\Psi_{k}\Psi_{k}^\dagger u=u \neq 0$. Therefore, we can start by acting on $v$ by the annihilation operators in increasing order of indices ($n+1,\ldots,2n$) and whenever the result is zero when acted by $\Psi_k$, acting by the creation operator $\Psi_{k}^\dagger$ and then $\Psi_{k}$ resolves this issue. By using this procedure and ACR, there is a reordering of the action by $\Psi_k$s such that the end result is $\Psi_{2n} \ldots \Psi_{n+1}(\prod_{i \in S} \Psi_{i}v)=v_0 \neq 0$ where $S$ is some subset from the creation indices. $v_0$ is in the image of the    product of all annihilation operators, i.e. $v_0 \in \mathcal{W}_0$. But $v$ was assumed to be orthogonal to the decomposition, implying that
$$(v_0,v_0)=(v, (\prod_{i \in S} \Psi_{i})^\dagger \Psi_{n} \ldots \Psi_{1}v_0)=0 \implies v_0=0,$$
which is a contradiction.
\end{fct}

\begin{fct}\label{fct2}
With the same settings of \hyperref[fct1]{\textit{Fact 1}}, consider a matrix $D$ satisfying $[D,\Psi_k]=0$ for all $\Psi_k \in \mathcal{F}_n$. It follows that $D$ preserves $\mathcal{W}_0$ and it is uniquely determined based on how it acts on $\mathcal{W}_0$. In particular, if there is a decomposition of $\mathcal{W}$ into $2^{s-n}$ irreducible representations where $D$ preserves the corresponding vacuums, then $D$ acts as a scalar on each one of them. This will be always the case in the proofs.
\end{fct}

\begin{fct}\label{fct3}
In addition to the spin-flip symmetry $(-1)^F$, the matrix $H$ has charge conjugation symmetry (which will also be called energy symmetry) provided by $C=\prod_i \sigma_i^z \prod_i (\sigma_i^x)^i$ which satisfies $CH=-HC$ implying each energy has one corresponding opposite energy. This is a necessary property which helps us to show that some non-zero scalar from the previous fact for $H$ can not happen as that would break the symmetry.
\end{fct} 
From \hyperref[fct1]{\textit{Fact 1}}, $(\mathcal{W}_n)_0$ is one dimensional from which a unit vector $1_n$ is chosen. Define
$$H':= \sum_{k\in [n]} \epsilon_k(\Psi_{+k}\Psi_{-k}-\Psi_{-k}\Psi_{+k}).$$
$H'$s eigenvectors are $\{\prod_{i \in S} \Psi_{i}1_n|\ \forall S\}$, each with the corresponding eigenvalue $\sum_{i \in S} \epsilon_i - \sum_{j \not \in S} \epsilon_j$. So $H'$ has $C$-symmetry. Further, one can easily see that $[H',\Psi_k]=2\epsilon_k\Psi_k$ and so, for  $D=H-H'$, $[D,\Psi_k]=0$. As  $(\mathcal{W}_n)_0$ is one dimensional, $D=\alpha\textbf{1}$. But $H'$ shifted by any $\alpha$ does not satisfy the energy symmetry. Therefore, $\alpha=0$ and $H'=H$. Taking the shift $H \rightarrow H+ \sum \epsilon_k$ and using $\{\Psi_{+k},\Psi_{-k}\}=\textbf{1}$,
$$H=\sum_{k \in [n]} 2\epsilon_k\Psi_{+k}\Psi_{-k}.$$
The final change to $H$ is $H \rightarrow \frac{2n+1}{8\pi}H$ and the desired Hamiltonian $\widetilde{L}_0^c$ is given by:
\begin{align}\label{eq17}
\widetilde{L}_0^c=\frac{2n+1}{\pi}\sum_{k \in [n]} \cos(\frac{k\pi}{2n+1})\Psi_{+k}\Psi_{-k}.
\end{align}
Defining the scaling limit requires defining the connecting maps. Before doing so, a renumbering $k \rightarrow k-\frac{1}{2}-n$ is performed to get the creation operators indices as $\{-\frac{1}{2},\ldots,-(n-\frac{1}{2})\}$. 

\textit{Notation}. $[(n+\frac{1}{2})]:=\{\frac{1}{2},\ldots,(n-\frac{1}{2})\}$ and $[-(n+\frac{1}{2})]:=\{-\frac{1}{2},\ldots,-(n-\frac{1}{2})\}$.

This will also change the coefficients from $\cos(\frac{k\pi}{2n+1})=-\sin(\frac{(k-\frac{1}{2}-n)\pi}{2n+1})$ to $\sin(\frac{-k\pi}{2n+1})$ and we will have
$$\widetilde{L}_0^c=\frac{2n+1}{\pi}\sum_{k \in [(n+\frac{1}{2})]} \sin(\frac{k\pi}{2n+1})\Psi_{-k}\Psi_{k}.$$
Next, we define 
$$\phi_n:\mathcal{W}_n \hookrightarrow \mathcal{W}_{n+1}, \ \ \text{where }\forall S \text{ we have } \phi_n(\prod_{i \in S} \Psi_{i}1_n)=\prod_{i \in S} \Psi_{i}1_{n+1}.$$
This is consistent with an embedding of $\mathcal{F}_n \hookrightarrow \mathcal{F}_{n+1}$ where $\Psi_i \hookrightarrow \Psi_i$ giving us in the limit the algebra of Dirac fermion operator $\mathcal{F}$. We will prove that there is a \textbf{strong} scaling limit (see \hyperref[dfn4]{\textbf{Definition 4}}), where the scaling limit space $\mathcal{V}$ can be constructed as the algebraic colimit of the sequence coming with the natural embedding maps $\rho_n: \mathcal{W}_n \hookrightarrow \mathcal{V}$. The connecting maps will turn out to be the restriction of $\phi_n$ to energy $M$ as it is required in \hyperref[dfn4]{\textbf{Definition 4}}. The natural orthonormal spanning set is 
$$\{\prod_{i \in S} \Psi_i1| \ \forall S \subset \mathbb{Z}_{<0}+\frac{1}{2}\}$$
for $\mathcal{V}$ where $1=\rho_n(1_n)$ is the \textit{vacuum} vector. We need to make sure that this is consistent with the definition of scaling limit obtained through the double colimit construction in \hyperref[dfn3]{\textbf{Definition 3}}.

Restricting to energy at most $M$, one has to check that $\phi_n$ gives isometries $\phi_n^M : \mathcal{W}_n^M \rightarrow \mathcal{W}_{n+1}^M$ for large enough $n$. It is not hard to see that any eigenvector $\prod_{-k \in S} \Psi_{-k} 1_n$ with energy $\frac{(2n+1)}{\pi}(\sum_{-k \in S} \sin(\frac{k\pi}{2n+1})) < M$ has the energy $\frac{(2n+3)}{\pi}(\sum_{-k \in S} \sin(\frac{k\pi}{2n+3}))$ given by $H_{n+1}$ also smaller than $M$ for large enough $n$. Indeed, by using the Taylor expansion we obtain
\begin{align}\label{eq18}
\frac{(2n+1)}{\pi}(\sum\limits_{-k \in S} \sin(\frac{k\pi}{2n+1}))=\sum\limits_{k \in S}k - \sum\limits_{-k \in S}\frac{k^3 \pi^2}{6(2n+1)^2}+\ldots
\end{align}
$$=\sum\limits_{-k \in S}k+O(\frac{1}{n}).$$
This also shows that at the scaling limit we have the energy $\sum_{-k \in S} k$ for $\prod_{-k \in S} \Psi_{-k} 1$. It is then easy to check that $\widetilde{L}_0^c \xrightarrow{SL} L_0$, where
$$L_0=\sum_{k \in \mathbb{N}-\frac{1}{2}}k\Psi_{-k}\Psi_k.$$
This gives the character 
$$\prod_{k=1}^{\infty} (1+q^{k-\frac{1}{2}})$$
which agrees with the character of $\rchi_0+\rchi_\frac{1}{2}$ (\hyperref[1.2]{section 1.2}). 

Now consider the natural action of $\widetilde{L}_0^c$ on $\mathcal{V}$ obtained through the embedding of $\mathcal{F}_n \hookrightarrow \mathcal{F}$. We could alternatively take the ``less'' natural action by extending $\widetilde{L}_0^c$ by zero on the orthogonal complement of $\mathcal{W}_n$ in $\mathcal{V}$ and no result on the convergence rate will be lost. By this embedding, the restriction of both $\widetilde{L}_0^c,L_0$ to subspace with energy at most $\sqrt[3]{n}$ denoted by $\widetilde{L}_0^c|_{\sqrt[3]{n}},L_0|_{\sqrt[3]{n}}$, can be compared. In order to finish the proof of 1(a), one needs to prove
$$\widetilde{L}_0^c|_{\sqrt[3]{n}}=L_0|_{\sqrt[3]{n}}+O(\frac{1}{n})$$
This is a stronger result than restriction to some finite energy $M$. This equation demands $M$ to be changing according to $n$ and yet have a convergence. If $L_0$ gives an energy smaller than $\sqrt[3]{n}$ to some eigenvector $\prod_{-k \in S} \Psi_{-k} 1$, then it must be shown that vector is inside $\mathcal{W}_n$. This means that for large enough $n$, we have $S \subset [-(n-\frac{1}{2})]$; that needs to be checked due to the less natural embedding used in the proofs of theorems in \hyperref[4.1]{section 4.1}. This is easy to show as if $\sum k < \sqrt[3]{n}$, then obviously there is no $k>n-1/2$ for large enough $n$. Further, we should show that $\widetilde{L}_0^c$ gives the same energy up to an error of $O(\frac{1}{n})$. That would imply that the error has norm at most $O(\frac{1}{n})$ as $\widetilde{L}_0^c|_{\sqrt[3]{n}}$ and $L_0|_{\sqrt[3]{n}}$ share the same eigenvectors in $\mathcal{V}$. Let us therefore estimate the difference 
$$|\frac{(2n+1)}{\pi}(\sum_{-k \in S} \sin(\frac{k\pi}{2n+1}))-\sum_{-k \in S} k|,$$
assuming $\sum_{-k \in S}k < \sqrt[3]{n}$. Using \hyperref[eq18]{(18)},
$$=|-\sum_{-k \in S}\frac{k^3 \pi^2}{6(2n+1)^2}+\text{h.o.t}| \le |\sum_{-k \in S}\frac{k^3 \pi^2}{6(2n+1)^2}|+|\text{h.o.t}|.$$
In general, if the sum $\sum_k x_k=t$ of non-negative numbers $x_k$ is a fixed value $t$, then $\sum_k x_k^j \le t^j$ with equality if and only if one of the numbers is $t$ and the others are zero. This implies that in the above, the maximum happens when $S=\{\sqrt[3]{n}\}$. The $\text{h.o.t}$ is (as a rough estimate) at most $O(\frac{1}{n^2})$ and the first term is exactly $O(\frac{1}{n})$. This finishes the proof of case 1(a).

Before moving to the next case, we need to investigate what ``separates'' the two irreducible modules $\rchi_0$ and $\rchi_\frac{1}{2}$ at the level of the finite spaces $\mathcal{W}_n$. The answer to this question will give some interesting identities that will be of use elsewhere.

$(-1)^F$ commutes with $\widetilde{L}_0^c$ and therefore preserves the vacuum as the eigenspace of the vacuum is one dimensional with energy zero. Further, it is easy to see that $\{(-1)^F,\psi_k\}=0$ and so $\{(-1)^F,\Psi_a\}=0$. Therefore any product of even number of creation operators (giving a vector inside $\rchi_0$) commutes with $(-1)^F$ and any product of odd number of creation operators (giving a vector inside $\rchi_\frac{1}{2}$), anti-commutes with $(-1)^F$. It remains to determine 
in which $\pm 1$ \textit{sector} of $(-1)^F$ the vacuum is. Going back to the previous labelling of creation and annihilation operators by integers $[2n]$, we have the following identities involving the annihilation operators
\begin{align}\label{eq19}
(-1)^F=i^n\prod_{j=1}^{2n} \psi_j\ \ , \ \ i^n \prod_{j=1}^{2n} \psi_j \prod_{k=1}^n \Psi_{-k}=\prod_{k=1}^n \Psi_{-k} .
\end{align}
The first identity  directly from the definition of the Majorana operators in terms of the Pauli operators. 
For the second identity, as $(-1)^F$ preserves the one dimensional image of the product (vacuum), we deduce $(-1)^F \prod_{k=1}^n \Psi_{-k}= \alpha \prod_{k=1}^n \Psi_{-k}$ for some $\alpha \in \{\pm 1\}$. According to the second identity, $\alpha$ must be $1$. To prove it, expand both sides of 
$$i^n \prod_{j=1}^{2n} \psi_j \prod_{k=1}^n \Psi_{-k}= \alpha \prod_{k=1}^n \Psi_{-k}$$
in terms of $\psi_i$s. The monomials in $\psi_i$s are linearly independent. On the RHS, all monomials have at most $n$ terms. Suppose that the coefficient of a term with \textit{less} than $n$ Majorana operators is nonzero. This gives a monomial with \textit{more} than $n$ terms on the left side because of the product $\prod_{j=1}^{2n} \psi_j$, so in the expansion of $\prod_{k=1}^n \Psi_{-k}$, only $n$-monomials will appear. To find $\alpha$, one needs to compare the coefficient of $\psi_1 \ldots \psi_n$ on both sides. This means the coefficients of $\psi_1\ldots \psi_n$ and $\psi_{n+1}\ldots\psi_{2n}$ in the RHS. 

The coefficient of any of the $n$-monomials is the determinant of some matrix. For the first one, it is the determinant of (notice $\Psi_{-k}=\Psi_{2n+1-k}$)

$$\begin{pmatrix}
         i\mu_{1,2n} & i^2 \mu_{2,2n} & i^3 \mu_{3,2n} & \ldots & i^n\mu_{n,2n} \\
         i\mu_{1,2n-1} & i^2 \mu_{2,2n-1} & i^3 \mu_{3,2n-1} & & i^n\mu_{n,2n-1} \\
         i\mu_{1,2n-2} & i^2 \mu_{2,2n-2} & i^3 \mu_{3,2n-2} & &   \\
         \vdots &  &  & & i^n\mu_{n,n+2} \\
         &  & & i^{n-1}\mu_{n-1,n+1} &  i^{n}\mu_{n,n+1}
\end{pmatrix}_{n \times n},$$

which has to be compared to the coefficient of $\psi_{n+1} \ldots \psi_{2n}$, the determinant of
\begin{gather*}
\scalebox{0.89}{$\begin{pmatrix}
         i^{n+1}\mu_{n+1,2n} & i^{n+2} \mu_{n+2,2n} & i^{n+3} \mu_{n+3,2n} & \ldots & i^{2n}\mu_{2n,2n} \\
         i^{n+1}\mu_{n+1,2n-1} & i^{n+2} \mu_{n+2,2n-1} & i^{n+3} \mu_{n+3,2n-1} & & i^{2n}\mu_{2n,2n-1} \\
         i^{n+1}\mu_{n+1,2n-2} & i^{n+2} \mu_{n+2,2n-2} & i^{n+3} \mu_{n+3,2n-2} & &   \\
         \vdots &  &  & & i^{2n}\mu_{2n,n+2} \\
         &  & & i^{2n-1}\mu_{2n-1,n+1} &  i^{2n}\mu_{2n,n+1}
\end{pmatrix}_{n \times n},$}
\end{gather*}
and by using $\mu_{2n+1-t,k}=(-1)^{k+1}\mu_{t,k}$, we obtain $\alpha=1$. 

Hence, the vacuum is in the $+1$ sector of $(-1)^F$. Therefore, for all $n$, $\rchi_0$ is in the $+1$ sector and $\rchi_\frac{1}{2}$ is in the $-1$ sector.

The procedure of taking scaling limit after the diagonalization and estimating the rate of convergence of the Hamiltonian to $L_0$ in all future cases will be similar and we will refer to this case. The focus will be only on the parts that have a different idea/explanation.

\subsubsection*{\textbf{Case 1(b) \& 1(c): $(0,0) \& (1,1) \xrightarrow{SL} \rchi_0$ and $(1,0) \& (0,1) \xrightarrow{SL}\rchi_{\frac{1}{2}}$}}
\hfill \break

Projection under $(-1)^F$ is non local. We would like to have exactly $\rchi_0$ and $\rchi_{\frac{1}{2}}$ in the scaling limit by at most a local projection. The boundary conditions involving $0$ and $1$-spin will provide that.

As in previous case, consider the ``would-be'' Hamiltonian $H=-\sum_{j=1}^{2n-1} t_je_j$ acting on the same Hilbert space, with $t_1=t_{2n-1}=0$ and all other $t_j$s being $1$. The Hilbert space $(\frac{1}{2},\frac{1}{2})$ is the sum of four subspaces given by anyonic chains starting and ending with the following spins: $(0,0),(1,0),(0,1),(1,1)$. It is easy to see that $H$ with $t_1=t_{2n-1}=0$ preserves each of the four subspaces and so, restricted to any of those subspaces, it gives the operator derived in \hyperref[1.4]{section 1.4} for the four boundary conditions given by $0$ and $1$. The matrix \hyperref[eq15]{(15)} corresponding to $H$ is
\begin{align}\label{eq20}
2\begin{pmatrix}
         0 & 0 & 0 & \ldots &  \\
         0 & 0 & 1 & & & \\
         0 & 1 & 0 & & &  \\
        \vdots & & & & & \\
        & & & 0 & 1 & 0 \\
         & &  & 1 & 0 & 0 \\
         & &  & & 0 & 0 
\end{pmatrix}_{2n \times 2n}
\end{align}
from which $n-1$ creation operators and annihilation operators $\Psi_k$s are derived corresponding to the same matrix of the case 1(a) for $n \rightarrow n-1$. 

These creation and annihilation operators change the boundary conditions as each operator is a linear combination of $\psi_{j}$ for $2 \le j \le 2n-1$ and, according to their definition, all have the Pauli operator $\sigma_1^x$ and do not have $\sigma_{n}^x$. Therefore, they flip the left boundary condition and pair the condition $(0,1)$ with $(1,1)$ and $(1,0)$ with $(0,0)$. 

Take the image of the product of all annihilation operators. As there are $n-1$ operators acting on a Hilbert space with dimension $2^{n}$, by \hyperref[fct1]{\textit{Fact 1}}, a two dimensional vacuum space exists. As all Dirac operators preserve the spaces $(0,1) \oplus (1,1)$ and $(1,0) \oplus (0,0)$, both $2^{n-1}$ dimensional, one can pick a unit vector in each space being the vacuum. Therefore, diagonalizing the Hamiltonian is done in exactly the same way by replacing $n$ with $n-1$ in case 1(a), defining $H'$ and showing $D=H-H'=0$ (proving the rate of convergence is also done similarly). The only difference is where the vacuum space $(\mathcal{W}_n)_0$ has $\dim (\mathcal{W}_n)_0=2$ and, as noted in \hyperref[fct2]{\textit{Fact 2}}, $D$ has to be shown to preserve two irreducible representations. This is clear as both $H$ and $H'$ preserve the spaces $(0,1) \oplus (1,1)$ and $(1,0) \oplus (0,0)$. Assume $D=\alpha\textbf{1}$ on the first and $D=\beta \textbf{1}$ one the second subspace. But $H'$ has the same symmetric spectrum on each of these subspaces and any shift $\alpha$ should be paired with another shift $\beta=-\alpha$ to preserve the energy symmetry. Therefore $D=-\alpha\sigma_n^z=-\alpha i^{n-1} \prod_{j=1}^{2n-1}\psi_j$. This is not possible, unless $\alpha=0$, as the monomials in Majorana operators $\psi_j$s are linearly independent and we only have binomials or identity in the expansion of either $H$ or $H'$.

The Hamiltonian preserves the boundary conditions. As the two boundary conditions $(1,1)$ and $(0,0)$ are supposed to give $\rchi_0$, it is reasonable to expect that the two vacuum vectors are inside those spaces. Similar to identity \hyperref[eq19]{(19)}, to show that the image of the product of annihilation operators is in the kernel of $\sigma_1^z-\sigma_n^z$:
$$(\sigma_1^z-\sigma_n^z)\prod_{\text{annihilation}} \Psi_k=0 \Leftrightarrow \sigma_1^z\prod_{\text{annihilation}} \Psi_k = \sigma_n^z \prod_{\text{annihilation}} \Psi_k.$$
As $\sigma_1^z=\psi_1$ and $i^{n-1} \prod_{j=1}^{2n-1}\psi_j=\sigma_n^z$,
$$\Leftrightarrow \prod_{\text{annihilation}} \Psi_k = i^{n-1} \prod_{j=2}^{2n-1}\psi_j \prod_{\text{annihilation}} \Psi_k.$$
The argument used for proving \hyperref[eq19]{(19)} applies in this case as well by simply replacing $n$ with $n-1$. As a result, taking the vacuum in, e.g. $(0,0)$, and acting on it with even number of creation operators gives a vector inside $(0,0)$ and with odd operators gives a vector inside $(1,0)$. This implies $\rchi_0$ is the scaling limit corresponding to $(0,0)$ (and $(1,1)$), and $\rchi_{\frac{1}{2}}$ corresponds to $(1,0)$ (and $(0,1)$). 

\subsubsection*{\textbf{Case 1(d) $(\frac{1}{2},0) \oplus (\frac{1}{2},1) \xrightarrow{SL} 2\rchi_\frac{1}{16}$}}

\hfill \break

To get $\rchi_\frac{1}{16}$ in the scaling limit, we specialize to the case $t_j=1$ for all $j$ except $t_{2n-1}=0$. This corresponds to the AC with boundary condition $\mathcal{W}_n=(\frac{1}{2},0) \oplus (\frac{1}{2},1)$, which are denoted by $\mathcal{W}_n^0$ and $\mathcal{W}_n^1$ respectively, and each will give a copy of $\rchi_\frac{1}{16}$. In this case, the matrix equation \hyperref[eq15]{(15)} becomes
\begin{align}\label{eq21}
 \begin{pmatrix}
           \mu_1' \\
           \mu_{2}' \\
           \vdots \\
            \\
           \mu_{2n-1}'\\    
           \mu_{2n}'
         \end{pmatrix}
= 2\begin{pmatrix}
     0 & 1 & 0 & \ldots &  \\
     1 & 0 & 1 & & & \\
     0 & 1 & 0 & & &  \\
     \vdots & & & & & \\
     & & & 0 & 1 & 0 \\
     & &  & 1 & 0 & 0 \\
     & &  & & 0 & 0 
\end{pmatrix}_{2n \times 2n}
         \begin{pmatrix}
         \mu_1 \\
         \mu_{2} \\
         \vdots \\
            \\
         \mu_{2n-1}\\
         \mu_{2n}
         \end{pmatrix}.
\end{align}

The matrix has $2n-1$ eigenvalues $\epsilon_k=4\cos(\frac{k\pi}{2n})$ for $k =1,\ldots,n,\ldots,2n-1$ with corresponding eigenvectors $(\mu_{b,k})_b=(\sin(\frac{bk\pi}{2n}))_b$. In addition to those, there is the eigenvalue $0$ with eigenvector $(\delta_{b,2n})_b$.

This gives a set of (normalized) raising and lowering operators $\Psi_{\pm k}$ for $k=1,\ldots,n-1$ (the creation operators) satisfying the usual relations. The $\epsilon_{n}$ will give $\Psi_n$ which also anticommutes with all other operators except its conjugate, which is easily checked (by looking at its corresponding eigenvector), to be exactly $-\Psi_n$,
$$\{\Psi_n,\Psi_n^\dagger\}=\textbf{1} \implies \Psi_n^2=-\frac{1}{2}.$$
Finally, the operator $i^{2n}\psi_{2n}$ corresponding to $(\delta_{b,2n})_b$ (and the zero eigenvalue) anticommutes with all other operators.

The algebra $\mathcal{F}_{n-1}$ acts on a $2^n$ dimensional space. By \hyperref[fct1]{\textit{Fact 1}}, the vacuum space $(\mathcal{W}_n)_0$ created by the product of all annihilation operators is two dimensional. Further, $\mathcal{W}_n^i$ ($i=0,1$) are preserved by the Hamiltonian. In fact, all the $2n-1$ Dirac operators $\Psi_k$ also preserve $\mathcal{W}_n^i$ as the coefficient for the only term containing $\sigma_n^x$ in their linear expansion in terms of the $\psi_j$s, i.e. $\psi_{2n}$, is $\sin(\frac{(2n)k \pi}{2n})=0$. This implies $(\mathcal{W}_n)_0$ splits into two one-dimensional subspaces of $\mathcal{W}_n^0$ and $\mathcal{W}_n^1$.

Hence, by restricting $H$ to $\mathcal{W}_n^i$, one can apply an argument similar to the case 1(a). Let us define $H'$ which also preserves $\mathcal{W}_n^i$s

$$H'=\sum_{k=1}^{n-1} \epsilon_k(\Psi_{+k}\Psi_{-k}-\Psi_{-k}\Psi_{+k}).$$

As $[H',\Psi_k]=2\epsilon_k\Psi_k$, similar to $H$, we conclude that $D=H-H'$ satisfies $[D,\Psi_k]=0$. \hyperref[fct2]{\textit{Fact 2}} implies that $D$ must be a scalar restricted to each $\mathcal{W}_n^i$s as they are both generated by a vacuum vector and $H,H'$ both preserve $\mathcal{W}_n^i$s. If the $C-$symmetry argument is applied as usual, as $C=\prod_i \sigma_i^z \prod_i (\sigma_i^x)^i$, only for even $n$, $C$ preserves $\mathcal{W}_n^i$. Hence, for even $n$, $H$ would have an energy symmetry and so $H=H'$ on each $\mathcal{W}_n^i$ and therefore on the whole $\mathcal{W}_n$. Even if $n$ is odd, a more involved argument is possible but as similar circumstances appear in the periodic chain case, an argument based on the $(-1)^F$ symmetry will be proposed.

As $(-1)^F$ can be easily seen to commute (or anti-commute based on the parity of $n$) with the product of all annihilation operators, $(\mathcal{W}_n)_0$ is preserved by $(-1)^F$. It is similarly preserved by $\Psi_n$. But $(-1)^F$ and $\Psi_n$ anticommutes. Therefore, any eigenvector of $(-1)^F$ in $(\mathcal{W}_n)_0$, by the action of $\Psi_n$, will go to another \textit{nonzero} eigenvector (since $\Psi_n^2=\frac{-1}{2}$) with the opposite eigenvalue.

The two unit eigenvectors $1_{\pm} \in (\mathcal{W}_n)_0$ with corresponding eigenvalue $\pm 1$ of $(-1)^F$ are sent to a scalar multiple of each other by $\Psi_n$. Then, defining $H'$ as before, and noticing that $[H',(-1)^F]=0$, $D$ preserves the sectors. By \hyperref[fct2]{\textit{Fact 2}}, $H$ restricted to any of the $\pm1$ sector of $(-1)^F$ is equal to $H'$ after some shift $\beta_{\pm}$ in each $\pm1$ sector. Hence, $D|_{+1}=\beta_+\textbf{1}$ and $D|_{-1}=\beta_-\textbf{1}$.

Also, it is important to note that $H'$s spectrum in both sector is the same. Indeed, given any index subset $S$ of the creation operators, starting with either $1_{\pm}$ based on the parity of $|S|$ ensures that the product gives a vector in the desired sector with the energy $\sum_{k \in S} \epsilon_k-\sum_{k \not \in S} \epsilon_k$.

Now suppose $\beta_++\beta_-<0$, then the lowest energy $x+\beta_+$ for $H=H'+\beta_+\textbf{1}$ does not have its opposite in $H'+\beta_-\textbf{1}$ in the other sector; if not, then $\exists E$ such that $-x-\beta_+=E+\beta_- \implies -\beta_+-\beta_-=x+E \le 0$ as $x$ is the lowest energy of $H'$. Similarly, $\beta_++\beta_->0$ is ruled out. So $\beta_++\beta_-=0$ and $H=H'+\beta_+(-1)^F$. But then looking at the expansion of $H'$ and $H$ in terms of Majorana operators, both have at most bilinear terms while $(-1)^F=i^{n} \prod \psi_j$ has $2n$ terms. Due to the linear independence of the Majorana monomials, $H=H'$ is the only possibility.

As was mentioned before, $H$ and the creation operators preserve $\mathcal{W}_n^i$ and therefore, one can pick the vacuum vectors $1_n^i \in \mathcal{W}_n^i$. Then, similar to 1(a), after a suitable shift and scaling, the Hamiltonian $\widetilde{L}_0^c$ is constructed
\begin{align}\label{eq22}
\widetilde{L}_0^c=\frac{2n}{\pi}\sum_{k=1}^{n-1} \cos(\frac{k\pi}{2n})\Psi_{+k}\Psi_{-k}+\frac{1}{16}\textbf{1},
\end{align}
the restriction of which to each $\mathcal{W}_n^i$ has eigenvectors $\{ \prod_{k \in S} \Psi_k 1_n^i| \ \forall S \subset [n-1]\}$. The shift $\frac{1}{16}\textbf{1}$ is not the \textit{natural} one, but for computational issues, it is better to have the exact shift. After the renumbering $\Psi_k \rightarrow \Psi_{k-n}$ for $k\neq n$ and $\Psi_n \to i\Psi_0$, one defines the scaling limit vector space in each boundary condition spanned by the orthogonal vectors $\{ \prod_{k \in S} \Psi_k 1_n| \ \forall S \subset \mathbb{N}\}$. Also, similar to 1(a), the scaling limit of $\widetilde{L}_0^c$ is obtained using the Taylor series of the coefficients in 
$$\widetilde{L}_0^c = \frac{2n}{\pi}\sum_{k \in [(n-1)]} \sin(\frac{k\pi}{2n})\Psi_{-k}\Psi_{k}+\frac{1}{16}\textbf{1},$$
which leads to
$$L_0=\sum_{k \in \mathbb{N}} k\Psi_{-k}\Psi_k+\frac{1}{16}\textbf{1}$$
with the desired rate of convergence. This gives the character
$$\prod_{k=1}^{\infty} (1+q^{k})$$
for the scaling limit which is that of $\rchi_\frac{1}{16}$. This finishes the proof for this case.

\subsubsection*{\textbf{Case 1(e) periodic and full CFT}}

\hfill \break

The periodic case, as mentioned in \hyperref[1.6]{section 1.6}, will be diagonalized differently from \cite{koo1994representations}, which involves taking the usual Fourier transform of the Majorana operators to get the Dirac operators. Here, we will continue applying the method in \cite{fendley2014free} and have $\cos()$ and $\sin()$ transform of the Majorana operators.

In the periodic chain with $2n$ TL operators acting on, there will be operators of the form
$$H= -\sum_{j=1}^{n}t_{2j-1}\sigma_j^x-\sum_{j=1}^n t_{2j}\sigma_j^z \sigma_{j+1}^z,$$
where the case 1(a) Pauli-TL relation are used. Rewriting this in the language of the Majorana operators gives
$$H=i \Big(\sum_{a=1}^{2n-1} t_a \psi_{a+1}\psi_a - t_{2n}\psi_1\psi_{2n}(-1)^F \Big).$$
$H$ has the $(-1)^F$ symmetry and divides the spectrum in two $\pm 1$ sectors which we can analyze separately. The method pursued is to first restrict $H$ to one of those sectors so that the sign of $(-1)^F$ is determined, and then extend $H$ in the obvious way to both sectors (as the Majorana operators can be extended). So effectively, two matrices each of which equal to $H$ in one of the $\pm1$ sector will be diagonalized. Then the spectrum at the scaling limit will be easy to find.

We will show that if $n$ is even, the scaling limit is the diagonal full CFT $\rchi_0\overline{\rchi}_0+\rchi_{\frac{1}{2}}\overline{\rchi}_{\frac{1}{2}}+\rchi_{\frac{1}{16}}\overline{\rchi}_{\frac{1}{16}}$ and if $n$ is odd, it is $\rchi_0\overline{\rchi}_{\frac{1}{2}}+\rchi_{\frac{1}{2}}\overline{\rchi}_0+\rchi_{\frac{1}{16}}\overline{\rchi}_{\frac{1}{16}}$. In this case, similar to how \hyperref[eq15]{(15)} was derived, the matrix is
\begin{align}\label{eq23}
   \begin{pmatrix}
           \mu_1' \\
           \mu_{2}' \\
           \vdots \\
            \\
           \mu_{2n}'
         \end{pmatrix}
         = 2\begin{pmatrix}
         0 & t_1 & 0 & \ldots & (-1)^{F+(n+1)}t_{2n} \\
         t_1 & 0 & t_2 & &  \\
         0 & t_2 & 0 & &   \\
         \vdots &  &  &  & t_{2n-1} \\
         (-1)^{F+(n+1)}t_{2n} &  & & t_{2n-1} & 0 
         \end{pmatrix}
         \begin{pmatrix}
         \mu_1 \\
         \mu_{2} \\
         \vdots \\
            \\
         \mu_{2n}
         \end{pmatrix},
\end{align}

where by $(-1)^F$ in the entries, the \textbf{sign} of the operator $(-1)^F$ when restricted to $\pm 1$ sector is considered. From now on, we will specialize to $t_j=1$ for all $j$. There are two cases based on the parity of $n$.
\hfill \break

\textit{Even $n$.}

Restricting to $+1$ sector, the matrix \hyperref[eq23]{(23)} becomes
\begin{align}\label{eq24}
2\begin{pmatrix}
         0 & 1 & 0 & \ldots & -1 \\
         1 & 0 & 1 & &  \\
         0 & 1 & 0 & &   \\
         \vdots &  &  &  & 1 \\
          -1 &  & & 1 & 0 
\end{pmatrix}_{2n\times 2n},
\end{align}
with the corresponding operator being
$$H^{(+1)}=i \Big(\sum_{a=1}^{2n-1} t_a \psi_{a+1}\psi_a - t_{2n}\psi_1\psi_{2n}\Big).$$
As mentioned before, we should think of $H^{(+1)}$ as an operator acting on the $2^n$ dimensional Hilbert space (on both $\pm 1$ sectors), and once the spectrum of $H^{(+1)}$ is found, we will restrict to the $+1$ sector.

The matrix \hyperref[eq24]{(24)} has eigenvalues $\epsilon_k=4\cos(\frac{(2k-1)\pi}{2n})$ for $k \in [2n]$ and there are repetitions. The corresponding eigenvectors are $(\mu_{b,k})_b=(\cos(\frac{(2k-1)b\pi}{2n}))_b$ for $k=1,\ldots,n$ and $(\mu'_{b,k})_b=(\sin(\frac{(2k-1)b\pi}{2n}))_b$ for $k=n+1,\ldots,2n$, where $(\mu_{b,k})_b$ and $(\mu'_{b,2n+1-k})_b$ correspond to the same eigenvalue $$\epsilon_k=4\cos(\frac{(2k-1)\pi}{2n})=4\cos(\frac{(2(2n+1-k)-1)\pi}{2n})=\epsilon_{2n+1-k}.$$

These in turn will give orthogonal eigenvectors constructing an algebra $\mathcal{F}_n$ of (normalized) creation operators $\mathcal{F}_n^+$ for $k = 1,\ldots,\frac{n}{2},\frac{3n}{2}+1,\ldots,2n$ and (normalized) annihilation operators $\mathcal{F}_n^-$ for $k = \frac{n}{2}+1,\ldots,\frac{3n}{2}$ where
$$\Psi_{-k}:=\Psi_{k}^\dagger=\Psi_{n+1-k} \text{ for } 1 \le k \le n,$$
and
$$\Psi_{-k}:=\Psi_{k}^\dagger=\Psi_{2n+1-k} \text{ for } n+1 \le k \le 2n.$$
One can see that the adjoint of operators corresponding to first quadrant ($k \le n/2$) are the ones in the second quadrant with the opposite eigenvalue and for those in the fourth quadrant ($\frac{3n}{2}+1 \le k \le 2n$), the adjoint is the one with the opposite eigenvalue in the third quadrant. 

Next, define
$$H'^{(+1)}= \sum_{k=1}^{\frac{n}{2}} \epsilon_k(\Psi_{+k}\Psi_{-k}-\Psi_{-k}\Psi_{+k}) + \sum_{k=\frac{3n}{2}+1}^{2n} \epsilon_k(\Psi_{+k}\Psi_{-k}-\Psi_{-k}\Psi_{+k}),$$
and as $[D,\Psi_{k}]=0$ for $D=H'^{(+1)}-H^{(+1)}$ for all $k$, $D=\alpha \textbf{1}$ is a scalar by \hyperref[fct2]{\textit{Fact 2}} as there are $n$ Dirac operators acting on $2^n$ dimensional Hilbert space, therefore the vacuum space is one dimensional. The charge conjugation symmetry $C$ always satisfies $C(-1)^F=(-1)^n(-1)^FC$. As $n$ is even, the charge conjugation symmetry applies on the $+1$ sector. Therefore $D|_{+1}=0 \implies \alpha=0 \implies H^{(+1)}=H'^{(+1)}$ on both sectors. Applying a suitable scaling gives the Hamiltonian
\begin{align}\label{eq25}
\widetilde{\mathbb{L}}_0^{c,(+1)}=\frac{n}{4\pi}\Big(\sum_{k=1}^{\frac{n}{2}} \epsilon_k\Psi_{+k}\Psi_{-k}+\sum_{k=\frac{3n}{2}+1}^{2n} \epsilon_k\Psi_{+k}\Psi_{-k}\Big).
\end{align}
Let us restrict to the $+1$ sector. Again, $(-1)^F$ preserves the vacuum but whether the vacuum itself is in the $+1$ sector or $-1$ is important. One has to prove a similar identity like \hyperref[eq19]{(19)} where the product of annihilation operators $\prod_{\frac{L}{2}+1}^{\frac{3L}{2}} \Psi_{k}$ is one side of the equation:
\begin{align}\label{eq26}
i^{2n}(-1)^F \prod_{\text{annihilation}} \Psi_{k}=\prod_{\text{annihilation}} \Psi_{k},
\end{align}
We will have to compute similar determinant of matrices while using the equalities:
$$(-1)^k\mu_{s,k}=\mu'_{s+n,2n+1-k}  \ \ \text{for} \ \  1\le s \le n, \frac{n}{2}+1 \le k \le n,$$  $$(-1)^{k+1}\mu_{s,k}=\mu'_{s-n,2n+1-k} \ \ \text{for} \ \ n+1\le s \le 2n, \frac{n}{2}+1 \le k \le n,$$
which can be compactly presented as
$$(-1)^k\cos(\frac{(2k-1)s\pi}{2n})=\sin(\frac{(n+s)(2(2n+1-k)-1)\pi}{2n})=$$
$$\sin(\frac{(n-s)(2(2n+1-k)-1)\pi}{2n}).$$
This means that for the two matrices, $i$-th row from one matrix will be equal to the \textbf{opposite} $(n+1-i)$-row on the other matrix up to a $(-1)^k$ factor. 

Since $i^{2n}=(-1)^n$, for $n$ even, the vacuum is in the $+1$ sector. As the eigenvectors of $\widetilde{\mathbb{L}}_0^{c}$ are $\{\prod_{k\in S}\Psi_k1_n| \ \forall S \subset \text{indices of creation operators}\}$ with $1_n$ the vacuum, and $(-1)^F$ anticommutes with all creation operators, the eigenvectors of interest are those with even number of creation operators. 

We shall call all operators with index $1\le k\le n$ the left-moving (LM) operators and $n+1\le k\le 2n$ the right-moving (RM) operators. If one takes odd number of operators from the LM part (giving us some energy in $\mathbb{N}-\frac{1}{2}$), then odd number of operators from the RM part should be taken as well and the same for even. Therefore the character of the scaling limit will be $\rchi_0\overline{\rchi}_0+\rchi_{\frac{1}{2}}\overline{\rchi}_{\frac{1}{2}}$. Scaling limit can be derived using the Fourier series and the rate of convergence can be proved similar to 1(a). Clearly, the relabelling will be $\Psi_k \to \Psi_{k-\frac{n+1}{2}}$ for the LM and $\Psi_k \to \overline{\Psi}_{\frac{3n+1}{2}-k}$ for the RM part, giving us
$$\widetilde{\mathbb{L}}_0^{c,(+1)}=\frac{n}{\pi}\Big(\sum_{k \in [(\frac{n}{2}-\frac{1}{2})]} \sin(\frac{k\pi}{n})\Psi_{-k}\Psi_{k}+\sum_{k\in [(\frac{n}{2}-\frac{1}{2})]} \sin(\frac{k\pi}{n})\overline{\Psi}_{-k}\overline{\Psi}_{k}\Big),$$
with scaling limit
$$\mathbb{L}_0|_{+1}=L_0|_{+1}+\overline{L}_0|_{+1}=(\sum_{k \in \mathbb{N}-\frac{1}{2}} k\Psi_{-k}\Psi_{k}+\sum_{k\in \mathbb{N}-\frac{1}{2}} k\overline{\Psi}_{-k}\overline{\Psi}_{k}\Big)$$

This finishes the proof for the sector $+1$ and even $n$.

For the $-1$ sector, we have the matrix
\begin{align}\label{eq27}
\begin{pmatrix}
         0 & 1 & 0 & \ldots & 1 \\
         1 & 0 & 1 & &  \\
         0 & 1 & 0 & &   \\
         \vdots &  &  &  & 1 \\
          1 &  & & 1 & 0 
\end{pmatrix},
\end{align}
with the corresponding operator, extended to act on both sectors, given by
$$H^{(-1)}=i \Big(\sum_{a=1}^{2n-1} t_a \psi_{a+1}\psi_a + t_{2n}\psi_1\psi_{2n}\Big).$$
Matrix \hyperref[eq27]{(27)} has eigenvalues $\epsilon_k=4\cos(\frac{2k\pi}{2n})$ for $k=1,\ldots,2n$ or equivalently $k=0,\le,2n-1$, where the eigenvalues corresponding to $k=1,\ldots,\frac{n}{2}$ are repeated twice and the one corresponding to $k=n,0$ are repeated once. The corresponding eigenvectors are $(\mu_{b,k})_b=(\cos(\frac{2kb\pi}{2n}))_b$ for $k=0, \ldots, n$ and another set of eigenvectors $(\mu'_{b,k})_b=(\sin(\frac{2kb\pi}{2n}))_b$ for $k=n+1,\ldots,2n-1$. Note that $(\mu_{b,k})_b$ and $(\mu'_{b,2n-k})_b$ are eigenvectors for the same eigenvalue as long as $k \neq n,0$. The corresponding (normalized) Dirac operators are
$$\Psi_{-k}:=\Psi_{k}^\dagger=\Psi_{n-k} \text{ for } 0 \le k \le n \ \& \ k\neq \frac{n}{2},$$
$$\Psi_{-k}:=\Psi_{k}^\dagger=\Psi_{2n-k} \text{ for } n+1 \le k \le 2n-1 \ \& \ k\neq \frac{3n}{2}.$$
Similar to the previous case, LM creation operators are in the first quadrant ($0 \le k \le \frac{n}{2}-1$), and the LM annihilation operators (adjoint to the first quadrant) in the second quadrant. The RM creation operators belong to the fourth quadrant ($\frac{3n}{2}+1 \le k \le 2n-1$) with their adjoint in the third quadrant. Also similar to the case of $\rchi_{\frac{1}{16}}$, we have the (``would-be'' zero-mode) operators $\Psi_{\frac{n}{2}},\Psi_{\frac{3n}{2}}$ corresponding to $\epsilon_{\frac{n}{2}}=\epsilon_{\frac{3n}{2}}=0$ with their adjoint being $\Psi_{\frac{n}{2}},-\Psi_{\frac{3n}{2}}$ respectively. All operators anticommute except with their adjoint, with which they give the identity. Hence, $\Psi_{\frac{n}{2}}^2=\frac{1}{2}=-\Psi_{\frac{3n}{2}}^2$.

Summing up, there are $\frac{n}{2}+(\frac{n}{2}-1)=n-1$ creation operators and the situation is same as $\rchi_\frac{1}{16}$. We have two vectors in the image of the product of all $n-1$ annihilation operators. As $(-1)^F$ preserve that image (since it anti commutes with all Dirac operators) but also since it anticommutes with $\Psi_{\frac{n}{2}}$ (which also preserves the vacuum space for the same reason), there are eigenvectors $1_{n}^\pm$ in each $\pm1$ sector in the vacuum space. Next, defining
$$H'^{(-1)}=\sum_{k=1}^{\frac{n}{2}-1} \epsilon_k(\Psi_{+k}\Psi_{-k}-\Psi_{-k}\Psi_{+k})+\sum_{k=\frac{3n}{2}+1}^{2n} \epsilon_k(\Psi_{+k}\Psi_{-k}-\Psi_{-k}\Psi_{+k}),$$
it can be shown that $[H^{(-1)}-H'^{(-1)},\Psi_k]=0$. Further, both operators commute with $(-1)^F$, so $D=H^{(-1)}-H'^{(-1)}$ preserves both sectors and acts as a scalar on each. As $n$ is even, we have charge conjugacy in each sector, therefore $D=0$. It then becomes clear that one must define $\widetilde{\mathbb{L}}_0^{c,(-1)}$ as
\begin{align}\label{eq28}
\widetilde{\mathbb{L}}_0^{c,(-1)}=\frac{n}{4\pi}\Big(\sum_{k=0}^{\frac{n}{2}-1} \epsilon_k\Psi_{+k}\Psi_{-k}+\sum_{k=\frac{3n}{2}+1}^{2n-1} \epsilon_k\Psi_{+k}\Psi_{-k}\Big),
\end{align}
where we note that a shift by some scalar (which will be $\frac{1}{8}$ in the limit) is not included yet and will be discussed later. As we are interested in the $-1$ sector, all combinations of the form
\begin{itemize}
    \item (odd LM) (odd RM) $1_{n}^-$
    \item (even LM) (even RM) $1_{n}^-$
    \item (odd LM) (even RM) $1_{n}^+$
    \item (even LM) (odd RM) $1_{n}^+$
\end{itemize}
are in the subspace the scaling limit should be taken. To do so, one needs to first apply the renumbering $\Psi_{k} \to \Psi_{k-\frac{n}{2}}$ for $0 \le k \le n$ and $\Psi_{k} \to \overline{\Psi}_{\frac{3n}{2}-k}$ for $n+1 \le k \le 2n-1$ except for $\Psi_\frac{3n}{2} \to i\overline{\Psi}_0$, and accordingly
$$\widetilde{\mathbb{L}}_0^{c,(-1)}=\frac{n}{\pi}\Big(\sum_{k \in [\frac{n}{2}]} \sin(\frac{k\pi}{n})\Psi_{-k}\Psi_{k}+\sum_{k \in [\frac{n}{2}]} \sin(\frac{k\pi}{n})\overline{\Psi}_{-k}\overline{\Psi}_{k}\Big).$$

In order to build the scaling limit vector space $\mathcal{V}$, the four possibilities outlined above need to be considered. The embeddings $\phi_n : \mathcal{W}_n^{(-1)} \hookrightarrow \mathcal{W}_{n+2}^{(-1)}$, done by mapping each vector to its obvious corresponding vector (also consistent with the embeddings of the Dirac operators algebra), gives already the character of $\overline{\mathbb{L}}_0^{c,(-1)}$ in the scaling limit as that of $\rchi_\frac{1}{16}\overline{\rchi}_\frac{1}{16}$ but we need to identify the scaling limit with the \textbf{space}  $\rchi_\frac{1}{16}\overline{\rchi}_\frac{1}{16}$. This is also clear as every configuration above can be identified with its counterpart in $\rchi_\frac{1}{16}\overline{\rchi}_\frac{1}{16}$
\begin{itemize}
    \item (odd LM) (odd RM) $\ket{\frac{1}{16}}\otimes\ket{\frac{1}{16}}$
    \item (even LM) (even RM) $\ket{\frac{1}{16}}\otimes\ket{\frac{1}{16}}$
    \item (odd LM) (even RM) $\ket{\frac{1}{16}}\otimes\ket{\frac{1}{16}}$
    \item (even LM) (odd RM) $\ket{\frac{1}{16}}\otimes\ket{\frac{1}{16}}$
\end{itemize}
This map is unitary and ``character''-preserving. Notice that only $1^-$ (the scaling limit of $1_{n}^-$) is identified with $\ket{\frac{1}{16}}\otimes\ket{\frac{1}{16}}$. The vector $1^+$ (the scaling limit of $1_{n}^+$) is not identified with anything inside $\rchi_\frac{1}{16}\overline{\rchi}_\frac{1}{16}$ as it is not even present in the finite spaces $\mathcal{W}_n^{(-1)}$. What we observe, is a merging of two copies of $\rchi_\frac{1}{16}\overline{\rchi}_\frac{1}{16}$ (generated by $1^+$ and $1^-$)  and selection of a subspace of both, which together form a copy of $\rchi_\frac{1}{16}\overline{\rchi}_\frac{1}{16}$. Finally, the scaling limit Hamiltonian is
$$\mathbb{L}_0|_{-1}=L_0|_{-1}+\overline{L}_0|_{-1}=(\sum_{k \in \mathbb{N}} k\Psi_{-k}\Psi_{k}+\sum_{k\in \mathbb{N}} k\overline{\Psi}_{-k}\overline{\Psi}_{k}\Big).$$
Let us discuss the issue of the scalings done to $H$ in both sectors. They can be seen to be clearly different. In fact, after the restriction to each sector and proving that $H=H'$, there was a scaling 
$$H \rightarrow \frac{n}{8\pi}(H+\sum_{k=1}^{\frac{n}{2}}\epsilon_k + \sum_{k=\frac{3n}{2}+1}^{2n}\epsilon_k)$$ 
for the $+1$ sector and the other being 
$$H \rightarrow \frac{n}{8\pi}(H+\sum_{k=0}^{\frac{n}{2}-1}\epsilon_k + \sum_{k=\frac{3n}{2}+1}^{2n-1}\epsilon_k).$$ 
There can only be one scaling to $H$. So the different scalings should differ by some scalar which is $\frac{1}{8}$ in the scaling limit as the ground state in the $(-1)^F=-1$ sector is $\ket{\frac{1}{16}}\otimes \ket{\frac{1}{16}}$. Indeed, by some trigonometric calculation, the two scalings differ by $\frac{n}{8\pi}(4\tan(\frac{\pi}{4n}))=\frac{1}{8}+O(\frac{1}{n^2})$ which in the limit $n\rightarrow \infty$ is $\frac{1}{8}$. Hence, the scaling limit for even $n$ is the diagonal full CFT with the Hamiltonian $\mathbb{L}_0$. 
\hfill \break

\textit{Odd $n$.}

This case will not be analyzed as it is not a full diagonal CFT and not used in any of the main results. Nevertheless, the proof can be seen to be similar to the $n$ even case, with the difference that the $+1$ sector gives $\rchi_\frac{1}{16}\overline{\rchi}_\frac{1}{16}$ and the $-1$ sector gives $\rchi_0\overline{\rchi}_\frac{1}{2}+\rchi_{\frac{1}{2}}\overline{\rchi}_{0}$. The reason the $-1$ sector is not diagonal is the fact that odd(even) number of LM operators have to act with even(odd) numbers of RM operators to take the vacuum from the $+1$ sector to $-1$ sector.

\subsubsection*{\textbf{Case 2; The higher Virasoro modes $L_m$s}}

\hfill \break

Changing the coefficients $t_j$ to a $\cos()$ and $\sin()$ tranform of the $e_j$s is necessary to obtain the higher Virasoro modes $L_m$s.

We will prove the case 2(a) ($\rchi_0+\rchi_{\frac{1}{2}}$) in \hyperref[thm3.1]{\textbf{Theorem 3.1}} with the rate of convergence. All other cases, including the periodic case, have a similar proof although there will be some comments for the periodic chain.
\hfill \break

\textit{The $\cos()$ transform and $\widetilde{L}_m^c+\widetilde{L}_{-m}^c$.}

\hfill \break
Let us fix $m \in \mathbb{N}$. The operator $L_{m}+L_{-m}$ is given by (see e.g. \cite{francesco2012conformal})
\begin{align}\label{eq29}
\sum_{k \ge \frac{m+1}{2},k \in \mathbb{Z}+\frac{1}{2}} (k-\frac{m}{2}) \Psi_{-k+m}\Psi_{k}+\sum_{k \ge \frac{-m+1}{2},k \in \mathbb{Z}+\frac{1}{2}} (k+\frac{m}{2}) \Psi_{-k-m}\Psi_{k},
\end{align}
which we want to obtain in the scaling limit. To understand what the observable $O=i\sum t_j(m)\psi_{j+1}\psi_{j}$ will be in terms of $\Psi_k$'s, one has to use the matrix equation \hyperref[eq15]{(15)}. We need to build another observable $O'$ using $\Psi_k$'s which has scaling limit $L_{m}+L_{-m}$, and that also satisfies $[O-O',\Psi_k]=0$. Then, going through the usual arguments, after some suitable scaling, the sequence $O_n \xrightarrow{SL} L_m+L_{-m}$ will be constructed.

Notice that $[L_m+L_{-m},\Psi_k]$ is the sum of exactly two Dirac operators with indices differing by $m$ from $k$. $t_j(m)$ should be such that the same result for $[O,\Psi_k]$ happens, with coefficients going to $k\pm \frac{m}{2}$ in the scaling limit. Using the indices \textit{before} the renumbering, i.e. $k \in [2n]$, a natural candidate for the coefficients would be $\cos(\frac{(k\mp \frac{m}{2})\pi}{2n+1})$. Hence, computing $[O,\Psi_k]$ using \hyperref[eq15]{(15)}, the following must hold
\begin{align}\label{eq30}
t_j(m)\mu_{k,j+1}+t_{j-1}(m)\mu_{k,j-1}=
\end{align}
$$\cos \Big(\frac{(k+\frac{m}{2})\pi}{2n+1}\Big)\mu_{k+m,j}+\cos \Big(\frac{(k-\frac{m}{2})\pi}{2n+1}\Big)\mu_{k-m,j}.$$
From simple trigonometric identities, the right side is equal to
\begin{align}\label{eq31}
\cos \Big(\frac{(k+\frac{m}{2})\pi}{2n+1}\Big)\sin \Big(\frac{(k+m)j\pi}{2n+1} \Big)+\cos \Big(\frac{(k-\frac{m}{2})\pi}{2n+1}\Big)\sin \Big(\frac{(k-m)j\pi}{2n+1} \Big) 
\end{align}
$$=\cos(\frac{m(j+\frac{1}{2})\pi}{2n+1})\sin \Big(\frac{k(j+1)\pi}{2n+1}\Big)+\cos(\frac{m(j-\frac{1}{2})\pi}{2n+1})\sin \Big(\frac{k(j-1)\pi}{2n+1} \Big),$$
which is in fact
$$\cos(\frac{m(j+\frac{1}{2})\pi}{2n+1})\mu_{k,j+1}+\cos(\frac{m(j-\frac{1}{2})\pi}{2n+1})\mu_{k,j-1}.$$
So $t_j(m)$ are \textit{forced} to be $\cos(\frac{m(j+\frac{1}{2})\pi}{2n+1})$. However, the coefficients $\cos(\frac{mj\pi}{2n})$, as used in the conjecture \cite[(7.5)]{koo1994representations}, do not satisfy the identity \hyperref[eq30]{(30)}. The implications will be discussed more in the second subsection. As the actual identity for the matrix \hyperref[eq15]{(15)} involves a two factor, $O$ is changed to $O/2$ to cancel this factor. This is important to compute the scaling factors until $O$ becomes the desired operator.

Although the identities above determine what $O'$ should be, what happens at the \textit{boundaries} when $k+m>2n$ or $k-m<1$ must be examined. In these cases, one has to consider $\sin(\frac{(k+m)j\pi}{2n+1})=-\sin(\frac{(2(2n+1)-k-n)j\pi}{2n+1})=-\mu_{2(2n+1)-(k+m),j}$ if $k+m>2n$ and $\sin(\frac{(k-m)j\pi}{2n+1})=-\sin(\frac{(m-k)j\pi}{2n+1})=-\mu_{m-k,j}$ if $k-m<1$. Therefore, $O'$ is defined as
\begin{gather}\label{eq32}
\scalebox{0.93}{$
\Big(\sum\limits_{k+m \le 2n} \cos\Big(\frac{(k+\frac{m}{2})\pi}{2n+1}\Big) \Psi_{k+m}\Psi_{k}^\dagger-\sum\limits_{k+m>2n} \cos\Big(\frac{(k+\frac{m}{2})\pi}{2n+1}\Big)\Psi_{2(2n+1)-k-m}\Psi_{k}^\dagger \Big)$}
\end{gather}
\begin{gather*}
\scalebox{0.93}{$+\Big(\sum\limits_{k-m \ge 1} \cos\Big(\frac{(k-\frac{m}{2})\pi}{2n+1}\Big)\Psi_{k-m}\Psi_{k}^\dagger-\sum\limits_{k-m<1} \cos\Big(\frac{(k-\frac{m}{2})\pi}{2n+1}\Big)\Psi_{m-k}\Psi_{k}^\dagger \Big).$}
\end{gather*}
For $k$ close to the boundaries, i.e. $1$ or $2n$ (which corresponds to high energy creation/annihilation), another operator of high energy is associated. Therefore, there are no low-high energy mix in the above formula. There are also no terms like $\Psi_k^\dagger \Psi_k$ as $m \neq 0$. Further, similar to the case of $L_0$, where $H'$ had terms like $(\Psi_k\Psi_k^\dagger-\Psi_k^\dagger \Psi_k)$ for each $\Psi_k$ with $1 \le k \le n$, there are also two terms in \textit{each} parenthesis in the summation above which have $\Psi_k^\dagger$. Indeed, one is $\Psi_{k+m}\Psi_k^\dagger$ and the other is $\Psi_{2n+1-k}\Psi_{2n+1-k-m}^\dagger=\Psi_{k}^\dagger\Psi_{k+m}$, and we have
\begin{gather*}
\scalebox{0.99}{$\cos(\frac{(k+\frac{m}{2})\pi}{2n+1})\Psi_{k+m}\Psi_k^\dagger=\cos(\frac{((2n+1-k-m)+\frac{m}{2})\pi}{2n+1})\Psi_{2n+1-k}\Psi_{2n+1-k-m}^\dagger.$}
\end{gather*}
Let us prove $O=O'$. The equation for $O'$ implies $[O-O',\Psi_k]=0$, meaning that $O-O'$ is a scalar. To prove that the shift is zero, we use the fact that monomials in Majorana operators are linear independent. The left side in $O=O'+\alpha\textbf{1}$ is bilinear, but the right side might have scalar terms coming from the multiplication of $\psi_j\psi_j=1$ when for $O'$ is expanded in terms of the $\psi_j$s. Those terms have to cancel each other so that a non-zero shift $\alpha$ will not be needed. In fact, it can be shown that the scalar from \textit{each} of $\Psi_{k+n}\Psi_k^\dagger$ is zero. Indeed, calculating the contribution from $\Psi_{k+n}\Psi_k^\dagger$, 
$$\sum_{j=1}^{2n} i^j\sin(\frac{(k+m)j\pi}{2n+1}).i^j\sin(\frac{(2n+1-k)j\pi}{2n+1})=$$
$$\sum_{j=1}^{2n} i^{2j}\sin(\frac{(k+m)j\pi}{2n+1}).(-1)^{j+1}\sin(\frac{kj\pi}{2n+1})=$$
$$-\sum_{j=1}^{2n} \sin(\frac{(k+m)j\pi}{2n+1}).\sin(\frac{kj\pi}{2n+1})=0,$$
the last equality holds as they are eigenvectors of different eigenvalues of matrix \hyperref[eq15]{(15)} for $t_j=1$. Hence, $O=O'$. To find the scaling limit, taking note of the fact that each term is repeated twice, we have the following operator
\begin{align}\label{eq33}
\widetilde{L}_m^c+\widetilde{L}_{-m}^c=\frac{2n+1}{2\pi}O'
\end{align}
where $O'$ is defined in \hyperref[eq32]{(32)} and $\widetilde{L}_m^c$ is the first, and $\widetilde{L}_{-m}^c$ corresponds to the second parenthesis. Further the factor $2$ in the denominator is due to the repetition of each term. As the proof for the convergence rate is similar, only $\widetilde{L}_{m}^c$ will be done. \hyperref[thm3.1]{\textbf{Theorem 3.1}} states that
$$\widetilde{L}_{m}^c|_{\sqrt[4]{n}}=L_{m}|_{\sqrt[4]{n}}+O(\frac{1}{n}),$$
where $m \le \sqrt[4]{n}$ can also depend on $n$. By applying the energy restriction up to $\sqrt[4]{n}$, for large enough $n$, it is not hard to observe that all the bilinear terms of Dirac operators in $\widetilde{L}_{m}^c|_{\sqrt[4]{n}}$ and $L_{m}|_{\sqrt[4]{n}}$ which will be nonzero operators after the energy restriction are the same. Indeed, most of the bilinear Dirac terms in both of the operators, are composed of an annihilation and a creation operator. If the annihilation operator annihilates energy more than $\sqrt[4]{n}$ then its restriction to vectors with energy at most $\sqrt[4]{n}$ is clearly zero. Lastly, for large enough $n$, as $2\sqrt[4]{n} < < n$, there are no terms close to the \textit{boundaries} (with index close to $1$ or $2n$) and so, all terms inside of $L_{m}|_{\sqrt[4]{n}}$ are also present inside $\widetilde{L}_{m}^c|_{\sqrt[4]{n}}$ and vice-versa, although with a different coefficient which their difference shall be estimated. After the renumbering $\Psi_k \to \Psi_{k-n-\frac{1}{2}}$, the term $\cos(\frac{(k+\frac{m}{2})\pi}{2n+1})\Psi_{k+m}\Psi_{k}^\dagger$ becomes 
$$\sin(\frac{(k-\frac{m}{2})\pi}{2n+1})\Psi_{-k+m}\Psi_{k}.$$
By taking the difference in the coefficient, we have
$$|\frac{2n+1}{\pi}\sin(\frac{(k-\frac{m}{2})\pi}{2n+1})-(k-\frac{m}{2})|\le O(\frac{(\sqrt[4]{n})^3}{(2n+1)^2}),$$
as $k,m$ are of order $\sqrt[4]{n}$ due to the discussion in the previous paragraph. Next, since there are $O(\sqrt[4]{n})$ of these differences (as $k$ can vary), and each Dirac bilinear term has norm at most one,
$$||\widetilde{L}_{m}^c|_{\sqrt[4]{n}}-L_{m}|_{\sqrt[4]{n}}|| \le O(\frac{(\sqrt[4]{n})^3.\sqrt[4]{n}}{(2n+1)^2})=O(\frac{1}{n}),$$
finishing the proof of the convergence rate. This obviously implies $\widetilde{L}_m^c \xrightarrow{SL} L_m$. The very similar reasoning can be made for all the other ACs. But we would like to note that when changing the boundary condition, the denominators in the fractions inside the trigonometric functions must be changed (e.g. in the $\rchi_\frac{1}{16}$ case, one needs $2n$ instead of $2n+1$). The changes happen the most in the periodic case and we  will have to apply the following
\begin{itemize}
    \item replace $2n+1$ by $2n$,
    \item replace $m$ in $t_j(m)$ by $2m$,
    \item replace $k$ by $2k$ (or $2k-1$),
\end{itemize}

where the change to even numbers $2k$ is for $\rchi_{\frac{1}{16}}\overline{\rchi}_{\frac{1}{16}}$ and odd number for $\rchi_0\overline{\rchi}_0+\rchi_{\frac{1}{2}}\overline{\rchi}_{\frac{1}{2}}$. One also has to check that the trigonometric identities like \hyperref[eq31]{(31)} holds for all the eigenvectors of the two different matrices (they are of two types: $\mu_i$ and $\mu'_i$ corresponding to $\cos()$ and $\sin()$). As an example, for the LM part (with $(\mu)_b=(\cos())_b$ eigenvectors) of $\rchi_0\overline{\rchi}_0+\rchi_{\frac{1}{2}}\overline{\rchi}_{\frac{1}{2}}$, the corresponding identity is:
\begin{gather*}
\scalebox{0.97}{$\cos\Big(\frac{(2m)(j+\frac{1}{2})\pi}{2n}\Big)\cos \Big(\frac{(2k-1)(j+1)\pi}{2n}\Big)+\cos\Big(\frac{(2m)(j-\frac{1}{2})\pi}{2n}\Big)\cos \Big(\frac{(2k-1)(j-1)\pi}{2n} \Big)= $}\\
\scalebox{0.97}{$\cos \Big(\frac{(2k-1+m)\pi}{2n}\Big)\cos \Big(\frac{(2k-1+2m)j\pi}{2n} \Big)+\cos \Big(\frac{(2k-1-m)\pi}{2n}\Big)\cos \Big(\frac{(2k-1-2m)j\pi}{2n} \Big)$}
\end{gather*}
and there is a similar identity when the second $\cos()$ terms in each product are replaced with $\sin()$ for the RM operators. This gives the operator $\widetilde{\mathbb{L}}_m^c+\widetilde{\mathbb{L}}_{-m}^c$ and the convergence to $\mathbb{L}_m+\mathbb{L}_{-m}$ can be proved in a similar way. 
\hfill \break

\textit{The $\sin()$ transform and $i(\widetilde{L}_m^s-\widetilde{L}_{-m}^s)$.}

\hfill \break
Having found the sum of the higher Virasoro modes, their difference is required to recover an operator $\widetilde{L}_m \xrightarrow{SL} L_m$, as shown in the remark after \hyperref[thm3.1]{\textbf{Theorem 3.1}}. Again, this case will be demonstrated for $\rchi_0+\rchi_\frac{1}{2}$. As the proof to the convergence and scaling limit is very similar to the previous $\cos()$ transform, we will refer to the previous arguments for that part of the problem.
 
We need the $\sin()$ transform of the form $O= -i\sum_{j=1}^{2n-2} t_j(n)[e_j,e_{j+1}]$ which is same as $O=i\sum_j t_j(n)\psi_{j}\psi_{j+2}$. The corresponding matrix for $[O,\Psi]$ where $\Psi=\sum i^b\mu_b \psi_b$ can be found as follows:
$$[O,\Psi]=\Psi'=i(\sum_b i^b\mu'_b\psi_b)$$
where the $i$ factor is needed to obtain $i(L_m-L_{-m})$, and we have
$$\mu'_b=2(t_b(m)\mu_{b+2}-t_{b-2}(m)\mu_{b-2})$$
except at the boundaries where the formula will be different. In the case of the non-periodic chain $(\frac{1}{2},\frac{1}{2})$, this can be turned into the corresponding matrix
\begin{align}\label{eq34}
\begin{pmatrix}
         0 & 0 & t_1(m) & & & 0 \\
         0 & 0 & 0 & t_2(m) & & \\
         -t_1(m) & 0 & 0 & & \ddots &  \\
         & -t_2(m) & & & & \\
         & & & & & \\
         & & \ddots & 0 & 0 & t_{2n-2}(m) \\
         & &  & 0 & 0 & 0 \\
        0 & &  & -t_{2n-2}(m) & 0 & 0 
\end{pmatrix}.
\end{align}
The matrix corresponding to the periodic case can be recovered similarly and as expected, it will have entries $(-1)^{n+F}t_{2n-1}(m),(-1)^{n+F}t_{2n}(m)$ and their opposite in the diagonals in the corners of the matrix.

First, an identity similar to \hyperref[eq31]{(31)} only with $\sin()$ functions is needed to \textit{force} the values for $t_j(m)$. Since the coefficients $(k\pm \frac{m}{2})$ needs to be recovered in the limit using $\sin()$ functions, the natural candidate for the coefficients of the bilinear Dirac terms are $\sin(\frac{(2k+m)\pi}{2n+1})$. This can be seen to only work when $t_j(m)=\sin(\frac{m(j+1)\pi}{2n+1})$, satisfying the following identity
$$t_j(m)\mu_{k,j+2}-t_{j-2}(m)\mu_{k,j-2}=$$
\begin{gather}\label{eq35}
\scalebox{0.999}{$
\sin\Big(\frac{m(j+1)\pi}{2n+1}\Big)\sin \Big(\frac{k(j+2)\pi}{2n+1}\Big)-\sin\Big(\frac{m(j-1)\pi}{2n+1}\Big)\sin\Big(\frac{k(j-2)\pi}{2n+1} \Big)=$}
\end{gather}
\begin{gather*}
\scalebox{0.999}{$\sin \Big(\frac{(2k+m)\pi}{2n+1}\Big)\sin \Big(\frac{(k+m)j\pi}{2n+1} \Big)-\sin \Big(\frac{(2k-m)\pi}{2n+1}\Big)\sin \Big(\frac{(k-m)j\pi}{2n+1} \Big).$}
\end{gather*}
It can be checked that at the boundaries $j=1,2,2n-1,2n$ the first equality above still holds since $\sin(0)=\sin(k\pi)=\sin(m\pi)=0$. Repeating the argument in the previous case, we set $-iO'=$
\begin{gather}\label{eq36}
\scalebox{0.93}{$
\Big(\sum\limits_{k+m \le 2n} \sin\Big(\frac{(2k+m)\pi}{2n+1}\Big) \Psi_{k+m}\Psi_{k}^\dagger-\sum\limits_{k+m>2n} \sin\Big(\frac{(2k+m)\pi}{2n+1}\Big)\Psi_{2(2n+1)-k-m}\Psi_{k}^\dagger \Big)$}
\end{gather}
\begin{gather*}
\scalebox{0.93}{$-\Big(\sum_{k-m \ge 1} \sin\Big(\frac{(2k-m)\pi}{2n+1}\Big)\Psi_{k-m}\Psi_{k}^\dagger-\sum_{k-m<1} \sin\Big(\frac{(2k-m)\pi}{2n+1}\Big)\Psi_{m-k}\Psi_{k}^\dagger \Big),$}
\end{gather*}
And the argument for $O/2=O'$ is exactly the same as before (where the factor $2$ is due to the matrix equation \hyperref[eq15]{(15)} as in the previous case). Therefore, let us set
\begin{align}\label{eq37}
i(\widetilde{L}_{m}^s-\widetilde{L}_{-m}^s)=\frac{2n+1}{2\pi}O',
\end{align}
where $O'$ is defined in \hyperref[eq36]{(36)} and $\widetilde{L}_m^s$ is the first, and $\widetilde{L}_{-m}^s$ corresponds to the second parenthesis. For the coefficients, after the renumbering $\Psi_{k} \to \Psi_{k-n-\frac{1}{2}}$, $\sin(\frac{(k-\frac{m}{2})\pi}{n+\frac{1}{2}})$ is the coefficient of $\Psi_{-k+m}\Psi_k$ in $\widetilde{L}_m^s$,
which means that in the scaling limit, the coefficient for the Dirac terms is as desired. The proof of convergence and its rate, are similar to the $\cos()$ transform. 

As previously, the denominators need to change for a change in boundary condition (e.g. $2n$ instead of $2n+1$ for $\rchi_\frac{1}{16}$). For the periodic chain, the following changes need to performed:
\begin{itemize}
    \item replace $2n+1$ by $2n$,
    \item replace $m$ in $t_j(m)$ by $2m$,
    \item replace $k$ by $2k$ (or $2k-1$),  
\end{itemize}
where the change to even numbers $2k$ is for $\rchi_{\frac{1}{16}}\overline{\rchi}_{\frac{1}{16}}$ and odd number for $\rchi_0\overline{\rchi}_0+\rchi_{\frac{1}{2}}\overline{\rchi}_{\frac{1}{2}}$. It can be checked that the new identity involving the eigenvectors of type $\cos()$ for the $+1$ sector $\rchi_0\overline{\rchi}_0+\rchi_{\frac{1}{2}}\overline{\rchi}_{\frac{1}{2}}$ is
\begin{gather*}
\scalebox{0.9}{$\sin\Big(\frac{(2m)(j+1)\pi}{2n}\Big)\cos \Big(\frac{(2k-1)(j+2)\pi}{2n}\Big)-\sin\Big(\frac{(2m)(j-1)\pi}{2n}\Big)\cos \Big(\frac{(2k-1)(j-2)\pi}{2n} \Big)= $}\\
\scalebox{0.9}{$\sin \Big(\frac{(2(2k-1)+2m)\pi}{2n}\Big)\cos \Big(\frac{((2k-1)+2m)j\pi}{2n} \Big)-\sin \Big(\frac{(2(2k-1)-2m)\pi}{2n}\Big)\cos \Big(\frac{((2k-1)-2m)j\pi}{2n} \Big),$}
\end{gather*}
which will also hold if $\cos()$s are replaced with $\sin()$, giving the identity for the $\sin()$ type eigenvectors. This finishes the proof of \hyperref[thm3.1]{\textbf{Theorem 3.1}}.
\hfill \break
\subsection*{2.\texorpdfstring{ 
$-\sum \sin(n\theta) e_\theta$}{TEXT} or \texorpdfstring{$-i\sum \cos(n\theta) [e_\theta,e_{\theta+1}]$}{TEXT}}\label{A.2}

\hfill \break
It was mentioned in \hyperref[1.6]{section 1.6}, and in the remark after \hyperref[cnj5.5]{\textbf{Conjecture 5.5}}, that the operators involved in a similar conjecture in \cite[(7.5)]{koo1994representations} have undesirable low-high energy mix. These claims and why the ``reasonable'' guess that the $\sin()$ transform of the $e_j$ must give $i(L_m-L_{-m})$ is not true, will be demonstrated by direct calculations. Let us start with the latter.

We start with the nonperiodic chains. As before, the example $\rchi_0+\rchi_{\frac{1}{2}}$ will be used to show the argument which can be applied similarly to other nonperiodic cases. In order to obtain $-\sum_{j=1}^{2n-1} \sin(\frac{m(j+\frac{1}{2})\pi}{2n+1})e_j$  or in other words $O_s=i\sum_{j=1}^{2n-1} \sin(\frac{m(j+\frac{1}{2})\pi}{2n+1})\psi_{j+1}\psi_j$ in terms of the creation and annihilation operators, $\psi_j$ should be expressed in terms of the $\Psi_k$s (the inverse of the transformation given by the eigenvectors of the matrix \hyperref[eq15]{(15)}). In \cite[(24)]{fendley2014free}, it is shown that
$$\psi_{2j}=(-1)^j\sum_{k=1}^n \frac{i}{N_k} \widetilde{Q}_{2j-1}(\epsilon_k^2)(\Psi_{+k}+\Psi_{-k}),$$$$\psi_{2j+1}=(-1)^j\sum_{k=1}^n \frac{1}{N_k} \widetilde{Q}_{2j}(\epsilon_k^2)(\Psi_{+k}-\Psi_{-k}),$$
where un-normalized Dirac operators are considered, and $\widetilde{Q}_{a-1}(\epsilon_{\pm k}^2)=\sin(\frac{ak\pi}{2n+1})$. Using the above, the coefficient of a term $\Psi_x\Psi_{-y}$ in $O_s$ for any $1 \le x,y\le n$ is
\begin{gather*}
\scalebox{0.999}{$\frac{1}{N_xN_y}\sum_{j=1}^{2n-1} \sin(\frac{m(j+\frac{1}{2})\pi}{2n+1})\Big( \sin(\frac{jx\pi}{2n+1})\sin(\frac{(j+1)y\pi}{2n+1})+\sin(\frac{(j+1)x\pi}{2n+1})\sin(\frac{jy\pi}{2n+1})\Big )$}
\end{gather*}
This can be computed by direct computation (substituting $\sin(\theta)=\frac{e^{i\theta}-e^{-i\theta}}{2i}$) and it can be shown that it is zero if and only if $m+x+y \equiv 0 \pmod{2}$. This implies that (even after \textbf{any} scaling of $O_s$), e.g. for $m=4$, there are terms like $\Psi_{1}\Psi_{-(n-1)}$ for $n$ odd and $\Psi_{1}\Psi_{-(n-2)}$ for $n$ even with nonzero coefficient. Therefore, this is not an energy local operator and further, it can not be a candidate for the sequence associated to $L(f)$ to obtain a strong SL-algebra in \hyperref[thm4.6]{\textbf{Theorem 4.6}}. 

As another example with chain size even and with scaling limit $\rchi_{\frac{1}{16}}$, we have the following
\begin{gather*}
\scalebox{0.999}{$\frac{1}{N_xN_y}\sum_{j=1}^{2n-2} \sin(\frac{m(j+\frac{1}{2})\pi}{2n})\Big( \sin(\frac{jx\pi}{2n})\sin(\frac{(j+1)y\pi}{2n})+\sin(\frac{(j+1)x\pi}{2n})\sin(\frac{jy\pi}{2n})\Big )$}
\end{gather*}
which can be computed and it is zero $\Leftrightarrow m+x+y \equiv 0 \pmod{2}$ and the same conclusions follow. There is a similar argument for the $\sin()$ transform of $[e_j,e_{j+1}]$.

The conjecture \cite[(7.5)]{koo1994representations} asserts that $l_m+l_{-m} \xrightarrow{SL} L_m+L_{-m}$ where  $l_m+l_{-m}$ is some scaling of the sum of the terms $\cos(\frac{mj\pi}{2n})e_j$. The difference with \hyperref[cnj5.5]{\textbf{Conjecture 5.5}} is the factor $j$ instead of $\frac{m(j+\frac{1}{2})\pi}{2n+1}$ and the denominator. We recall that we interpret the point $\frac{(j+\frac{1}{2})\pi}{2n+1}$ as the ``center'' of the action of $e_j$ in the half-circle in \hyperref[fig1]{\textit{Figure 1}}. Also, we recall that the identity \hyperref[eq30]{(30)} forced the coefficient $\cos(\frac{m(j+\frac{1}{2})\pi}{2n+1})$, therefore any other coefficient (including $\cos(\frac{mj\pi}{2n})$) should have some \textit{undesirable effect}, as it does not satisfy the identity, although these effects could vanish in the scaling limit.  We can compute the coefficient of a term $\Psi_x\Psi_{-y}$ in $O_s'=i\sum_{j=1}^{2n-1} \cos(\frac{mj\pi}{2n})\psi_{j+1}\psi_j$ for any $1 \le x,y\le n$:
\begin{gather*}
\scalebox{0.999}{$\frac{1}{N_xN_y}\sum_{j=1}^{2n-1} \cos(\frac{mj\pi}{2n})\Big( \sin(\frac{jx\pi}{2n+1})\sin(\frac{(j+1)y\pi}{2n+1})+\sin(\frac{(j+1)x\pi}{2n+1})\sin(\frac{jy\pi}{2n+1})\Big ),$}
\end{gather*}
which can be observed to have the similar property of being zero if and only if $m+x+y \not \equiv 0 \pmod{2}$. As an example, for $m=9,x=14,n=52,y=49=n-3$, the above gives approximately $-0.0256625\neq 0$. This suggests that the conjecture \cite[(7.5)]{koo1994representations} does not provide the \textit{right} candidates if a strong SL-algebra is desired as there are low-high energy mix. These terms could make even the convergence of simple products such as the convergence of commutators to the commutators of scaling limit impossible. Of course, they should vanish at the scaling limit, along with all those terms with non zero coefficient and energy shift other than $m$ and a numerical simulation shows that happening but with a slower rate (as in \cite[table 19]{koo1994representations}). Hence, the rate of convergence \textit{could} also be another reason to consider the operators in  \hyperref[cnj5.5]{\textbf{Conjecture 5.5}} for obtaining the higher Virasoro modes.
\hfill \break

\textit{Interchiral observables $\overline{\mathbb{L}}_m$.}

\hfill \break
The summations considered above for the chiral CFTs need to be also analyzed for the full CFT. This time, an \textbf{interchiral} observable in the limit is obtained which could be thought of as:
$$\colon \partial_{\overline{z}}\overline{\psi}(\overline{z}) \psi(z) \colon+\colon \partial_z\psi(z) \overline{\psi}(\overline{z}) \colon.$$
Note that the stress energy tensor is 
$$\colon \partial_{\overline{z}}\overline{\psi}(\overline{z}) \overline{\psi}(\overline{z}) \colon+\colon \partial_z\psi(z) \psi(z) \colon.$$
Therefore the LM and RM part will mix due to the identities
\begin{gather*}
\scalebox{0.98}{$\sin\Big(\frac{2m(j+\frac{1}{2})\pi}{2n}\Big)\sin \Big(\frac{(2k-1)(j+1)\pi}{2n}\Big)+\sin\Big(\frac{2m(j-\frac{1}{2})\pi}{2n}\Big)\sin\Big(\frac{(2k-1)(j-1)\pi}{2n} \Big)=$} \\
\scalebox{0.98}{$\cos \Big(\frac{(2k-1-m)\pi}{2n}\Big)\cos \Big(\frac{(2k-1-2m)j\pi}{2n} \Big)-\cos \Big(\frac{(2k-1+m)\pi}{2n}\Big)\cos \Big(\frac{(2k-1+2m)j\pi}{2n} \Big),$}   
\end{gather*}
for the summation $-\sum_{j=1}^{2n} \sin(\frac{2m(j+\frac{1}{2})\pi}{2n})e_j$ in the $+1$ sector. The RM part ($\sin()$ eigenvectors) pairs up with the LM part ($\cos()$ eigenvectors).
 
As expected, one needs a similar identity for the LM part to RM part
\begin{gather*}
\scalebox{0.98}{$
\sin\Big(\frac{2m(j+\frac{1}{2})\pi}{2n}\Big)\cos \Big(\frac{(2k-1)(j+1)\pi}{2n}\Big)+\sin\Big(\frac{2m(j-\frac{1}{2})\pi}{2n}\Big)\cos\Big(\frac{(2k-1)(j-1)\pi}{2n} \Big)=
$} \\
\scalebox{0.98}{$\cos \Big(\frac{(2k-1+m)\pi}{2n}\Big)\sin \Big(\frac{(2k-1+2m)j\pi}{2n} \Big)-\cos \Big(\frac{(2k-1-m)\pi}{2n}\Big)\sin \Big(\frac{(2k-1-2m)j\pi}{2n} \Big).$}
\end{gather*}
Both identities above can be modified (replacing $2k-1$ with $2k$) to make them work for the $(-1)^F=-1$ sector.

There are similar equations for $-i\sum_{j=1}^{2n} \cos(\frac{2m(j+1)\pi}{2n})[e_j,e_{j+1}]$. Thus, as in the case of $\mathbb{L}_m+\mathbb{L}_{-m}$, we get an operator $\widetilde{\overline{\mathbb{L}}}_m$ converging to $\overline{\mathbb{L}}_m$ in the scaling limit which contains bilinear terms $\overline{\Psi}_{-k+m}\Psi_k$ instead of $\Psi_{-k+m}\Psi_k$. Finally, all results on the convergence (rate) for $\mathbb{L}_m$ applies to its interchiral \textit{counterpart} $\overline{\mathbb{L}}_m$.

The above illustrates why the similar conjecture \cite[(4.24) and (4.25)]{koo1994representations} likely involves a different diagonalization (as noted in \hyperref[1.6]{section 1.6}), in order to be true.
\end{document}